\documentclass[onecolumn, draftcls, 11pt]{IEEEtran}

\usepackage{todonotes}

\usepackage[utf8]{inputenc} 
\usepackage[T1]{fontenc}
\usepackage{url}
\usepackage{ifthen}
\usepackage{cite}
\usepackage[cmex10]{amsmath} 
\usepackage{graphicx} 
\usepackage{amssymb}
\usepackage{amsthm}
\usepackage{mathtools}
\usepackage{hyperref}
\usepackage{bbm}
\usepackage{tikz,lipsum,lmodern}
\usetikzlibrary{patterns, calc, intersections}
\usepackage[most]{tcolorbox}
\newtheorem{theorem}{Theorem}
\theoremstyle{definition}
\newtheorem{remark}{Remark}

\newtheorem{lemma}{Lemma}
\newtheorem{example}{Example}

\usepackage{algorithm}
\usepackage{algpseudocode}

\usepackage{tikz}
\usepackage{amsmath}
\usetikzlibrary{shapes, shapes.misc, arrows.meta, calc}

\newcommand{\by}{\mathbf{y}}

\newcommand{\MSE}{\mathsf{MSE}}
\newcommand{\PA}{\mathsf{PA}}

\usepackage{amsmath} 

\newcommand{\moh}[1]{{\color{blue}#1}}

\begin{document}
\title{Game of Coding for Vector-Valued Computations\thanks{
  The work of Mohammad Ali Maddah-Ali, Hanzaleh Akbari Nodehi, and Parsa Moradi has been partially supported by the National Science Foundation under Grant CCF-2348638. The work of Soheil Mohajer is  supported in part by AFOSR under the grant FA9550-23-1-0057. 
  }} 


\author{}


\author{%
  \IEEEauthorblockN{Hanzaleh Akbari Nodehi, Parsa Moradi, Soheil Mohajer, and Mohammad Ali Maddah-Ali}
  \IEEEauthorblockA{University of Minnesota, Twin Cities\\
                    Minneapolis, MN, USA\\
                    Email: \{akbar066, moradi, soheil, maddah\}@umn.edu}
}

\maketitle
\begin{abstract}
Traditional coding theory guarantees valid decoding only if a minority of symbols are adversarially manipulated. In contrast, the game of coding framework ensures reliable decoding, even in the presence of an adversarial majority.
 This formulation is motivated by emerging permissionless applications, particularly decentralized machine learning (DeML), where computation tasks are outsourced to external volunteer nodes that are predominantly rational and reward-seeking.


Prior investigations have analyzed the game of coding in the scalar setting. Since the results of most major computations in machine learning are vectors (e.g., computing the gradient of the loss for a machine learning model), we extend the framework in this paper to the general multi-dimensional Euclidean space. As a first, yet fundamental step, in this paper, we study a two-repetition code in which at least one node is controlled by a rational adversary, and we fully characterize the equilibrium and the optimal strategies of the players. Similar to the scalar case, this result serves as a cornerstone for addressing more general scenarios.
 

\end{abstract}

\section{Introduction}
Consider a system comprising a data collector (DC) and a set of $M$ external worker nodes. The DC outsources a  (perhaps approximate) computational task, such as the calculation of the gradient of a loss function for a machine learning model, to these workers, who return their results to the DC for aggregation. The network consists of two disjoint sets of workers: a set of honest nodes, denoted by $\mathcal{H}$, who faithfully adhere to the protocol, and a set of adversarial nodes, denoted by $\mathcal{T}$. We assume that these sets partition the network, such that $\mathcal{H} \cap \mathcal{T} = \emptyset$ and $\mathcal{H} \cup \mathcal{T} = \{1, \dots, M\}$. The DC is unaware of the membership of each set.

In the presence of adversarial nodes,  outsourced tasks are intentionally designed with (coded) redundancy, allowing the DC to recover an (approximate) result by aggregating all received outputs. Repetition is a specific type of redundancy, where the DC asks the computing nodes to evaluate the same function. In classical coding theory, this setting corresponds to basic repetition coding, where  successful decoding relies on explicit trust assumption that $|\mathcal{H}| \ge |\mathcal{T}| + 1$. For more advanced and efficient codes, where $K$ computing tasks can be performed in parallel, the requirement becomes stricter: for instance, a Reed--Solomon $(K,M)$ codes \cite{SudanBook} require $|\mathcal{H}| \ge |\mathcal{T}| + K$, while Lagrange coded computing~\cite{yu2019lagrange} for a polynomial function of degree~$d$ requires $|\mathcal{H}| > |\mathcal{T}| + (K-1)d$. Similar hard thresholds characterize recoverability in analog coding schemes \cite{jahani2018codedsketch,ZamirCoded, roth2020analog,BACC}.
In all these cases, a fundamental trust assumption is imposed: the honest workers must outnumber the adversaries by a certain margin. Consequently, if the majority of the network is adversarial, classical approaches fail to 
decode a correct result.

This limitation is particularly problematic in the emerging landscape of Web3 \cite{bitcoin2008bitcoin, buterin2013ethereum, ruoti2019sok}, specifically in decentralized machine learning (DeML). In DeML, training or inference is often coordinated by smart contracts (as the DC) on a blockchain to ensure transparency and accountability \cite{shafay2023blockchain, ding2022survey, kayikci2024blockchain, taherdoost2023blockchain, taherdoost2022blockchain, tian2022blockchain, salah2019blockchain}.  Given that blockchains cannot handle the heavy computational loads of modern AI, tasks must be outsourced to off-chain networks of volunteer workers \cite{zhao2021veriml}. However, the inherent permissionless nature of these systems allows unrestricted access to any contributor, making the conventional assumption of an honest majority difficult to guarantee. Thus, we cannot rely on classic coding theory to guarantee reliable decoding. 
 
On the other hand, these systems exhibit another important aspect. In decentralized (blockchain-based) environments, the behavior of worker nodes is primarily driven by economic incentives (e.g., cryptocurrency rewards). As a result, they act as rational agents optimizing their payoffs, rather than as purely malicious actors intent on system disruption. Recognizing that adversarial nodes behave as \textit{rational} players, rather than purely malicious ones, fundamentally changes the nature of the problem.

In this setting, the DC announces a reward policy: computations that satisfy specified acceptance criteria are rewarded, while others are rejected. These acceptance conditions are verifiable by the DC. However, a na\"ive criterion such as ``being correct'' is not feasible, without 
incurring the substantial overhead associated with cryptographic methods. On the other hand, the DC may require that any two reported results lie within a prescribed small distance of each other. If at least one of the nodes is honest, this can serve as a reasonable indication of (approximate) correctness.


In such scenarios, rational adversaries face a conflict of interest: they wish to maximize the error into the DC's final estimate of the result, but they also desire the financial reward, which is contingent on their results being accepted. This highlights a key characteristic of such systems, \emph{liveness}, defined as the probability that the DC's acceptance condition is satisfied and that it can (approximately) recover the desired computation.

Purely malicious actors do not care about liveness and may inject arbitrary errors into their inputs that violate the acceptance rule, as the DC's safeguard.
In contrast, rational players optimize their strategies by balancing their interest in the probability of acceptance (liveness) with the magnitude of the error they can successfully inject.
Conversely, the DC strategically optimizes its decision rule to increase the probability of acceptance while minimizing estimation error.
This interaction creates a game-theoretic scenario, formally introduced as the \textit{game of coding} framework in \cite{GoCJournal, GoDSybil, nodehi2025unknown, nodehi2026game}.

The game of coding framework offers a viable alternative to existing outsourcing solutions for DeML (see \cite{nodehi2026game} for a survey):
\begin{itemize}
    \item \textbf{Verifiable Computing:} This approach guarantees correctness of the results by requiring workers to generate cryptographic proof of correctness along with their results~\cite{thaler2022proofs,feng2021zen}. However, this method is often computationally prohibitive~\cite{liu2021zkcnn, xing2023zero, mohassel2017secureml, lee2024vcnn, weng2021mystique} and is restricted to exact computation~\cite{weng2021mystique, chen2022interactive, garg2022succinct, setty2012taking}, which conflicts with the approximate nature of AI.
    
    \item \textbf{Optimistic Verification:} This common approach assumes computations are correct by default and relies on a challenge-response mechanism to ensure correctness~\cite{bhat2023sakshi, conway2024opml}. In this model, the system assumes a result is correct unless a node acting as a challenger sends a fraud proof message to the blockchain claiming the computation is incorrect; the blockchain then initiates a judgment procedure to determine which party, either the worker who performed the computation or the challenger, is acting maliciously. The honest party is rewarded while the malicious one is punished. The primary failure of this method is that it suffers from delayed finality, because it requires a sufficiently large window of time to allow for the submission of a fraud proof message, and critically, this mechanism does not support approximate computing.
    
    \item \textbf{Classical Coded Computing:} This method utilizes algorithmic redundancy to manage latency and approximation~\cite{yu2017polynomial, jahani2018codedsketch, yu2019lagrange}. While effective, it lacks resilience against an adversarial majority, making it unsuitable for permissionless environments.
\end{itemize}

To overcome the limitations of the above approaches, the game of coding emerged as a powerful alternative. As established in \cite{GoCJournal}, this framework lies at the intersection of game theory and coding theory. The initial investigation in \cite{GoCJournal} laid the theoretical foundation by analyzing computation over scalar values. A key finding of \cite{GoCJournal} is that accurate estimation and reliable decodability are achievable even when the majority of the network is adversarial; a feat impossible under classical coding theory. Following this, subsequent research sought to capture critical practical considerations necessary for real-world deployment. Specifically, \cite{GoDSybil} addressed the threat of attackers masquerading as multiple workers to gain unfair influence, known as a Sybil attack; the work proved that the framework is inherently Sybil resistant, which means it maintains robustness even if an attacker creates numerous fake identities to manipulate the system. Furthermore, a bandit-based algorithm is proposed in~\cite{nodehi2025unknown} to handle scenarios where the DC does not know the adversary's strategy in advance; these are machine learning techniques that allow the system to learn the most effective reward policies over time by observing the adversary's behavior and adapting to it dynamically. A comprehensive summary of these motivations and comparisons is available in \cite{nodehi2026game}.


\subsection{Contributions of This Paper}
While all prior research on the game of coding was limited to scalar computations, in this paper we extend the framework to the general multi-dimensional Euclidean space. This extension is critical for practical applicability, since most real-world computations, such as gradient calculations in machine learning, involve vector-valued results rather than scalars. As a first, yet fundamental step, we study a two-repetition code in which one node is controlled by an adversary. Focusing on the two-node case is a deliberate choice; experience from the scalar setting has shown that this case represents the most fundamental and technically challenging part of the game of coding. By fully addressing the complexities of the two-node interaction, we establish a cornerstone that provides the necessary theoretical tools to solve more general and complex scenarios in higher dimensions.

In this work, we provide a rigorous problem formulation for the multi-dimensional setting; we formally define the class of utility functions that each player seeks to maximize and define the equilibrium of this game. In this strategic interaction, the DC first commits to a \emph{parametric} acceptance policy, comprising a specific decision rule governed by a tunable free parameter. For any given parameter setting, the adversary chooses a noise distribution that maximizes its own utility, balancing the trade-off between the probability of passing the acceptance policy and the magnitude of the injected error. The DC, anticipating this rational behavior and knowing the adversary's utility function, can effectively predict the adversary's optimal strategy, along with the resulting system state, for any choice of the parameter. By evaluating the expected outcome across the parameter space, the DC identifies and commits to the optimal parameter value that maximizes its own utility.

We assume very minimal and natural assumptions for these utility functions to ensure the framework captures a wide range of practical scenarios.
However, in this interaction, finding the equilibrium is directly related to the specific forms of these utility functions; it is a significant challenge to find the equilibrium if we stick to such minimal assumptions for the utilities. To resolve this issue, we define an intermediary optimization problem which is \emph{independent} of the specific utility functions of the players. Then, we prove that by having access to the result of that optimization problem, one can find the equilibrium of the game very readily using a two-dimensional searching procedure. This is a fundamentally important contribution, since it converts an optimization problem over infinitely-many dimension (the space of adversarial noise distributions and acceptance policies) to a problem with a two-dimensional feasible set.  
We present detailed numerical examples to clarify the theoretical findings and visualize the system dynamics. This work significantly extend the scope of the game of coding framework, capturing a critical aspect of real-world decentralized applications, where multi-dimensional data is the norm.

\subsection{Organization of The Paper}
The remainder of this paper is organized as follows. Section \ref{sec:Problem Formulation} formally introduces the problem formulation,  the utility functions for both the DC and the adversary, and the game-theoretic formulation of the problem. In Section~\ref{sec:Main Results}, we present the main theoretical findings of this work. The detailed mathematical proofs of the main theorems are provided in Section \ref{proof:theorem: equivalence_two_problem} and Section~\ref{proof:theorem:Main_Minimax_Result}. Section~\ref{sec:Illustrative_Examples} provides numerical examples across different cases to visualize the equilibrium and demonstrate the impact of different strategies. Finally, Section \ref{sec:conclusion} concludes the paper and discusses potential directions for future research. 

Several results in this paper are derived based on geometric structures of hypershperes and hyperspherical caps. We review some properties and characterize some required quantities on these $N$-dimensional objects in Appendices~\ref{sec:second_moment_n_ball},~\ref{app:moments_derivation},  and~\ref{sec:intersection_proof}. The proof of the main result is based on some  standalone lemmas, which are proved in Appendices~\ref{proof:lemma:prob_acceptance_general}, \ref{proof:lemma:mse_general}, \ref{proof:lemma:optimal_support}, and~\ref{proof:lemma:exact_acc_noise_existence}. Finally, Appendix~\ref{sec:case_N2} illustrates the main result of this work for the special case of $2$-dimensional vectors. 

\subsection{Notation}
 We denote random variables using uppercase letters and deterministic values (or realizations) using lowercase letters. Furthermore, we distinguish vectors from scalars by using boldface type for the former and standard type for the latter. For example, $\mathbf{X}$ represents a random vector, whereas $\mathbf{x}$ denotes a deterministic vector. Similarly, $X$ represents a scalar random variable, while $x$ denotes a deterministic scalar. Unless stated otherwise, all vectors are elements of the $N$-dimensional Euclidean space $\mathbb{R}^N$, and we denote the standard Euclidean ($\ell_2$) norm of a vector $\mathbf{x} = (x_1, \dots, x_N)$ by ${\|\mathbf{x}\|_2 = \sqrt{\sum_{i=1}^N x_i^2}}$. 

The symbol $\Gamma(\cdot)$ denotes the Euler Gamma function, which generalizes the factorial function to real and complex arguments. For any real number $x > 0$, it is defined by the integral
    \begin{align}\label{def:Gamma}
        \Gamma(x) = \int_{0}^{\infty} t^{x-1} e^{-t} \, dt.
    \end{align}
    If $n$ is non-negative integer, we know that $\Gamma(n+1) = n!$, and , $\Gamma(n + \frac{1}{2}) = (n - \frac{1}{2}) \cdot (n - \frac{3}{2}) \cdots \frac{1}{2} \cdot \sqrt{\pi}$.
 We define the $N$-dimensional closed ball of radius $r > 0$ centered at a point $\mathbf{c} \in \mathbb{R}^N$ as
\begin{equation} \label{eq:ball_general}
    \mathcal{B}_N( r, \mathbf{c}) \triangleq \left\{ \mathbf{x} \in \mathbb{R}^N : \|\mathbf{x} - \mathbf{c}\|_2 \leq r \right\}.
\end{equation}
For simplicity, when the center is at the origin (i.e., $\mathbf{c} = \mathbf{0}$), we denote the ball by $\mathcal{B}_N( r)$. The volume of an $N$-ball depends only on its radius and is independent of its center. We denote this volume by $V_N(r)$, which is given by
\begin{equation} \label{eq:Ball_Volume}
    V_N(r) =  C_N r^n = \frac{\pi^{N/2}}{\Gamma(\frac{N}{2} + 1)} r^N,
\end{equation}
where $C_N = \frac{\pi^{N/2}}{\Gamma(\frac{N}{2} + 1)}$. 
 Accordingly, we denote the uniform distribution over this ball by ${\mathbf{X} \sim \text{Unif}(\mathcal{B}_N(r))}$, which is the distribution characterized by the probability density function (PDF) ${f_{\mathbf{X}}(\mathbf{x}) = 1/V_N(r)}$ for $\mathbf{x} \in \mathcal{B}_N(r)$ and $0$ otherwise.
 
Let $\mathbb{R}^*$ denote the Euclidean space of arbitrary dimension. For any set $\mathcal{S} \subseteq \mathbb{R}^*$ and an arbitrary function $f: \mathcal{S} \to \mathbb{R}$, the notation $\underset{x \in \mathcal{S}}{\arg\max} ~f (x)$ represents the set comprising all elements $x$ in $\mathcal{S}$ that maximize $f (x)$. Similarly we define $\underset{x \in \mathcal{S}}{\arg\min} ~f (x)$. For $a, b \in \mathbb{R}$, the notation $[a, b]$ represents the closed interval $\{x \in \mathbb{R} : a \le x \le b\}$.

\section{Problem Formulation}\label{sec:Problem Formulation}

In this section, we establish the formal mathematical framework for the $N$-Dimensional game of coding. We consider a setting comprised of a data collector (DC) and a set of $K=2$ computational nodes\footnote{While a two-node system may appear structurally simple, it represents the fundamental unit of our strategic interaction; even in this minimal setting, the game-theoretic dynamics exhibit significant technical complexity and provide the necessary intuition for larger networks.}, denoted by $\mathcal{K} \triangleq \{1, 2\}$, operating in an $N$-dimensional Euclidean space $\mathbb{R}^N$. The system architecture is illustrated in Figure \ref{fig:Two_node_model}. Let $\mathbf{U} \in \mathbb{R}^N$ be a random vector representing the ground truth, which is characterized by a probability density function $f_{\mathbf{U}}(\mathbf{u})$. The ultimate goal of the DC is to compute/estimate $\mathbf{U}$, which can be found from the data available to the computing nodes. However, the DC does not have direct access to the realization of $\mathbf{U}$ and must instead rely on the reports provided by the nodes to estimate its value.

\begin{figure}[htbp]
    \centering
    \resizebox{0.7\columnwidth}{!}{%
        \begin{tikzpicture}[>=Stealth, thick]

            \tikzset{
                block/.style={draw, rectangle, minimum height=10mm, minimum width=22mm, align=center, fill=white},
                diamondblock/.style={draw, diamond, aspect=1.6, minimum width=35mm, minimum height=20mm, align=center, fill=white, inner sep=0pt},
                rejectblock/.style={draw, rounded rectangle, minimum height=10mm, minimum width=22mm, align=center, fill=white}
            }

            \def\boxL{0}       
            \def\boxR{9.3}     
            \def\boxT{2.2}     
            \def\boxB{-2.2}    
            
            \def\nodeX{-3.8}   
            \def\nodeY{1.4}    
            
            \def\diaX{3.0}     
            
            \def\estX{7.5}     
            \def\rejY{-3}    

            \draw[thick, fill=gray!5] (\boxL, \boxT) rectangle (\boxR, \boxB);
            \node[anchor=south east, font=\large] at (\boxR-0.1, \boxB+0.2) {Data Collector};

            \node[block] (node1) at (\nodeX, \nodeY) {Node 1};
            \node[block] (node2) at (\nodeX, -\nodeY) {Node 2};
            
            \node[diamondblock] (decision) at (\diaX, 0) {${||\mathbf{Y}_1 - \mathbf{Y}_2||}_2 \leq \eta\Delta$};
            
            \node[block] (estimation) at (\estX, 0) {Estimation};
            
            \node[rejectblock] (reject) at (\diaX, \rejY) {Reject};
            
            \node[font=\Large] (output) at (\boxR + 1.2, 0) {$\hat{\mathbf{U}}$};

            
            \draw[->] (node1.east) -- node[above] {$\mathbf{Y}_1 = \mathbf{U} + \mathbf{N}_h$} (\boxL, \nodeY);
            \draw[->] (\boxL, \nodeY) -- (decision.155);
            
            \draw[->] (node2.east) -- node[above] {$\mathbf{Y}_2 = \mathbf{U} + \mathbf{N}_a$} (\boxL, -\nodeY);
            \draw[->] (\boxL, -\nodeY) -- (decision.205);
            
            \draw[->] (decision.east) -- node[above] {Yes} (estimation.west);
            \draw[->] (decision.south) -- node[right] {No} (reject.north);
            
            \draw[->] (estimation.east) -- (output.west);

        \end{tikzpicture}%
    }
  \caption{System model for the $N$-Dimensional game of coding. The network consists of one honest node and one adversarial node. Each node reports a noisy version of the ground truth $\mathbf{U}$ to the DC. For the honest node, the noise $\mathbf{N}_h$ is uniformly distributed within $\mathcal{B}_N( \Delta)$, while for the adversarial node, the noise $\mathbf{N}_a$ follows an arbitrary distribution $g(\cdot)$ chosen by the adversary. Upon receiving the data, the DC decides whether to accept or reject the inputs based on a consistency threshold $\eta$. If accepted, the DC outputs an estimate of $\mathbf{U}$. In this game, the DC acts as the leader choosing $\eta$, and the adversary acts as the follower choosing $g(\cdot)$.}
    \label{fig:Two_node_model}
\end{figure}
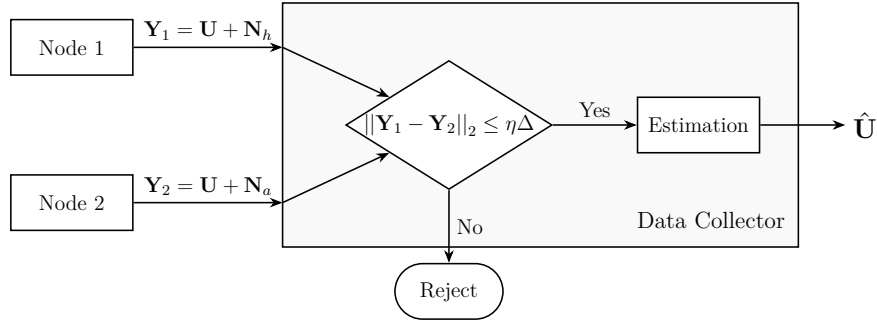


The set of nodes is partitioned into two disjoint singleton sets: an honest node $\mathcal{H}=\{h\}$ and an adversarial node $\mathcal{T}=\{a\}$. Thus, $\mathcal{K} = \{h, a\}$. The identity of the adversary is unknown to the DC, and we assume the adversary is selected uniformly at random from $\mathcal{K}$. Each node $k \in \mathcal{K}$ transmits a report $\mathbf{Y}_k \in \mathbb{R}^N$ to the DC. The honest node reports a noisy version of the ground truth, denoted by $\mathbf{Y}_h$, where
\begin{align}\label{eq:honest_model}
    \mathbf{Y}_h = \mathbf{U} + \mathbf{N}_h,
\end{align}
and $\mathbf{N}_h \in \mathbb{R}^N$. This noise represents inherent noise of approximate computing, measurement error, quantization and compression, or oracle inaccuracy. The noise is uniformly distributed within an $N$-dimensional ball of radius $\Delta$, denoted as $\mathcal{B}_N( \Delta)$, where $\Delta > 0$. Specifically, we have
\begin{align}
    \mathbf{N}_h \sim \text{Unif} \left( \mathcal{B}_N( \Delta) \right).
\end{align}
The parameter $\Delta$ is assumed to be universally known at all parties.
 This distribution implies that the honest node provides an unbiased approximation within a strictly defined accuracy radius.

Conversely, the adversarial node possesses knowledge of the exact realization of $\mathbf{U}$ and generates a report denoted by $\mathbf{Y}_a$, where 
\begin{align}\label{eq:adv_model}
    \mathbf{Y}_a = \mathbf{U} + \mathbf{N}_a.
\end{align}
The adversarial noise $\mathbf{N}_a$ is drawn from an arbitrary PDF $g(\cdot)$ chosen by the adversary, which is kept private from the DC. 
We assume that both noise components $\mathbf{N}_h$ and $\mathbf{N}_a$ are independent of the ground truth $\mathbf{U}$ and are also independent of each other.

The DC collects the reports into a tuple $\underline{\mathbf{Y}} \triangleq (\mathbf{Y}_1, \mathbf{Y}_2)$ and processes them in two stages: \textbf{Acceptance} and \textbf{Estimation}.

\begin{enumerate}
  \item \textbf{Acceptance via Consistency Check:} The DC accepts the computation if and only if the Euclidean distance between the reports does not exceed a threshold scaled by the honest noise bound $\Delta$. More precisely, the acceptance event, denoted by $\mathcal{A}_\eta$, occurs if
\begin{align}
    \mathcal{A}_\eta: \quad \left\| \mathbf{Y}_1 - \mathbf{Y}_2 \right\|_2 \leq \eta \Delta,
\end{align}
where $\eta$ is a scalar parameter controlling the strictness of the check. The probability of acceptance (PA) is defined as
\begin{align}\label{eq:PA_def}
    \mathsf{PA}(g(\cdot), \eta) \triangleq \Pr(\mathcal{A}_\eta) = \Pr \left( \left\| \mathbf{Y}_1 - \mathbf{Y}_2 \right\|_2 \leq \eta \Delta \right),
\end{align}
where the probability is evaluated over the randomness of $\mathbf{U}$, $\mathbf{N}_a$ and $\mathbf{N}_h$.
    \item \textbf{Estimation:} When the reported vectors are accepted, the DC estimates the ground truth using the average of the two reported vectors. More precisely, we have
\begin{align}
    \hat{\mathbf{U}} = \frac{\mathbf{Y}_1 + \mathbf{Y}_2}{2}.
\end{align}
The performance of this estimator is measured by the mean squared error (MSE), as
\begin{align}\label{eq:MSE_def}
    \mathsf{MSE}(g(\cdot), \eta) \triangleq \mathbb{E} \left[ \left\|\mathbf{U} - \frac{\mathbf{Y}_1 + \mathbf{Y}_2}{2}\right\|_2^2 \mathrel{\Big|} \mathcal{A}_\eta \right].
\end{align}
\end{enumerate}

The threshold parameter $\eta$ governs the fundamental compromise between the system's liveness, the probability to accept the computation and produce an output, and the accuracy of the final estimate. If $\eta$ is set to a very large value, the system achieves near-perfect liveness, but this allows the adversary to introduce unbounded error into the estimate of $\mathbf{U}$. On the other hand, setting a strict and small threshold for $\eta$  limits the error magnitude. However, this strictness makes the system vulnerable to denial-of-service (DoS) attacks. A rational adversary could intentionally provide data that slightly violates the threshold, causing the DC to reject the inputs and preventing the system from producing any estimate.

Furthermore, the choice of $\eta$ directly influences the adversary's behavior. In many decentralized applications, such as oracle networks and decentralized machine learning (DeML) \cite{eskandari2021sok, breidenbach2021chainlink, benligiray2020decentralized}, the adversary only receives rewards when their input is accepted. If the system rejects the data, the adversary gains no rewards and exerts no influence on the outcome. This structure creates a partial alignment of interests: to maximize the error, the adversary should choose a large noise; however, the adversary must first ensure that the system remains functional and its reported vector is accepted. Consequently, the adversary is incentivized to keep their induced noise within a range that satisfies the acceptance criteria, rather than simply forcing the system to shut down.

To rigorously capture this mechanism, we model the interaction between the DC and the adversary as a \textbf{Stackelberg game} \cite{von2010market}. In game theory, a Stackelberg model describes a sequential hierarchy where a \textbf{leader} commits to a strategy first, and a \textbf{follower} moves only after observing the leader's action. This stands in contrast to a standard Nash equilibrium in simultaneous games, where players act at the same time without observing the opponent's choice.

In our context, the DC acts as the \textbf{leader}. This role is mandated by the practical implementation of the system: the DC typically operates as a smart contract. Due to the inherent transparency of blockchain technology, the DC's acceptance policy, specifically the threshold parameter $\eta$, is a public code. The adversary, acting as the \textbf{follower}, can inspect the smart contract to see the exact value of $\eta$ before generating any data. Because the adversary chooses their strategy with full knowledge of the DC's commitment, the interaction is inherently sequential rather than simultaneous.

To formalize the game, we define the admissible action sets for both players. To choose the action set for the DC, we note that even in the hypothetical and optimistic case where both nodes are honest, the inherent approximate nature of the computation implies that each report $\mathbf{Y}_i$ may deviate from the ground truth by up to $\Delta$; consequently, the distance $\|\mathbf{Y}_1 - \mathbf{Y}_2\|_2$ can be as large as $2\Delta$. To ensure that the DC does not reject these honest reports, the threshold parameter $\eta$ must be at least 2.\footnote{While exploring $\eta < 2$ could offer an interesting trade-off between the risk of rejecting honest nodes and the potential for tighter error control, such an extension does not fundamentally alter the core analysis of this paper and can be viewed as a complementary direction for future research.} Thus, the DC's action set is defined as
\begin{align}
    \Lambda_{\text{DC}} \triangleq [2, \infty).
\end{align}
The adversary, in turn, selects a noise distribution from the action set $\Lambda_{\text{AD}}$, which consists of all valid probability density functions over the noise space $\mathbb{R}^N$, More precisely, we have
\begin{align}
    \Lambda_{\text{AD}} \triangleq \left\{ g: \mathbb{R}^N \to \mathbb{R}_{\ge 0} \middle| \int_{\mathbb{R}^N} g(\mathbf{x}) d\mathbf{x} = 1 \right\}.
\end{align}

The players aim to maximize their respective utility functions. These objectives are captured by the following utility functions
\begin{align}
    \mathsf{U}_{\text{DC}}(g(\cdot), \eta) &\triangleq Q_{\text{DC}}\left( \mathsf{MSE}(g(\cdot), \eta), \mathsf{PA}(g(\cdot), \eta) \right), \label{defeq:U_DC}\\
    \mathsf{U}_{\text{AD}}(g(\cdot), \eta) &\triangleq Q_{\text{AD}}\left( \mathsf{MSE}(g(\cdot), \eta), \mathsf{PA}(g(\cdot), \eta) \right), \label{defeq:U_AD}
\end{align}
where $Q_{\text{DC}}$ is monotonically non-increasing in MSE and non-decreasing in PA, while $Q_{\text{AD}}$ is strictly\footnote{To determine the game equilibrium, we utilize an intermediate optimization problem defined in \eqref{c_small_definition}, which is independent of $\mathsf{U}_{\text{DC}}$ and $\mathsf{U}_{\text{AD}}$. Theorem \ref{theorem: equivalence_two_problem} establishes that by solving \eqref{c_small_definition}, we can determine the optimal strategies for both players, specifically, the noise distribution for the adversary and the acceptance parameter for the DC. The strict monotonicity of $Q_{\text{AD}}$ is a necessary condition for the validity of this theorem (see Section \ref{proof:theorem: equivalence_two_problem} for details). Intuitively, this condition ensures that any adversarial best response must maximize the induced error for a given probability of acceptance, as the adversary would otherwise have a strict incentive to further increase the system error. In contrast, for the DC, we rely only on the natural assumption of non-decreasing monotonicity to encompass the broadest range of practical scenarios.} increasing in both arguments.  We assume that functions $Q_{\text{DC}}$ and $Q_{\text{AD}}$ are publicly known\footnote{even though we assume the utility functions are universally known, it turns out the adversary's strategy is independent of the DC’s. However, it is crucial for our solution that DC should know $Q_{\text{AD}}$. } by all parties.

The game is resolved via backward induction.  First, for any fixed threshold $\eta \in \Lambda_{\text{DC}}$ committed to by the leader, the follower identifies the set of optimal strategies to maximize its own utility function; this strategic response is captured by the adversary's best response set, which we define as
\begin{align}\label{eq:best_set_adv_definition}
    \mathcal{B}^{\eta}_{\text{AD}} \triangleq \underset{g(\cdot) \in \Lambda_{\text{AD}}}{\arg\max} ~ \mathsf{U}_{\text{AD}}(g(\cdot), \eta).
\end{align}

It is crucial to observe that the adversary is indifferent among all strategies within $\mathcal{B}^{\eta}_{\text{AD}}$, as they all yield the same maximal utility. However, these strategies may produce different utilities for the DC. To ensure a robust security guarantee, we adopt a conservative worst-case formulation. We assume that, among the adversary's optimal strategies, the specific $g(\cdot)$ chosen is the one most detrimental to the DC. We therefore define the set of worst-case adversarial responses as
\begin{align}\label{wrost_case_best_response}
    \bar{\mathcal{B}}^{\eta}_{\text{AD}} \triangleq \underset{g(\cdot) \in \mathcal{B}^{\eta}_{\text{AD}}}{\arg\min} ~ \mathsf{U}_{\text{DC}}(g(\cdot), \eta).
\end{align}
Note that the DC can also solve the optimization problem in~\eqref{wrost_case_best_response}, and hence, it knows that for every acceptance parameter $\eta$, what noise density function $g(\cdot)$ will be chosen by the adversary. Finally, the DC acts as the leader by selecting the optimal threshold $\eta^*$ that maximizes its utility under this worst-case noise, introduced by the adversary. More precisely, for any $\eta$, let $g^*_{\eta}(\cdot)$ be an arbitrary noise distribution in $\bar{\mathcal{B}}^{\eta}_{\text{AD}}$. Since every noise in $\bar{\mathcal{B}}^{\eta}_{\text{AD}}$ provides the same utility for the DC, we have
\begin{align}\label{eq:DC_optim}
    \eta^* = \underset{\eta \in \Lambda_{\text{DC}}}{\arg\max} ~  \mathsf{U}_{\text{DC}}(g^*_{\eta}(\cdot), \eta).
\end{align}

The Stackelberg equilibrium is therefore characterized by the pair $(\eta^*, g^*_{\eta^*}(\cdot))$, where $g^*_{\eta^*}(\cdot)$ is any noise in the set $\bar{\mathcal{B}}^{\eta^*}_{\text{AD}}$. The corresponding MSE, probability of acceptance, and utility values for this equilibrium are denoted by $\mathsf{MSE}^* = \mathsf{MSE}(g^*_{\eta^*}(\cdot), \eta^*)$, $\mathsf{PA}^* = \mathsf{PA}(g^*_{\eta^*}(\cdot), \eta^*)$, $\mathsf{U}^*_{\text{DC}} = \mathsf{U}_{\text{DC}}(g^*_{\eta^*}(\cdot), \eta^*)$, and $\mathsf{U}^*_{\text{AD}} = \mathsf{U}_{\text{AD}}(g^*_{\eta^*}(\cdot), \eta^*)$, respectively.

\section{Main Results}\label{sec:Main Results}

Based on \eqref{eq:best_set_adv_definition},  \eqref{wrost_case_best_response}, and \eqref{eq:DC_optim}, The DC's optimal threshold $\eta^*$ is determined by solving the following optimization problem
\begin{align} \label{eq:etastar}
    \eta^* = \underset{\eta \in \Lambda_{\mathsf{DC}}}{\arg\max} ~ \underset{g(\cdot) \in \mathcal{B}^{\eta}_{\text{AD}}}{\min} ~ Q_{\mathsf{DC}} \left( \mathsf{MSE}\left(g(\cdot), \eta \right), \mathsf{PA} \left( g(\cdot), \eta \right)\right).
\end{align}
The optimization problem in~\eqref{eq:etastar} is difficult to solve directly due to two fundamental challenges.

\begin{enumerate}
    \item \textbf{Utility Function Dependency and Minimal Assumptions:} We aim to solve the game for a broad class of utility functions. The optimization problem in \eqref{eq:etastar} is highly dependent on the specific forms of $Q_{\mathsf{DC}}$ and $Q_{\mathsf{AD}}$, defined in \eqref{defeq:U_DC} and \eqref{defeq:U_AD}, respectively. However, we make no restrictive mathematical assumptions, such as convexity or concavity, on these utility functions. Our only requirement is that they satisfy the intuitive, common sense monotonicity properties defined earlier (e.g., the adversary always prefers higher error). This strong dependency combined with our minimal assumptions precludes the use of standard  optimization tools that rely on specific functional forms.

     \item \textbf{Infinite Dimensional Strategy Space in Multiple Dimensions:} The adversary's optimization domain is vast. The inner minimization in \eqref{eq:etastar} requires searching over $\Lambda_{\mathsf{AD}}$, which contains \textit{every} possible probability density function on $\mathbb{R}^N$. Unlike previous papers of game of coding \cite{GoCJournal, GoDSybil, nodehi2025unknown, nodehi2026game} where the noise is scalar, here the noise is $N$ dimensional. This spatial multidimensionality fundamentally changes the problem since the adversary is free to shape the noise distribution arbitrarily across multiple dimensions, rather than shifting probability mass along a single line. 
\end{enumerate}

To circumvent the first challenge, we define an intermediate optimization problem that is \textbf{independent of the utility functions} $Q_{\mathsf{DC}}$ and $Q_{\mathsf{AD}}$. Consider a scenario where the adversary is constrained to maintain a specific level of system liveness. That is, for a fixed threshold $\eta$ and a minimum target acceptance probability $\alpha \in (0, 1]$, we determine the maximum MSE the adversary can strictly enforce. This defines the system's characteristic function, denoted by $c_{\eta}(\alpha)$. 
More precisely, for a fixed threshold $\eta \in \Lambda_{\text{DC}}$ and a given probability of acceptance $\alpha \in (0, 1]$, we define the intermediary optimization problem  as
\begin{align}\label{c_small_definition}
    c_{\eta}(\alpha) \triangleq \max_{g(\cdot) \in \Lambda_{\text{AD}}} \quad & \MSE(g(\cdot), \eta) \nonumber \\
    \text{subject to} \quad & \PA(g(\cdot), \eta) \geq \alpha.
\end{align}

\begin{remark}\label{remark_of_C_eta}
    We note that the optimization problem in \eqref{c_small_definition}, does not depend on the specific utility functions of the adversary (AD) or the data collector (DC). The function $c_\eta(\alpha)$ can be universally evaluated for every $\eta$ and $\alpha$. This decouples problem of optimum strategy of noise injection at the adversary from other party's utility function.  
\end{remark}

Intuitively, $c_{\eta}(\alpha)$ traces the Pareto frontier of the attack surface, representing the maximum damage (error) the adversary can inflict for any required probability of acceptance. We first note that $c_{\eta}(\alpha)$ is a non-increasing function of $\alpha$. This follows from the fact that if a noise distribution $g(\cdot)$ satisfies $\PA(g(\cdot), \eta) \geq \alpha_1$, it necessarily satisfies $\PA(g(\cdot), \eta) \geq \alpha_2$ for any $\alpha_2 < \alpha_1$; consequently, the optimization domain in \eqref{c_small_definition} for $\alpha_2$ is a superset of that for $\alpha_1$, implying $c_{\eta}(\alpha_2) \geq c_{\eta}(\alpha_1)$. 

As illustrated in Figure \ref{fig:pareto_frontier}, the $c_{\eta}(\alpha)$  curve demarcates the feasible region of attacks. Point A (in red) represents an inefficient strategy for a rational adversary; suppose that for a committed $\eta$, an adversarial noise $g_A(\cdot)$ achieves the outcome at A. By replacing it with the noise $g_B(\cdot)$ corresponding to point B (in black), the adversary maintains the same probability of acceptance while inducing a strictly higher MSE. Since the adversary's utility $\mathsf{U}_{\text{AD}}$, defined in \eqref{defeq:U_AD}, is strictly increasing with respect to the induced error, a rational follower will always prefer point B over point A. Conversely, point C (in gray) in Figure \ref{fig:pareto_frontier} represents an outcome that is strictly unattainable. By the definition of $c_{\eta}(\alpha)$ in \eqref{c_small_definition}, there exists no feasible noise distribution $g(\cdot) \in \Lambda_{\text{AD}}$ capable of inducing the level of MSE shown at C without violating the corresponding probability of acceptance constraint. Thus, a rational adversary will always restrict its strategy set to the frontier defined by $c_{\eta}(\alpha)$.

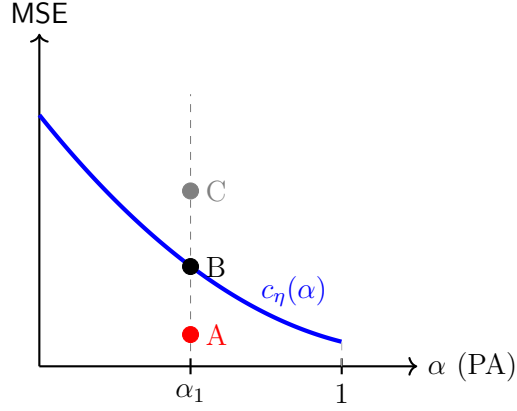
\begin{figure}[ht]
    \centering
    \begin{tikzpicture}[scale=2]
        \draw[->, thick] (0,0) -- (2.5,0) node[right] {$\alpha$ (PA)};
        \draw[->, thick] (0,0) -- (0,2.2) node[above] {$\MSE$};
        
        \draw[thick] (1, 1.5pt) -- (1, -1.5pt) node[below] {$\alpha_1$};
        \draw[thick] (2, 1.5pt) -- (2, -1.5pt) node[below] {$1$};
        
        \draw[ultra thick, blue, samples=100, domain=0:2] plot (\x, {0.25*(\x-2.5)*(\x-2.5) + 0.1});
        \node[blue, right] at (1.4, 0.5) {$c_{\eta}(\alpha)$};
        
        \draw[dashed, gray!60] (2, 0) -- (2, 0.16); 
        
        \def\px{1.0}
        \def\py{0.66} 
        \draw[dashed, gray] (\px, 0.1) -- (\px, 1.8);
        
        \filldraw[gray] (\px, \py + 0.5) circle (1.5pt) node[right, xshift=2pt] {C};
        
        \filldraw[black] (\px, \py) circle (1.5pt) node[right, xshift=2pt] {B};
        
        \filldraw[red] (\px, \py - 0.45) circle (1.5pt) node[right, xshift=2pt] {A};
        
    \end{tikzpicture}
    \caption{The Pareto frontier $c_{\eta}(\alpha)$ and adversarial rationality.  Point A (red) is inefficient compared to point B (black), while Point C (gray) lies in the unattainable region beyond the maximum possible error for $\alpha_1$.}
    \label{fig:pareto_frontier}
\end{figure}
Perhaps surprisingly, it can be shown that  characterizing \eqref{c_small_definition} is sufficient to resolve the entire game. More precisely, by leveraging $c_{\eta}(\alpha)$, we can collapse the complex, infinite-dimensional search over probability distributions in \eqref{eq:etastar}, into a tractable, finite-dimensional scalar optimization. This reduction is formalized in Algorithm \ref{Alg:finding_eta}, which takes the utility functions and the derived curve $c_{\eta}(\cdot)$ as inputs to efficiently compute the optimal strategy $\hat{\eta}$. The following theorem establishes that this scalar reduction is exact and that the output of Algorithm \ref{Alg:finding_eta} corresponds precisely to the Stackelberg equilibrium of the original game.

\begin{theorem}\label{theorem: equivalence_two_problem}
    The optimal threshold for the DC in the Stackelberg game formulated in \eqref{eq:etastar} is given by the output of Algorithm \ref{Alg:finding_eta}, denoted as $\hat{\eta}$. That is, $\eta^* = \hat{\eta}$.
\end{theorem}

\begin{algorithm}[t]
\caption{Determination of the Optimal Threshold $\eta^*$}
\label{Alg:finding_eta}
\begin{algorithmic}[1]
\State \textbf{Input:} Utility functions $\mathsf{U}_{\text{AD}}(\cdot, \cdot)$, $\mathsf{U}_{\text{DC}}(\cdot, \cdot)$, and the function $\{c_{\eta}(\cdot):\eta \in \Lambda_{\text{DC}}\}$.
\State \textbf{Output:} Optimal threshold $\hat{\eta}$.

\vspace{0.5em}

\State \textbf{Step 1: Follower's Rational Response}
\State For any fixed $\eta$, determine the set of optimal acceptance probabilities $\mathcal{L}_{\eta}$ that maximize the adversary's utility along the curve $c_{\eta}(\alpha)$:
\begin{align}\label{eq:alg_step1}
    \mathcal{L}_{\eta} = \arg\max_{0 < \alpha \leq 1 } \mathsf{U}_{\text{AD}}\left(c_{\eta}(\alpha), \alpha\right).
\end{align}

\State \textbf{Step 2: Leader's Strategic Choice}
\State The DC identifies $\hat{\eta}$ by maximizing its utility, accounting for the adversary's best response:
\begin{align}\label{eq:alg_step2}
    \hat{\eta} = \arg\max_{\eta \in \Lambda_{\text{DC}}} \left( \min_{\alpha \in \mathcal{L}_{\eta}} \mathsf{U}_{\text{DC}} \left(c_{\eta}(\alpha), \alpha\right) \right).
\end{align}
\end{algorithmic}
\end{algorithm}

The comprehensive proof of Theorem \ref{theorem: equivalence_two_problem} is provided in Section \ref{proof:theorem: equivalence_two_problem}; however, we outline the underlying intuition here. When the DC commits to a threshold $\eta$, a rational adversary responds by selecting a noise distribution that maximizes their utility, resulting in a $(\PA, \MSE)$ pair.  As discussed earlier and illustrated in Figure \ref{fig:pareto_frontier}, the curve $c_{\eta}(\alpha)$ serves as the boundary of the feasible attack space. Any point situated below this frontier, such as point A (red), is suboptimal for the adversary, as they could achieve a higher utility at point B (black) for the same acceptance probability. Conversely, points above the curve, such as point C (gray), are strictly unattainable. Consequently, for a fixed $\eta$, the adversary's optimal choice must lie on the frontier $c_{\eta}(\cdot)$, which allows the DC to characterize the adversary's behavior through the simplified optimization in \eqref{eq:alg_step1}. By anticipating this rational response, the DC can then optimize its own utility across all possible values of $\eta$ by solving \eqref{eq:alg_step2}, ensuring the equilibrium strategy is captured.

Theorem \ref{theorem: equivalence_two_problem} establishes that the original game is entirely determined by the characteristic function $c_{\eta}(\cdot)$. Consequently, finding the optimal strategy reduces to deriving the explicit form of this curve. The following theorem provides the exact analytical characterization of $c_{\eta}(\cdot)$ for any system dimension.

\begin{theorem}\label{theorem:Main_Minimax_Result}
    For any dimension $N \geq 1$, decoding threshold $\eta \in \Lambda_{\mathsf{DC}}$, and $\alpha \in (0, 1]$, we have
    \begin{align}\label{eq:c_eta_result}
       c_{\eta}(\alpha) = \frac{\tilde{\Psi}_N^*(\alpha)}{4 \alpha},
    \end{align}
    where $\tilde{\Psi}_N^*(q)$ denotes the upper concave envelope of the function $\tilde{\Psi}_N(q)$ over the domain $q \in [0, 1]$. The function $\tilde{\Psi}_N(q)$ is defined as
    \begin{align}\label{eq:Psi_tilde_def}
        \tilde{\Psi}_N(q) \triangleq \Psi_N\left( \Phi_N^{-1}(q) \right),
    \end{align}
    where $\Phi_N^{-1}(q)$ is the inverse of the function  
    \begin{align}
    \Phi_N(z) &= \frac{V_{\text{lens}}(\Delta, \eta\Delta, z)}{V_N(\Delta)}, \label{eq:Phi_n_def}
    \end{align}
    with
    \begin{align}
        V_{\text{lens}}(\Delta, \eta\Delta, z) &= K_N(\Delta, u_c(z)) + K_N(\eta\Delta, z - u_c(z)),  \label{eq:V_lens_def} \\
        K_N(r, c) &= \frac{\pi^{(N-1)/2} r^N}{\Gamma(\frac{N+1}{2})} \int_{c/r}^{1} (1 - t^2)^{\frac{N-1}{2}} \, dt, \label{eq:K_cap_def} \\
        u_c(z) &= \frac{z^2 + \Delta^2(1 - \eta^2)}{2z},\label{eq:uc_def}
    \end{align}
    for 
    $ z \in [(\eta-1)\Delta, (\eta+1)\Delta]$.
    Moreover, for the same range of $z$, we have
    \begin{align}\label{eq:Psi_n_def}
        \Psi_N(z) = \frac{1}{V_N(\Delta)} \Bigg( &\left[ J_N(\Delta, u_c(z)) + z^2 V_1 \right] \nonumber \\
        &+ \left[ J_N(\eta\Delta, z-u_c(z)) + 4z^2 V_2 - 2z Q_N(\eta\Delta, z-u_c(z)) \right] \Bigg),
    \end{align}
    where 
    \begin{align}
       V_1 &= K_N(\Delta, u_c(z)), \label{def:v1} \\
       V_2 &= K_N(\eta\Delta, z-u_c(z)), \label{def:v2}\\
       Q_N(r, d) &= \frac{r^2 - d^2}{N+1} V_{N-1}(\sqrt{r^2 - d^2}),\label{eq:Q_n_def} \\
       J_N(r, d) &= \frac{N r^2}{N+2} K_N(r, d) + \frac{2 d}{N+2} Q_N(r, d)\label{eq:J_n_def}.
    \end{align}
    Finally, the function $\Gamma(.)$ is  defined in \eqref{def:Gamma},  $V_N(.)$ is defined in \eqref{eq:Ball_Volume}.
\end{theorem}

The detailed proof of Theorem \ref{theorem:Main_Minimax_Result} is provided in Section \ref{proof:theorem:Main_Minimax_Result}; however, the following intuitive interpretation of the theorem would be helpful to better understand the proof. As highlighted earlier, one of the primary challenges in characterizing $c_{\eta}(\cdot)$ via the optimization problem in \eqref{c_small_definition} is that the adversarial noise is an $N$ dimensional vector. To make this variational problem mathematically tractable, we try to reduce it to a one dimensional scalar optimization problem. This reduction is achievable by exploiting the inherent geometry of the interaction. 

Considering a scalar random variable $Z \sim f_Z(z)$ representing the magnitude of the adversarial noise, we can use the law of total probability to express $\mathsf{PA}(g(\cdot), \eta)$ as
\begin{align}
    \mathsf{PA}(g(\cdot), \eta) = \Pr(\mathcal{A}_\eta) = \int_{0}^{\infty} \Pr(\mathcal{A}_\eta \mid Z=z) f_Z(z) \, dz.
\end{align}
It can be analytically shown that, due to the spherical symmetry of the honest node noise and the specific structure of the acceptance policy, the conditional probability $\Pr(\mathcal{A}_\eta \mid Z=z)$ depends solely on the scalar magnitude $z$. Similarly, we can apply the exact same structural decomposition to $\mathsf{MSE}(g(\cdot), \eta)$ (see \eqref{eq:proof_integral_sub}) and prove that the conditional estimation error is also completely determined by $z$, a property that fundamentally relies on the symmetric characteristics of the honest node noise, the acceptance policy, and the estimation rule. More precisely, we establish in Lemmas \ref{lemma:prob_acceptance_general} and \ref{lemma:mse_general} that for both the probability of acceptance and the estimation error, we have 
\begin{align}
    \mathsf{PA}(g(\cdot), \eta) &= \int_{0}^{\infty} \Phi_N(z) f_Z(z) \, dz, \\
    \mathsf{MSE}(g(\cdot), \eta) &= \frac{1}{4 \mathsf{PA}(g(\cdot), \eta)} \int_{0}^{\infty} \Psi_N(z) f_Z(z) \, dz,
\end{align}
where $\Phi_N(z)$ and $\Psi_N(z)$ are the geometric kernels defined in \eqref{eq:Phi_n_def} and \eqref{eq:Psi_n_def}. Characterizing these kernels, which entails finding $\Phi_N(z)$ and $\Psi_N(z)$ for each value of $z$, is completely independent of the adversarial noise. Instead, this characterization reduces to a separate geometric problem, which is solved in detail within the proofs of Lemmas \ref{lemma:prob_acceptance_general} and \ref{lemma:mse_general} in Appendices \ref{proof:lemma:prob_acceptance_general} and \ref{proof:lemma:mse_general}, respectively. This scalar transformation is crucial, as it allows us to reformulate $c_{\eta}(\alpha)$ as an optimization over the one dimensional density $f_Z(z)$ instead of the vast $N$ dimensional density $g(\cdot)$.

Furthermore, Lemma \ref{lemma:optimal_support} proves that we lose no optimality by restricting the support of $Z$ to 
$z \in [(\eta-1)\Delta, (\eta+1)\Delta]$.
Similarly, Lemma \ref{lemma:exact_acc_noise_existence} establishes that simplifying the constraint from $\PA(g(\cdot), \eta) \geq \alpha$
in \eqref{c_small_definition} to 
$\PA(g(\cdot), \eta) = \alpha$,
does not result in any loss of optimality. Following these simplifications, we define the random variable 
$Q \triangleq \Phi_N(Z)$. According to the lemmas above, the constraints and the objective function in \eqref{c_small_definition} can be rewritten as
\begin{align}
    \PA(g(\cdot), \eta) &= \mathbb{E}[\Phi_N(Z)] = \mathbb{E}[Q] = \alpha, \\
    \MSE(g(\cdot), \eta) &= \frac{1}{4\alpha}\mathbb{E}[\Psi_N(Z)] = \frac{1}{4\alpha}\mathbb{E}[\tilde{\Psi}_N(Q)],
\end{align}
where we define\footnote{The function $\Phi_N(z)$ is strictly decreasing over the domain $[(\eta-1)\Delta, (\eta+1)\Delta]$, making it a bijection and thus invertible over the range $[0, 1]$.}
$\tilde{\Psi}_N(q) \triangleq \Psi_N(\Phi_N^{-1}(q))$.
Consequently, the optimization problem in \eqref{c_small_definition} turns to a maximization over the distribution of the random variable $Q$:
\begin{align}
    \text{Maximize: } & \frac{1}{4\alpha} \mathbb{E}[\tilde{\Psi}_N(Q)] \label{eq:opt_q_expectation} \\
    \text{Subject to: } & \mathbb{E}[Q] = \alpha, \quad Q \in [0, 1].
\end{align}

To solve \eqref{eq:opt_q_expectation}, let $\tilde{\Psi}_N^*(\cdot)$ denote the upper concave envelope of $\tilde{\Psi}_N(\cdot)$. By definition, we have $\frac{1}{4\alpha} \mathbb{E}[\tilde{\Psi}_N(Q)] \leq \frac{1}{4\alpha} \mathbb{E}[\tilde{\Psi}_N^*(Q)]$. Moreover, since the upper concave envelope is a concave function, applying Jensen's inequality yields 
\begin{align}\label{intuition_for_c_eta}
    \frac{1}{4\alpha} \mathbb{E}[\tilde{\Psi}_N^*(Q)] \leq \frac{1}{4\alpha} \tilde{\Psi}_N^*(\mathbb{E}[Q]) = \frac{\tilde{\Psi}_N^*(\alpha)}{4 \alpha}.
\end{align}
Equation \eqref{intuition_for_c_eta} demonstrates that $c_{\eta}(\alpha)$ is bounded from above by $\frac{\tilde{\Psi}_N^*(\alpha)}{4 \alpha}$. To rigorously complete the intuitive proof of Theorem \ref{theorem:Main_Minimax_Result}, we must establish that this theoretical upper bound is strictly achievable. To do so, we evaluate two distinct cases based on the geometry of the function. First, consider the regions where the function naturally coincides with its upper concave envelope, meaning $\tilde{\Psi}_N(\alpha) = \tilde{\Psi}_N^*(\alpha)$. In this scenario, we can simply set the random variable $Q$ to be exactly equal to the deterministic value $\alpha$. Practically, since $Q = \Phi_N(Z)$, this assignment dictates that $\Phi_N(Z) = \alpha$, which means we must construct a noise distribution in the $N$ dimensional space where the magnitude of the noise is exactly $z_1 = \Phi_N^{-1}(\alpha)$. To physically achieve this, the adversary samples the noise vector uniformly at random from the  surface of an $N$ dimensional ball with a radius of $z_1$.

Conversely, in regions where the original function $\tilde{\Psi}_N(\cdot)$ exhibits a convex dip, it strictly falls below its concave envelope. In these specific intervals, the adversary bridges the geometric gap using a linear chord, as depicted in Figure \ref{fig:envelope_achievability}. This mathematical chord physically represents a mixed strategy between two distinct noise magnitudes. More precisely, by strategically randomizing the noise distribution between the surfaces of two $N$ dimensional balls of different radii, the adversary perfectly interpolates and bridges the geometric gap to reach the upper concave envelope. This careful randomization ensures that the upper bound is completely achievable for all possible values of $\alpha$, thereby concluding the structural derivation of the worst case attack.

\begin{remark}
    In Theorem \ref{theorem:Main_Minimax_Result}, we characterize the function $c_{\eta}(\alpha)$ defined in \eqref{c_small_definition} for general dimensions $N \geq 1$. It is worth noting that if we choose $N=1$, the characterization of $c_{\eta}(\alpha)$ reduces to the one-dimensional case, which have been evaluated and analyzed previously in \cite{GoCJournal}. Specifically, the explicit functions for that specific case have been characterized in Appendix G of \cite{GoCJournal}, and one can verify the consistency of the general result. In addition, to provide a concrete example of the multidimensional setting, we explicitly evaluate this function for the case of $N=2$ in Appendix \ref{sec:case_N2}.
\end{remark}

\begin{remark}\label{rem:Numerical_Calculation}
    One might initially view the calculation of $c_{\eta}(\alpha)$ in \eqref{eq:c_eta_result} as analytically intractable, particularly because the function $\Phi_N(z)$ involves transcendental terms (e.g., for even $N$) or high-order polynomials (for odd $N$) that do not admit a closed-form inverse. Consequently, obtaining an explicit expression for the composite function $\tilde{\Psi}_N(q)$ is generally not possible. However, numerically evaluating the concave envelope is straightforward and does not require explicit inversion. Instead, one can adopt a parametric approach: by sweeping the variable $z$ across its domain $[(\eta-1)\Delta, (\eta+1)\Delta]$, we generate the locus of points $(q_z, y_z) = (\Phi_N(z), \Psi_N(z))$. The function $\tilde{\Psi}_N^*(q)$ is then simply the upper boundary of the convex hull of this set of points, which can be efficiently computed using standard numerical libraries. We have used this technique to derive these functions for different settings and finally determined the equilibrium for different cases, as described in Section \ref{sec:Illustrative_Examples}.

\end{remark}

    In the detailed proof of Theorem \ref{theorem:Main_Minimax_Result} provided in Section \ref{proof:theorem:Main_Minimax_Result}, we not only derive the worst-case error bound but also explicitly characterize the adversarial noise distribution that achieves this bound. This optimal noise density, denoted as $g_{\mathbf{N}_a}^*(\mathbf{x})$, is constructed in Algorithm \ref{Alg:finding_noise_N_dim}. The algorithm utilizes the geometric properties of $\Psi_N$ defined in \eqref{eq:Psi_n_def}, and $\Phi_N$ defined in \eqref{eq:Phi_n_def}, to determine whether a single spherical shell or a mixture of two spherical shells constitutes the optimal noise distribution.

\begin{remark}
    It is worth emphasizing that the results established in Theorems \ref{theorem: equivalence_two_problem} and \ref{theorem:Main_Minimax_Result}, as well as the procedures in Algorithms \ref{Alg:finding_eta} and \ref{Alg:finding_noise_N_dim}, do not rely on specific functional forms for the utilities of the DC or the adversary. We only impose the intuitive conditions that the adversary's utility is strictly increasing with respect to both arguments, whereas the DC's utility is non-increasing in its first argument and non-decreasing in its second. These broad and  common-sense assumptions ensure that our framework remains versatile enough to encompass a wide array of practical security and estimation scenarios without loss of generality.
\end{remark}

\begin{remark}\label{remark_for_unknown}
It is worth noting a subtle point regarding the robustness and universality of the proposed solution. On the one hand, the optimization problem for $c_{\eta}(\alpha)$ in \eqref{c_small_definition} is solved entirely irrespective of the utility functions. Specifically, both the exact value of $c_{\eta}(\alpha)$ and the specific noise distribution that achieves this bound and satisfies the required conditions in \eqref{c_small_definition} are characterized independently of any utility assumptions. On the other hand, characterizing the actual best response of the adversary for each given value of $\eta$, which consequently determines the final estimation error and probability of acceptance, strictly requires knowledge of the adversary's utility function. Furthermore, finding the optimal threshold $\eta^*$ for the DC requires full knowledge of the utility functions of both players. 
\end{remark}

As mentioned in Remark \ref{remark_for_unknown} the validity of Theorem~\ref{theorem: equivalence_two_problem} requires the underlying game to be complete, in the sense that both the DC has full knowledge of adversary's utility functions. Hence, the current solution may fail if there is a mismatch between the actual utility functions, and what the other party thinks. The game of coding with unknown adversary's utility function is solved in \cite{nodehi2025unknown} for the scalar case. While a similar approach may be generalizable for the $N$-dimensional case, a rigorous  analysis for the vector case remains for future work.

\begin{algorithm}[t]
\caption{Characterizing the Optimal $N$-Dimensional Adversarial Noise Distribution}
\label{Alg:finding_noise_N_dim}
\begin{algorithmic}[1]
\State \textbf{Input:} Dimension $N$, decoding threshold $\eta$, bound $\Delta$, the utility function $Q_{\mathsf{AD}}(\cdot, \cdot)$, and the derived function $c_{\eta}(\cdot)$ from Theorem \ref{theorem:Main_Minimax_Result}.
\State \textbf{Output:} The optimal noise distribution PDF $g_{\mathbf{N}_a}^*(\mathbf{x})$.

\State Define $\Phi_N(z)$ as in \eqref{eq:Phi_n_def}, and  $\Psi_N(z)$ as in \eqref{eq:Psi_n_def}.
\State Define $\tilde{\Psi}_N(q) \triangleq \Psi_N\left( \Phi_N^{-1}(q) \right)$ for $q \in [0, 1]$.
\State Let $\tilde{\Psi}_N^*(q)$ denote the upper concave envelope of $\tilde{\Psi}_N(q)$ over $q \in [0, 1]$.

\vspace{1em}
\State \textbf{Step 1: Optimal Operating Point Selection}
\State Calculate the optimal acceptance probability $\alpha^*$ that maximizes the adversary's objective
\begin{align*}
    \alpha^* = \underset{0 < \alpha \leq 1 }{\arg\max} ~Q_{\mathsf{AD}}(c_{\eta} (\alpha), \alpha).
\end{align*}

\vspace{1em}
\State \textbf{Step 2: Construction of Noise Distribution}
\If {$\tilde{\Psi}_N^*(\alpha^*) = \tilde{\Psi}_N(\alpha^*)$}
    \State \textit{// Case 1: The function lies on its concave envelope.}
    \State Calculate the optimal noise radius: $z^* = \Phi_N^{-1}(\alpha^*)$.
    \State Output the distribution uniform over a single $N$-sphere of radius $z^*$:
    \begin{align*}
        g_{\mathbf{N}_a}^*(\mathbf{x}) = \frac{1}{S_N(z^*)} \delta(\|\mathbf{x}\|_2 - z^*),
    \end{align*}
    where $S_N(r) = \frac{2\pi^{N/2}}{\Gamma(N/2)}r^{N-1}$.
\Else
    \State \textit{// Case 2: The function lies below its concave envelope.}
    \State Find probabilities $q_1 < \alpha^* < q_2$ such that the envelope touches the function at the endpoints:
    \begin{align*}
        \tilde{\Psi}_N^*(q_1) = \tilde{\Psi}_N(q_1) \quad \text{and} \quad \tilde{\Psi}_N^*(q_2) = \tilde{\Psi}_N(q_2),
    \end{align*}
    and is linear in between.
    \State Calculate the corresponding radii: $z_1 = \Phi_N^{-1}(q_1)$ and $z_2 = \Phi_N^{-1}(q_2)$.
    \State Calculate the mixing weights:
    \begin{align*}
        \beta_1 = \frac{q_2 - \alpha^*}{q_2 - q_1}, \quad \beta_2 = \frac{\alpha^* - q_1}{q_2 - q_1}.
    \end{align*}
    \State Output the mixture distribution uniform over two $N$-spheres:
    \begin{align*}
        g_{\mathbf{N}_a}^*(\mathbf{x}) = \beta_1 \frac{1}{S_N(z_1)} \delta(\|\mathbf{x}\|_2 - z_1) + \beta_2 \frac{1}{S_N(z_2)} \delta(\|\mathbf{x}\|_2 - z_2).
    \end{align*}
\EndIf
\end{algorithmic}
\end{algorithm}

\section{Proof of Theorem \ref{theorem: equivalence_two_problem}}\label{proof:theorem: equivalence_two_problem}

In this section, we establish the validity of Theorem \ref{theorem: equivalence_two_problem}. We begin by comparing the optimization performed in Algorithm \ref{Alg:finding_eta} with the theoretical definition of $\eta^*$ in \eqref{eq:etastar}.
Algorithm \ref{Alg:finding_eta} computes $\hat{\eta}$ by solving the following optimization problem
\begin{align}\label{eq:alg_optimization}
    \hat{\eta} = \underset{\eta \in \Lambda_{\mathsf{DC}}}{\arg\max} ~ \underset{\alpha \in \mathcal{L}_{\eta}}{\min} ~ Q_{\mathsf{DC}} \left(c_{\eta}(\alpha), \alpha\right),
\end{align}
where $\mathcal{L}_{\eta}$ is defined in Algorithm \ref{Alg:finding_eta} as
\begin{align}\label{def:L_eta}
    \mathcal{L}_{\eta} = \underset{0 < \alpha \leq 1}{\arg\max} ~ Q_{\mathsf{AD}}(c_{\eta}(\alpha), \alpha).
\end{align}

In contrast, based on \eqref{eq:etastar}, the value of $\eta^*$ can be reformulated in terms of the set of realizable performance pairs. More precisely, let $\mathcal{J}_{\eta}$ denote the set of operating points corresponding to the adversary's best responses
\begin{align}\label{J_definition}
    \mathcal{J}_{\eta} \triangleq \left\{ \left( \mathsf{MSE}(g(\cdot), \eta), \mathsf{PA}(g(\cdot), \eta) \right) \mathrel{\Big|} g(\cdot) \in \mathcal{B}^{\eta}_{\text{AD}} \right\}.
\end{align}
Using this set, based on \eqref{eq:etastar}, the value of $\eta^*$ is given by
\begin{align}\label{seprating_utilities}
    \eta^* = \underset{\eta \in \Lambda_{\mathsf{DC}}}{\arg\max} ~ \underset{(\beta, \alpha) \in \mathcal{J}_{\eta}}{\min} ~ Q_{\mathsf{DC}} (\beta, \alpha).
\end{align}
Comparing \eqref{eq:alg_optimization} and \eqref{seprating_utilities}, it is evident that to prove $\hat{\eta} = \eta^*$, it suffices to demonstrate that the set of best-response points $\mathcal{J}_{\eta}$ is identical to the set of points derived from the algorithm. More precisely, let us define the set $\mathcal{K}_{\eta}$ as
\begin{align}\label{KK_definition}
    \mathcal{K}_{\eta} \triangleq \left\{ (c_{\eta}(\alpha), \alpha) \mid \alpha \in \mathcal{L}_{\eta} \right\}.
\end{align}
Thus, the proof of Theorem \ref{theorem: equivalence_two_problem} reduces to showing that $\mathcal{J}_{\eta} = \mathcal{K}_{\eta}$. We establish this equality by proving mutual inclusion: first showing $\mathcal{J}_{\eta} \subseteq \mathcal{K}_{\eta}$, and subsequently $\mathcal{K}_{\eta} \subseteq \mathcal{J}_{\eta}$. The intermediate steps are formally shown the Sections~\ref{subsec:JinK} and~\ref{subsec:KinJ} below.

Before proceeding with the main inclusion arguments, we first state and prove the following lemma.

\begin{lemma}
\label{lemmaJSC}
Let define the set $\mathcal{C}_{\eta}$ as
\begin{align}\label{defining_the_Set_C_eta}
    \mathcal{C}_{\eta} \triangleq \left\{ \left( c_{\eta}(\alpha), \alpha \right) ~\middle|~ 0 < \alpha \leq 1 \right\}.
\end{align}
Then, for any threshold $\eta \in \Lambda_{\mathsf{DC}}$, the set of adversarial best responses $\mathcal{J}_{\eta}$ satisfies $\mathcal{J}_{\eta} \subseteq \mathcal{C}_{\eta}$.
\end{lemma}
\begin{proof}
    Consider an arbitrary operating point $(\beta, \alpha) \in \mathcal{J}_{\eta}$ resulting from an adversarial best-response $g^*(\cdot) \in \mathcal{B}^{\eta}_{\text{AD}}$. By the definition in \eqref{J_definition}, we have $\alpha = \mathsf{PA}(g^*(\cdot), \eta)$ and $\beta = \mathsf{MSE}(g^*(\cdot), \eta)$. That means $g^*(\cdot)$ satisfies the constraint of the optimization problem in~\eqref{c_small_definition}. Therefore, the value of the objective function in~\eqref{c_small_definition} at the feasible point $g^*(\cdot)$, i.e., $ \mathsf{MSE}(g^*(\cdot), \eta)=\beta$, cannot exceed the maximum of the objective function, which is $c_\eta(\alpha)$. This immediately implies $\beta \leq c_{\eta}(\alpha)$. 

    We prove that equality must hold by contradiction. Suppose that $\beta < c_{\eta}(\alpha)$, represented by point $A$ (red one) in Figure \ref{fig:lemma_visual}. By the definition of the characteristic function $c_{\eta}(\alpha)$ in \eqref{c_small_definition}, there must exist an alternative distribution $g'(\cdot)$, corresponding to point $B$ (black one) in Figure \ref{fig:lemma_visual}, such that
    \begin{align}
        \mathsf{MSE}(g'(\cdot), \eta) = c_{\eta}(\alpha), \label{eq:MSEE_new} \\
        \mathsf{PA}(g'(\cdot), \eta) \geq \alpha. \label{eq:PAAA_new}
    \end{align}
    
    Comparing the utilities, we observe that
    \begin{align}\label{eq:utility_contradiction}
        \mathsf{U}_{\mathsf{AD}}(g'(\cdot), \eta) 
        &= Q_{\mathsf{AD}}\left(\mathsf{MSE}(g'(\cdot), \eta), \mathsf{PA}(g'(\cdot), \eta)\right) \nonumber \\
        &\overset{(a)}{=} Q_{\mathsf{AD}}\left(c_{\eta}(\alpha), \mathsf{PA}(g'(\cdot), \eta)\right) \nonumber \\
        &\overset{(b)}{\geq} Q_{\mathsf{AD}}\left(c_{\eta}(\alpha), \alpha\right) \nonumber \\
        &\overset{(c)}{>} Q_{\mathsf{AD}}(\beta, \alpha) \nonumber \\
        &= \mathsf{U}_{\mathsf{AD}}(g^*(\cdot), \eta),
    \end{align}
    where (a) follows from \eqref{eq:MSEE_new}; (b) follows from \eqref{eq:PAAA_new} and the fact that $Q_{\mathsf{AD}}$ is non-decreasing in its second argument; and (c) holds because $Q_{\mathsf{AD}}$ is strictly increasing in its first argument and $c_{\eta}(\alpha) > \beta$. This strictly higher utility for $g'(\cdot)$ contradicts our initial assumption that $g^*(\cdot)$ is a best response in~$\mathcal{B}^{\eta}_{\text{AD}}$. Consequently, we must have $\beta = c_{\eta}(\alpha)$, which implies the point $(\beta, \alpha)$ lies within $\mathcal{C}_{\eta}$ defined in \eqref{defining_the_Set_C_eta}.
\end{proof}
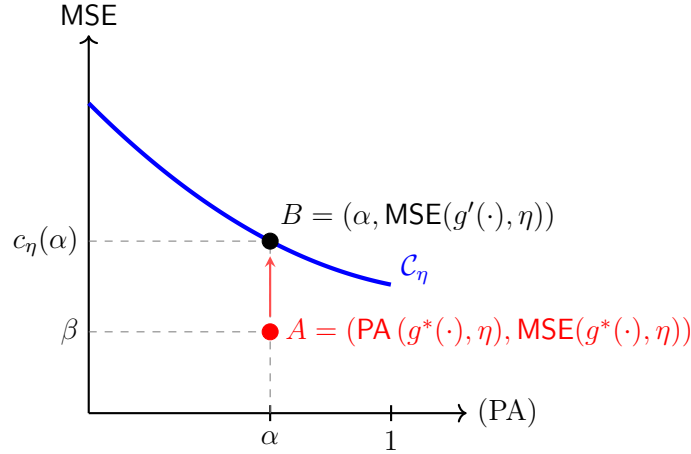
\begin{figure}[ht]
    \centering
    \begin{tikzpicture}[scale=2]
        \draw[->, thick] (0,0) -- (2.5,0) node[right] { (PA)};
        \draw[->, thick] (0,0) -- (0,2.5) node[above] {$\MSE$};
        
        \draw[thick] (1.2, 1.5pt) -- (1.2, -1.5pt) node[below] {$\alpha$};
        \draw[thick] (2, 1.5pt) -- (2, -1.5pt) node[below] {$1$};
        
        \draw[ultra thick, blue, samples=100, domain=0:2] plot (\x, {0.2*(\x-2.5)*(\x-2.5) + 0.8});
        \node[blue, right] at (2, 0.95) {$\mathcal{C}_{\eta}$};
        
        \def\px{1.2}
        \def\py{1.138} 
        
        \draw[dashed, gray] (\px, 0) -- (\px, \py);
        \draw[dashed, gray] (0, \py) -- (\px, \py);
        \node[left] at (0, \py) {$c_{\eta}(\alpha)$};
        
        \draw[dashed, gray] (0, \py-0.6) -- (\px, \py-0.6);
        \node[left] at (0, \py-0.6) {$\beta$};

        \filldraw[black] (\px, \py) circle (1.5pt) node[above right] {$B=(\alpha, \MSE(g'(\cdot),\eta))$};
        
        \filldraw[red] (\px, \py-0.6) circle (1.5pt) node[right, xshift=2pt] {$A=(\mathsf{PA}\left(g^*(\cdot), \eta), \MSE(g^*(\cdot),\eta)\right)$};
        
        \draw[->, >=stealth, thick, red!70] (\px, \py-0.5) -- (\px, \py-0.1);
        
    \end{tikzpicture}
    \caption{Geometric proof of Lemma \ref{lemmaJSC}. Point A represents a suboptimal response where $\beta < c_{\eta}(\alpha)$. By choosing $g'$ rather than $g^*$ to move to Point B on the boundary $\mathcal{C}_{\eta}$ (where the probability of acceptance is at least $\alpha$), the adversary increases their MSE and potentially their probability of acceptance, leading to strictly higher utility.}
    \label{fig:lemma_visual}
\end{figure}

We now prove $\mathcal{J}_{\eta} = \mathcal{K}_{\eta}$ via double inclusion, in the following sections.

\subsection{\texorpdfstring{Proof of $\mathcal{J}_{\eta} \subseteq \mathcal{K}_{\eta}$}{Proof of J subset K}}
\label{subsec:JinK}
Consider an arbitrary pair $(\beta, \alpha) \in \mathcal{J}_{\eta}$, denoted as point $A$ in Figure \ref{fig:J_subset_K_visual}. We claim that $(\beta, \alpha) \in \mathcal{K}_{\eta}$, and we prove this by contradiction.

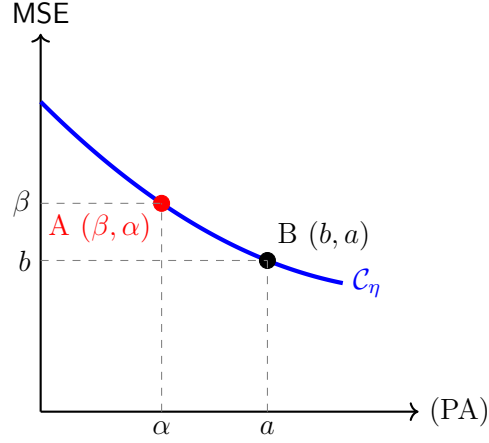
\begin{figure}[ht]
    \centering
    \begin{tikzpicture}[scale=2]
        \draw[->, thick] (0,0) -- (2.5,0) node[right] {(PA)};
        \draw[->, thick] (0,0) -- (0,2.5) node[above] {$\MSE$};
        
        \draw[ultra thick, blue, samples=100, domain=0:2] plot (\x, {0.2*(\x-2.5)*(\x-2.5) + 0.8});
        \node[blue, right] at (2, 0.85) {$\mathcal{C}_{\eta}$};
        
        \def\ax{0.8}
        \def\ay{1.378} 
        \filldraw[red] (\ax, \ay) circle (1.5pt) node[below left] {A $(\beta, \alpha)$};
        \draw[dashed, gray] (\ax, 0) -- (\ax, \ay);
        \draw[dashed, gray] (0, \ay) -- (\ax, \ay);
        
        \def\bx{1.5}
        \def\by{1.0} 
        \filldraw[black] (\bx, \by) circle (1.5pt) node[above right] {B $(b, a)$};
        \draw[dashed, gray] (\bx, 0) -- (\bx, \by);
        \draw[dashed, gray] (0, \by) -- (\bx, \by);
        
        \node[below] at (\ax, 0) {$\alpha$};
        \node[left] at (0, \ay) {$\beta$};
        \node[below] at (\bx, 0) {$a$};
        \node[left] at (0, \by) {$b$};
        
    \end{tikzpicture}
    \caption{Visual representation of the contradiction for $\mathcal{J}_{\eta} \subseteq \mathcal{K}_{\eta}$. Point A is a best response and thus lies on $\mathcal{C}_{\eta}$ by Lemma \ref{lemmaJSC}. If A is not in $\mathcal{K}_{\eta}$, there must exist a point B on the same boundary $\mathcal{C}_{\eta}$ that provides strictly higher utility, contradicting the optimality of A.}
    \label{fig:J_subset_K_visual}
\end{figure}

Assume, as a contradictory hypothesis, that $A = (\beta, \alpha) \notin \mathcal{K}_{\eta}$. By Lemma \ref{lemmaJSC}, we know that every adversarial best response lies on the boundary, so $A \in \mathcal{C}_{\eta}$, i.e., $\beta = c_\eta(\alpha)$. Note that if $\alpha \in \mathcal{L}_{\eta}$, then from the definition of $\mathcal{K}_{\eta}$ in \eqref{KK_definition}, we would have $A\in \mathcal{K}_{\eta}$. Since we assumed $A \notin \mathcal{K}_{\eta}$, we can conclude that $\alpha \notin \mathcal{L}_{\eta}$. Then, from the definition of $\mathcal{L}_{\eta}$ in \eqref{def:L_eta}, there must exists another $a\in \mathcal{L}_{\eta}$ and another point $B = (b,a) \in \mathcal{K}_{\eta}$ with $b=c_\eta(a)$ that yields a strictly higher adversarial utility than $A$. That is,
\begin{align}\label{JnsL_contradiction_assuption}
    Q_{\mathsf{AD}}(b,a) > Q_{\mathsf{AD}}(\beta, \alpha).
\end{align}
Since $(\beta, \alpha) \in \mathcal{J}_{\eta}$, based on the definition in \eqref{J_definition}, there exists a noise distribution $g_{\alpha}(\cdot) \in \mathcal{B}^{\eta}_{\mathsf{AD}}$ where
\begin{align}\label{definition_of_beta_alpha}
    (\beta, \alpha) = \left( \mathsf{MSE}(g_{\alpha}(\cdot), \eta), \mathsf{PA}(g_{\alpha}(\cdot), \eta) \right).
\end{align}
Similarly, for the point $(b,a) \in \mathcal{K}_{\eta}$, since  $b=c_\eta(a)$, from the optimization problem in~\eqref{c_small_definition}, there exists a noise distribution $g_b(\cdot) \in \Lambda_{\mathsf{AD}}$ such that
\begin{align}\label{b_to_MMSE_relation}
    b = \mathsf{MSE}(g_{b}(\cdot), \eta),
\end{align}
and
\begin{align}\label{a_to_PA_relationship}
    \mathsf{PA} \left( g_b(\cdot), \eta \right) \geq a.
\end{align}
Specifically, $g_b(\cdot)$ is an optimal solution to the following maximization problem
\begin{align}\label{gb_to_opt}
    \max_{g(\cdot) \in \Lambda_{\mathsf{AD}}} \left\{ \mathsf{MSE}(g(\cdot), \eta) ~\middle| ~\mathsf{PA}(g(\cdot), \eta) \geq a \right\}.
\end{align}
We can now evaluate the utility of the distribution $g_b(\cdot)$ as follows
\begin{align}\label{JnsL_contradiction_point}
    \mathsf{U}_{\mathsf{AD}}(g_b(\cdot), \eta) &= Q_{\mathsf{AD}}\left( \mathsf{MSE}(g_{b}(\cdot), \eta), \mathsf{PA}(g_b(\cdot), \eta) \right) \nonumber \\
    &\overset{(a)}{=} Q_{\mathsf{AD}}\left( b, \mathsf{PA}(g_b(\cdot), \eta) \right) \nonumber \\
    &\overset{(b)}{\geq} Q_{\mathsf{AD}}(b, a) \nonumber \\
    &\overset{(c)}{>} Q_{\mathsf{AD}}(\beta, \alpha) \nonumber \\
    &\overset{(d)}{=} Q_{\mathsf{AD}}\left( \mathsf{MSE}(g_{\alpha}(\cdot), \eta), \mathsf{PA}(g_{\alpha}(\cdot), \eta) \right) \nonumber \\
    &= \mathsf{U}_{\mathsf{AD}}(g_{\alpha}(\cdot), \eta),
\end{align}
where (a) follows from \eqref{b_to_MMSE_relation}; (b) follows from \eqref{a_to_PA_relationship} and the non-decreasing property of $Q_{\mathsf{AD}}$; (c) follows from the contradictory assumption in \eqref{JnsL_contradiction_assuption}; and (d) follows from \eqref{definition_of_beta_alpha}.

The result of \eqref{JnsL_contradiction_point} implies $\mathsf{U}_{\mathsf{AD}}(g_b(\cdot), \eta) > \mathsf{U}_{\mathsf{AD}}(g_{\alpha}(\cdot), \eta)$. This is in contradiction with the fact that $g_{\alpha}(\cdot)\in \mathcal{B}_{\mathsf{AD}}^\eta$ is a best response strategy, meaning no other strategy  including $g_b(\cdot)$, can yield strictly higher adversarial utility. Therefore, our initial assumption was incorrect, and we must have $(\beta, \alpha) \in \mathcal{K}_{\eta}$. Consequently, $\mathcal{J}_{\eta} \subseteq \mathcal{K}_{\eta}$.

\subsection{\texorpdfstring{Proof of $\mathcal{K}_{\eta} \subseteq \mathcal{J}_{\eta}$}{Proof of K subset J}}
\label{subsec:KinJ}

Consider an arbitrary point $(b,a) \in \mathcal{K}_{\eta}$. Let $g_b(\cdot) \in \Lambda_{\mathsf{AD}}$ be the noise distribution associated with $(b,a)$, where the relationships \eqref{b_to_MMSE_relation}, \eqref{a_to_PA_relationship}, and \eqref{gb_to_opt} hold.

Now, consider a point $(\beta, \alpha) \in \mathcal{J}_{\eta}$, and let $g_{\alpha}(\cdot) \in \mathcal{B}^{\eta}_{\mathsf{AD}}$ be the corresponding noise distribution such that \eqref{definition_of_beta_alpha} holds. Note that based on Lemma \ref{lemmaJSC}, we have $(\beta, \alpha) \in \mathcal{C}_{\eta}$, i.e., $\beta = c_\eta(\alpha)$. 

Since $(b,a) \in \mathcal{K}_{\eta}$, by the definition~\eqref{KK_definition}, we should have $b=c_\eta(a)$ and $a\in \mathcal{L}_{\eta}$, which together with~\eqref{def:L_eta} further implies, $Q_{\mathsf{AD}}\left( c_\eta(a),a \right) \geq Q_{\mathsf{AD}}\left( c_\eta(a'), a' \right)$ for any $0<a'\leq 1 $. In particular, for $a'=\alpha$, this implies 
\begin{align}\label{ba_tobetaalpha_relation}
   Q_{\mathsf{AD}}\left( b,a \right) \geq Q_{\mathsf{AD}}\left( \beta, \alpha \right).
\end{align}
Consider the following chain of inequalities regarding the utility of $g_b(\cdot)$
\begin{align}\label{jset_to_lset}
    \mathsf{U}_{\mathsf{AD}}\left(g_b(\cdot), \eta \right) &= Q_{\mathsf{AD}}\bigl( \mathsf{MSE}\left(g_{b}(\cdot), \eta \right), \mathsf{PA} \left(g_{b}(\cdot), \eta \right) \bigr) \nonumber \\
    &\overset{(a)}{=} Q_{\mathsf{AD}}\big( b, \mathsf{PA} \left(g_{b}(\cdot), \eta \right) \big) \nonumber \\
    &\overset{(b)}{\geq} Q_{\mathsf{AD}}\left( b, a \right) \nonumber \\
    &\overset{(c)}{\geq} Q_{\mathsf{AD}}\left( \beta, \alpha \right) \nonumber \\
    &\overset{(d)}{=} Q_{\mathsf{AD}}\big( \mathsf{MSE}\left(g_{\alpha}(\cdot), \eta \right), \mathsf{PA} \left( g_{\alpha}(\cdot), \eta \right) \big) \nonumber \\
    &=\mathsf{U}_{\mathsf{AD}}\left(g_{\alpha}(\cdot), \eta \right),
\end{align}
where (a) follows from \eqref{b_to_MMSE_relation}; (b) follows from \eqref{a_to_PA_relationship} and the fact that $Q_{\mathsf{AD}}(\cdot,\cdot)$ is a strictly increasing function with respect to its second argument; (c) follows from \eqref{ba_tobetaalpha_relation}; and (d) follows from \eqref{definition_of_beta_alpha}.

The result of \eqref{jset_to_lset} implies that $\mathsf{U}_{\mathsf{AD}}\left(g_b(\cdot), \eta \right) \geq \mathsf{U}_{\mathsf{AD}}\left(g_{\alpha}(\cdot), \eta \right)$. On the other hand, since ${g_{\alpha}(\cdot) \in \mathcal{B}^{\eta}_{\mathsf{AD}}}$ is a global best response, we must have $\mathsf{U}_{\mathsf{AD}}\left(g_{\alpha}(\cdot), \eta \right) \geq \mathsf{U}_{\mathsf{AD}}\left(g(\cdot), \eta \right)$ for any $g(\cdot)$, and in particular $g(\cdot) = g_b(\cdot)$. Combining these two inequalities, we conclude that
\begin{align}\label{eq:equality_utilities}
    \mathsf{U}_{\mathsf{AD}}\left(g_b(\cdot), \eta \right) = \mathsf{U}_{\mathsf{AD}}\left(g_{\alpha}(\cdot), \eta \right).
\end{align}
Consequently, all inequalities in the chain \eqref{jset_to_lset} must hold with equality. Specifically, looking at step (b) of \eqref{jset_to_lset}, we must have
\begin{align}\label{eq:equal_PA_Q}
    Q_{\mathsf{AD}}\big( b, \mathsf{PA} \left(g_{b}(\cdot), \eta \right) \big) = Q_{\mathsf{AD}}\left( b, a \right).
\end{align}
Since $Q_{\mathsf{AD}}(\cdot, \cdot)$ is a strictly increasing function with respect to its second argument, this equality holds if and only if $\mathsf{PA} \left(g_{b}(\cdot), \eta \right) = a$. This together with~\eqref{b_to_MMSE_relation} implies
\begin{align}\label{eq:point_identity}
    (b,a) = \bigl( \mathsf{MSE}\left(g_{b}(\cdot), \eta \right), \mathsf{PA} \left(g_{b}(\cdot), \eta \right) \bigr).
\end{align}
Finally, since $g_{\alpha}(\cdot) \in \mathcal{B}^{\eta}_{\mathsf{AD}}$ and we showed in \eqref{eq:equality_utilities} that $\mathsf{U}_{\mathsf{AD}}\left(g_b(\cdot), \eta \right) = \mathsf{U}_{\mathsf{AD}}\left(g_{\alpha}(\cdot), \eta \right)$, it implies that $g_b(\cdot)$ achieves the global maximum utility. Therefore, we have $g_b(\cdot) \in \mathcal{B}^{\eta}_{\mathsf{AD}}$, as defined in~\eqref{eq:best_set_adv_definition}. By the definition in \eqref{J_definition}, this means $(b,a) \in \mathcal{J}_{\eta}$, which holds for every $(b,a)\in \mathcal{K}_{\eta}$. Therefore,  $\mathcal{K}_{\eta} \subseteq \mathcal{J}_{\eta}$.

\vspace{0.5em}
\noindent \textbf{Conclusion:} Having established both inclusions, we conclude that $\mathcal{J}_{\eta} = \mathcal{K}_{\eta}$. Substituting $\mathcal{K}_{\eta}$ for $\mathcal{J}_{\eta}$ in \eqref{seprating_utilities} completes the proof of Theorem \ref{theorem: equivalence_two_problem}.

\section{Proof of Theorem \ref{theorem:Main_Minimax_Result}}\label{proof:theorem:Main_Minimax_Result}

In this section, we provide the proof of the result stated in Theorem \ref{theorem:Main_Minimax_Result}. Our first step is to derive a general expression for the probability of acceptance, $\Pr(\mathcal{A}_\eta)$, as a function of the magnitude of the adversarial noise. Specifically, let us define $Z$ as the Euclidean norm of the adversarial noise
\begin{align}\label{eq:Z_definition}
    Z \triangleq \|\mathbf{N}_a\|_2.
\end{align}
We assume $Z$ is distributed according to a general probability density function 
\begin{align}
    f_Z(z)=\int_{\mathbf{n}_a: \|\mathbf{n}_a\|=z} g(\mathbf{n}_a) d \mathbf{n}_a,
\end{align}
supported on $[0, \infty)$. As established in \eqref{eq:PA_def}, the DC accepts the computation if and only if the Euclidean distance between the two reports, $\mathbf{Y}_1$ and $\mathbf{Y}_2$, does not exceed the threshold $\eta \Delta$. Recall from the system model that the reports are given by $\mathbf{Y}_h = \mathbf{U} + \mathbf{N}_h$ and $\mathbf{Y}_a = \mathbf{U} + \mathbf{N}_a$, where $\mathbf{U}$ is the ground truth. When the DC computes the difference between the two reports, the common ground truth signal $\mathbf{U}$ cancels out entirely
\begin{align}\label{eq:noise_diff}
    \|\mathbf{Y}_1 - \mathbf{Y}_2\|_2 &= \|(\mathbf{U} + \mathbf{N}_h) - (\mathbf{U} + \mathbf{N}_a)\|_2 \nonumber \\
    &= \|\mathbf{N}_h - \mathbf{N}_a\|_2.
\end{align}
Consequently, the acceptance condition $\|\mathbf{Y}_1 - \mathbf{Y}_2\|_2 \le \eta \Delta$ reduces to a constraint on the relative distance between the noise vectors
\begin{align}\label{acceptance_new}
    \mathcal{A}_\eta : \quad \|\mathbf{N}_h - \mathbf{N}_a\|_2 \le \eta \Delta.
\end{align}

In order to prove Theorem \ref{theorem:Main_Minimax_Result}, we first show that the probability of acceptance, $\mathsf{PA}(g(\cdot), \eta)$, only depends on the adversarial noise through the distribution of it the magnitude distribution $f_Z(z)$, as shown in the following lemma.

\begin{lemma}\label{lemma:prob_acceptance_general}
    For any adversarial noise distribution characterized by the marginal magnitude PDF $f_Z(z)$, the probability of acceptance is given by
    \begin{align}\label{eq:prob_acceptance_lemma}
        \Pr(\mathcal{A}_\eta) = \int_{0}^{\infty} \Phi_N(z) f_Z(z) \, dz,
    \end{align}
    where $\Phi_N(z)$ is defined piecewise as
    \begin{align}\label{eq:Phi_piecewise}
        \Phi_N(z) = 
        \begin{cases} 
            1 & \text{if } 0 \le z \le (\eta-1)\Delta, \\[6pt]
            \frac{V_{\text{lens}}(\Delta, \eta\Delta, z)}{V_N(\Delta)} & \text{if } (\eta-1)\Delta \leq z \leq (\eta+1)\Delta, \\[6pt]
            0 & \text{if } z \ge (\eta+1)\Delta,
        \end{cases}
    \end{align}
    with
    \begin{align}\label{eq:V_lens_lemma}
        V_{\text{lens}}(\Delta, \eta\Delta, z) &= K_N(\Delta, u_c(z)) + K_N(\eta\Delta, z - u_c(z)), \nonumber \\
        u_c(z) &= \frac{z^2 + \Delta^2(1 - \eta^2)}{2z} ,\nonumber \\
        K_N(r, c) &= \frac{\pi^{(N-1)/2} r^N}{\Gamma(\frac{N+1}{2})} \int_{c/r}^{1} (1 - t^2)^{\frac{N-1}{2}} \, dt,
    \end{align}
    and $\Gamma(\cdot)$ is defined in \eqref{def:Gamma}, and $V_N(\cdot)$ is defined in \eqref{eq:Ball_Volume}.
\end{lemma}

The detailed proof of this lemma can be found in Appendix \ref{proof:lemma:prob_acceptance_general}, but here we provide an intuitive geometric overview of the calculation. For any realization of the adversarial noise vector $\mathbf{n}_a$ with a fixed magnitude $\|\mathbf{n}_a\|_2 = z$, the DC accepts the reports if the honest noise $\mathbf{n}_h$ falls within an $N$-ball of radius $\eta\Delta$ centered at $\mathbf{n}_a$. While the conditional probability $\Pr(\mathcal{A}_\eta \mid Z=z)$ formally requires averaging over all possible realizations of $\mathbf{n}_a$ on the shell of radius $z$, the spherical \emph{symmetry of the honest noise} distribution ensures that the intersection volume remains invariant, regardless of the specific direction of $\mathbf{n}_a$. Consequently, this probability is simply the volume of the intersection between the honest noise support (centered at the origin) and the acceptance ball (centered at $\mathbf{n}_a$) divided by the total volume of the honest noise support. This geometry is illustrated in Figure \ref{fig:intersection_geometry}, and in  Appendix \ref{proof:lemma:prob_acceptance_general}, we evaluate the ratio of the two volumes for the different cases of $z$, as presented in \eqref{eq:Phi_piecewise}.

\begin{figure}[htbp]
    \centering
    \begin{tikzpicture}[scale=1.1]
        \def\rHonest{1.8}     
        \def\rAccept{2.8}     
        \def\zDist{2.2}       

        \definecolor{honestColor}{RGB}{100, 149, 237} 
        \definecolor{adversaryColor}{RGB}{255, 99, 71} 

        \coordinate (Origin) at (0,0);      
        \coordinate (Na) at (\zDist,0);    

        \path [name path=circleH] (Origin) circle (\rHonest);
        \path [name path=circleA] (Na) circle (\rAccept);
        \path [name intersections={of=circleH and circleA, by={I1,I2}}];

        \begin{scope}
            \clip (Origin) circle (\rHonest);
            \fill[gray!30] (Na) circle (\rAccept);
        \end{scope}

        \draw[thick, honestColor] (Origin) circle (\rHonest);
        \draw[thick, adversaryColor] (Na) circle (\rAccept);

        \draw[dashed, gray] (-2.2,0) -- (5.2,0);
        \fill (Origin) circle (2pt) node[below left] {$\mathbf{0}$};
        \fill (Na) circle (2pt) node[below right] {$\mathbf{n}_a$};

        \draw[<->, thick] (0,-0.4) -- (\zDist,-0.4) node[midway, below] {$z = \|\mathbf{n}_a\|_2$};

        \draw[->, thick, honestColor] (Origin) -- (130:\rHonest) node[midway, left, text=black] {$\Delta$};
        \draw[->, thick, adversaryColor] (Na) -- ($(Na)+(50:\rAccept)$) node[midway, right, text=black] {$\eta \Delta$};

        \node[honestColor, font=\small] at (0, 2.1) {Honest Support};
        \node[adversaryColor, font=\small] at (\zDist, 3.1) {Acceptance Ball};
        
        \node[text=black!70, font=\footnotesize, align=center] at (0.5, 0.75) {Intersection\\Volume};
        
    \end{tikzpicture}
    \caption{ The honest noise $\mathbf{N}_h$ is uniformly distributed on the blue ball of radius $\Delta$. Given any adversarial noise $\mathbf{n}_a$ with magnitude $z$, the condition $\|\mathbf{N}_h - \mathbf{n}_a\|_2 \le \eta\Delta$ is satisfied if $\mathbf{N}_h$ falls within the red ball. Due to spherical symmetry, the conditional probability $\Pr(\mathcal{A}_\eta \mid Z=z)$ depends only on the scalar distance $z$.}
    \label{fig:intersection_geometry}
\end{figure}
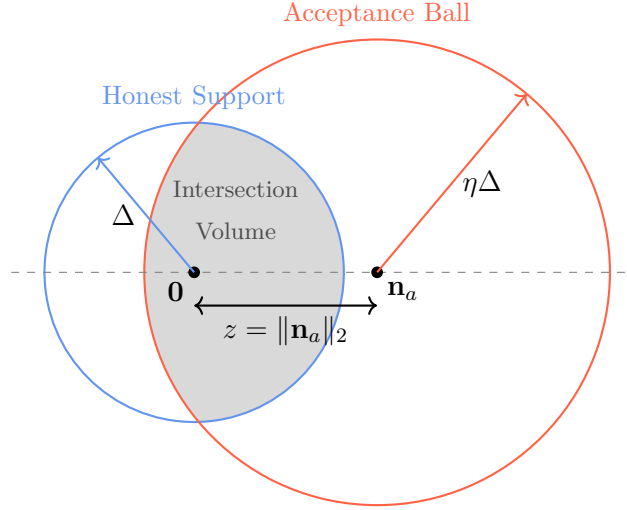
Having characterized the probability of acceptance in Lemma \ref{lemma:prob_acceptance_general}, the next step is to derive a corresponding analytical relationship for the estimation error. Specifically, we seek to express the Mean Squared Error (MSE) in terms of the adversarial noise distribution $g_{\mathbf{N}_a}(\mathbf{n}_a)$, and ideally show that it depends only on $f_Z(z)$. Recall that the DC estimates the ground truth $\mathbf{U}$ by averaging the two reports, $\hat{\mathbf{U}} = \frac{1}{2}(\mathbf{Y}_h + \mathbf{Y}_a)$. The estimation error is therefore the magnitude of the average noise vector. More precisely, we have
\begin{align}
    \|\mathbf{U} - \hat{\mathbf{U}}\|_2^2 = \left\| \mathbf{U} - \left(\mathbf{U} + \frac{\mathbf{N}_h + \mathbf{N}_a}{2}\right) \right\|_2^2 = \left\| \frac{\mathbf{N}_h + \mathbf{N}_a}{2} \right\|_2^2 
\end{align}
The following lemma establishes the relationship between this error and the adversarial noise distribution, and shows that it is fully characterized by its marginal magnitude probability density function $f_Z(z)$. The lemma provides  an analytical framework that characterizes the estimation performance through the density of the magnitude of the adversarial noise.

\begin{lemma}\label{lemma:mse_general}
    For any adversarial noise distribution characterized by the marginal magnitude PDF $f_Z(z)$, the conditional MSE of the estimator is given by
    \begin{align}\label{eq:mse_lemma_result}
        \mathbb{E}\left[ \|\mathbf{U} - \hat{\mathbf{U}}\|_2^2 \mid \mathcal{A}_\eta \right] = \frac{1}{4 \Pr(\mathcal{A}_\eta)} \int_{0}^{\infty} \Psi_N(z) f_Z(z) \, dz.
    \end{align}
    Here, $\Pr(\mathcal{A}_\eta)$ is the acceptance probability that is derived in Lemma \ref{lemma:prob_acceptance_general}, and 
    \begin{align}\label{eq:Psi_piecewise}
        \Psi_N(z) = 
        \begin{cases} 
            z^2 + \frac{N}{N+2}\Delta^2 & \text{if } 0 \le z \le (\eta-1)\Delta, \\[10pt]
            \frac{1}{V_N(\Delta)} \Psi_N^{\text{lens}}(z) & \text{if } (\eta-1)\Delta \leq z \leq (\eta+1)\Delta, \\[10pt]
            0 & \text{if } z \ge (\eta+1)\Delta,
        \end{cases}
    \end{align}
    where $\Psi_N^{\text{lens}}(z)$ is given by
    \begin{align}\label{eq:Psi_lens_def}
        \Psi_N^{\text{lens}}(z) = \Big[ J_N(\Delta, u_c) + z^2 V_1 \Big] + \Big[ J_N(\eta\Delta, z-u_c) + 4z^2 V_2 - 2z Q_N(\eta\Delta, z-u_c) \Big],
    \end{align}
    with 
    \begin{align}
        u_c &= \frac{z^2 + \Delta^2(1 - \eta^2)}{2z} \nonumber \\
        K_N(r, c) &= \frac{\pi^{(N-1)/2} r^N}{\Gamma(\frac{N+1}{2})} \int_{c/r}^{1} (1 - t^2)^{\frac{N-1}{2}} \, dt, \nonumber \\
        V_1 &= K_N(\Delta, u_c), \nonumber \\
        V_2 &= K_N(\eta\Delta, z-u_c), \nonumber \\
        Q_N(r, d) &= \frac{r^2 - d^2}{N+1} V_{N-1}(\sqrt{r^2 - d^2}), \nonumber \\
        J_N(r, d) &= \frac{N r^2}{N+2} K_N(r, d) + \frac{2 d}{N+2} Q_N(r, d),
    \end{align}
    and  $\Gamma(\cdot)$ and $V_N(\cdot)$ are defined in \eqref{def:Gamma} and~\eqref{eq:Ball_Volume}, respectively.
    
\end{lemma}

The proof of this lemma can be found in Appendix \ref{proof:lemma:mse_general}.

Having established the general expressions for the probability of acceptance and the conditional MSE in Lemmas \ref{lemma:prob_acceptance_general} and \ref{lemma:mse_general}, the next step in the proof of Theorem \ref{theorem:Main_Minimax_Result} is to simplify the search space for the worst-case adversarial noise distribution. We show that without loss of optimality, we can restrict the support of the adversarial noise magnitude $Z$ to the interval $[(\eta-1)\Delta, (\eta+1)\Delta]$. We formalize this reduction in the following lemma. 

For a given adversarial noise magnitude distribution $f_Z(z)$ with an arbitrary support, we denote $\Pr(\mathcal{A}_\eta; f_Z)$ as the probability of acceptance when the noise magnitude follows the density $f_Z(z)$. Similarly, $\mathbb{E}\left[ \|\mathbf{U} - \hat{\mathbf{U}}\|_2^2 \mid \mathcal{A}_\eta; f_Z \right]$ denotes the resulting estimation error for the case where the noise magnitude density is $f_Z(z)$.

\begin{lemma}\label{lemma:optimal_support}
    Let $f_Z(z)$ be the probability density function  of the adversarial noise magnitude, satisfying $\Pr(\mathcal{A}_\eta; f_Z) > 0$. There exists an alternative adversarial noise distribution with magnitude probability density function $f^*_Z(z)$, supported strictly on the interval $[(\eta-1)\Delta, (\eta+1)\Delta]$, such that
    \begin{align}
        \Pr(\mathcal{A}_\eta; f^*_Z) &\ge \Pr(\mathcal{A}_\eta; f_Z), \label{eq:better_acc} \\
        \mathbb{E}\left[ \|\mathbf{U} - \hat{\mathbf{U}}\|_2^2 \mid \mathcal{A}_\eta; f^*_Z \right] &\ge \mathbb{E}\left[ \|\mathbf{U} - \hat{\mathbf{U}}\|_2^2 \mid \mathcal{A}_\eta; f_Z \right]. \label{eq:better_mse}
    \end{align}
\end{lemma}

The proof of this lemma is provided in Appendix \ref{proof:lemma:optimal_support}. 

Lemma \ref{lemma:optimal_support} implies that the search for the optimal adversarial noise can be restricted to noises with magnitude PDFs supported on the interval $[(\eta-1)\Delta, (\eta+1)\Delta]$. We refer to the requirement as the \textit{support condition}.

In the following lemma, we further simplify the analysis of the trade-off curve $c_\eta(\alpha)$. Specifically, we show that  the inequality constraint $\PA(g(\cdot), \eta)  \ge \alpha$ in the optimization problem in \eqref{c_small_definition} can be replaced with the equality constraint $\PA(g(\cdot), \eta)  = \alpha$, without affecting the optimal value

\begin{lemma}\label{lemma:exact_acc_noise_existence}
    Fix some $\alpha\in(0,1]$, and let $f_{Z,1}(z)$ be a PDF of the adversarial noise magnitude satisfying the support condition (i.e., supported on $[(\eta-1)\Delta, (\eta+1)\Delta]$), with an acceptance probability $\Pr(\mathcal{A}_\eta; f_{Z,1}) = \alpha_1 > \alpha$. There exists another noise magnitude PDF $f_{Z,2}(z)$, satisfying the support condition, such that  ${\Pr(\mathcal{A}_\eta; f_{Z,2}) = \alpha}$, and 
    \begin{align}
        \mathbb{E}\left[ \|\mathbf{U} - \hat{\mathbf{U}}\|_2^2 \mid \mathcal{A}_\eta; f_{Z,2} \right] = \mathbb{E}\left[ \|\mathbf{U} - \hat{\mathbf{U}}\|_2^2 \mid \mathcal{A}_\eta; f_{Z,1} \right].
    \end{align}
\end{lemma}

The proof of this lemma is provided in Appendix \ref{proof:lemma:exact_acc_noise_existence}.

With Lemmas \ref{lemma:prob_acceptance_general}, \ref{lemma:mse_general}, \ref{lemma:optimal_support}, and \ref{lemma:exact_acc_noise_existence} established, we now proceed to prove the main result of Theorem~\ref{theorem:Main_Minimax_Result}.
Lemma \ref{lemma:optimal_support} restricts the search space to noise distributions supported on $[(\eta-1)\Delta, (\eta+1)\Delta]$, and Lemma \ref{lemma:exact_acc_noise_existence} allows us to fix the acceptance probability constraint to equality.
To prove \eqref{eq:c_eta_result}, we proceed in two steps: first, we establish the upper bound by showing
\begin{align}
    c_{\eta}(\alpha) \leq \frac{\tilde{\Psi}_N^*(\alpha)}{4 \alpha},
\end{align}
and subsequently, we demonstrate that this bound is achievable.

\subsection{Derivation of the Upper Bound}\label{sec:upperbound}

Our goal in this section is to obtain an upper bound for $c_\eta(\alpha)$ defined in~\eqref{c_small_definition}. Recall that, without loss of generality, we may assume that the noise magnitude satisfies the conditions of Lemma \ref{lemma:optimal_support} and satisfies the constraint on the acceptance probability with equality, as shown in Lemma~\ref{lemma:exact_acc_noise_existence}. 
Specifically, based on Lemma \ref{lemma:exact_acc_noise_existence} and the definition of $\Phi_N(z)$ in \eqref{eq:Phi_piecewise}, the probability of acceptance is given by
\begin{align} \label{eq:PA_z_limited}
    \Pr(\mathcal{A}_\eta) = \int_{(\eta-1)\Delta}^{(\eta+1)\Delta} \Phi_N(z) f_Z(z) \, dz = \alpha.
\end{align}
Furthermore, using Lemma \ref{lemma:mse_general} restricted to this support, the conditional MSE is given by
\begin{align} \label{eq:MSE_z_limited}
    \mathbb{E}\left[ \|\mathbf{U} - \hat{\mathbf{U}}\|_2^2 \mid \mathcal{A}_\eta \right] = \frac{1}{4 \alpha} \int_{(\eta-1)\Delta}^{(\eta+1)\Delta} \Psi_N(z) f_Z(z) \, dz.
\end{align}
For a given $z\in [(\eta-1)\Delta, (\eta+1)\Delta]$, let $q=\Phi_N(z)$ to be the conditional probability of acceptance for a given noise magnitude $z$, as defined in~\eqref{definition_I(z)}.
Recall from the definition of $\Phi_N(z)$ in~\eqref{eq:Phi_piecewise}. As shown in Figure~\ref{fig:intersection_geometry}, the function $\Phi_N(z)$ is proportional to the intersection volume of two balls, with centers at distance  $z$. This shows that $\Phi_N(z)$ is a decreasing function, mapping $[0,\infty)$ to $[0,1]$. Therefore, we can study the inverse function $\Phi^{-1}_N(\cdot)$ and its differential transformation,
\begin{align} \label{eq:z_q_transform}
    z = \Phi_N^{-1}(q), \quad \text{and} \quad dq = \Phi_N'(z) \, dz.
\end{align}
Now, let us define
\begin{align} \label{def:w}
    w(q) \triangleq \frac{-f_Z(\Phi_N^{-1}(q))}{\Phi_N'(\Phi_N^{-1}(q))}.
\end{align}
Since $f_Z(z)$ is a valid PDF satisfying the support condition, we have  $\int_{(\eta-1)\Delta}^{(\eta+1)\Delta} f_Z(z) \, dz = 1$. This implies
\begin{align} \label{eq:constraint_norm_w}
    \int_{0}^{1} w(q) \, dq =
    \int_{0}^1 \frac{-f_Z(\Phi_N^{-1}(q))}{\Phi_N'(\Phi_N^{-1}(q))} \, dq \overset{(a)}{=}  
    \int_{(\eta+1)\Delta}^{(\eta-1)\Delta} 
    \frac{-f_Z(z)}{\Phi_N'(z)} \, \Phi'_N(z)dz = \int_{(\eta-1)\Delta}^{(\eta+1)\Delta} 
    f_Z(z) \, dz =  1,
\end{align}
where (a) follows from the change of variable $z=\Phi^{-1}_N(q)$ and \eqref{eq:z_q_transform}.
Similarly, applying the same change of variable $z=\Phi^{-1}_N(q)$ in \eqref{eq:PA_z_limited} leads to
\begin{align} \label{eq:constraint_alpha_w}
    \int_{0}^{1} q w(q) \, dq = \alpha.
\end{align}
Finally, we apply the same change of variable to~\eqref{eq:MSE_z_limited}, and arrive at 
\begin{align} \label{eq:MSE_q_domain}
    \mathbb{E}\left[ \|\mathbf{U} - \hat{\mathbf{U}}\|_2^2 \mid \mathcal{A}_\eta \right]  = \frac{1}{4 \alpha}\int_{1}^0 \Psi_N(\phi^{-1}_N(q)) f_Z(\phi^{-1}_N(q)) \frac{dq}{\Phi'_N(\phi^{-1}_N(q))}
    =\frac{1}{4 \alpha} \int_{0}^{1} \tilde{\Psi}_N(q) w(q) \, dq,
\end{align}
where 
\begin{align}\label{eq:def:tilde-Psi}
    \tilde{\Psi}_N(q) \triangleq \Psi_N(\Phi_N^{-1}(q)).
\end{align}

Let $\tilde{\Psi}_N^*(q)$ be the upper concave envelope of the function $\tilde{\Psi}_N(q)$ over the interval $q \in [0, 1]$, i.e., $\tilde{\Psi}_N^*(q)$ is a concave function and satisfies  $\tilde{\Psi}_N(q) \le \tilde{\Psi}_N^*(q)$ for all $q$. Applying Jensen's inequality and treating $w(q)$ as a probability density function (justified by \eqref{eq:constraint_norm_w}), we get
\begin{align}
    \int_{0}^{1} \tilde{\Psi}_N(q) w(q) \, dq &\le \int_{0}^{1} \tilde{\Psi}_N^*(q) w(q) \, dq \nonumber \\
    &\le \tilde{\Psi}_N^* \left( \frac{\int_{0}^{1} q w(q) \, dq}{\int_{0}^{1} w(q) \, dq} \right) \cdot \int_{0}^{1} w(q) \, dq\nonumber\\
    &\overset{(a)}{\le} \tilde{\Psi}_N^* \left( \frac{\alpha}{1} \right) \cdot 1 = \tilde{\Psi}_N^*(\alpha),
\end{align}
where (a) follows from \eqref{eq:constraint_norm_w} and \eqref{eq:constraint_alpha_w}. 
Finally, substituting this bound back into \eqref{eq:MSE_q_domain} yields the upper bound on the worst-case conditional expectation
\begin{align}\label{upper_bound}
    c_\eta(\alpha) = \max_{f_Z} \mathbb{E}\left[ \|\mathbf{U} - \hat{\mathbf{U}}\|_2^2 \mid \mathcal{A}_\eta \right] \leq \frac{\tilde{\Psi}_N^*(\alpha)}{4 \alpha}.
\end{align}

\subsection{Achievability of the Upper Bound}\label{sec:achivable_noise}

In the previous subsection, we established the upper bound on the worst-case error. Specifically, we showed that
\begin{align}\label{eq:upper_bound_recap}
    c_{\eta}(\alpha) \leq \frac{\tilde{\Psi}_N^*(\alpha)}{4 \alpha}.
\end{align}
In order to complete the proof of Theorem \ref{theorem:Main_Minimax_Result}, we need to demonstrate the reverse inequality
\begin{align}\label{eq:lower_bound_req}
    c_{\eta}(\alpha) \geq \frac{\tilde{\Psi}_N^*(\alpha)}{4 \alpha}.
\end{align}
Based on the definition of $c_{\eta}(\alpha)$ in \eqref{c_small_definition}, proving \eqref{eq:lower_bound_req} is equivalent to showing that there exists at least one admissible noise magnitude distribution $f_Z(z)$ that satisfies the following two conditions simultaneously. First, the resulting probability of acceptance must equal the target $\alpha$, that is
\begin{align}\label{eq:target_PA}
    \Pr(\mathcal{A}_\eta) = \int \Phi_N(z) f_Z(z) \, dz = \alpha.
\end{align}
Second, the resulting conditional MSE must equal the upper bound derived in \eqref{eq:upper_bound_recap}
\begin{align}\label{eq:target_MSE}
    \mathbb{E}\left[ \|\mathbf{U} - \hat{\mathbf{U}}\|_2^2 \mid \mathcal{A}_\eta \right] = \frac{1}{4 \Pr(\mathcal{A}_\eta)} \int \Psi_N(z) f_Z(z) \, dz = \frac{\tilde{\Psi}_N^*(\alpha)}{4\alpha}.
\end{align}

\begin{figure}[htbp]
    \centering
    \begin{tikzpicture}[scale=4.5]
        \draw[->] (-0.1,0) -- (1.1,0) node[right] {$q$};
        \draw[->] (0,-0.1) -- (0,0.6); 
        
        \draw[thick, red, domain=0:0.333, samples=50] plot (\x, {2*(\x*\x + 0.111)});
        \draw[thick, red, domain=0.333:1, samples=50] plot (\x, {2*(\x - \x*\x)});
        \node[red, right] at (1, 0.1) {$\tilde{\Psi}_N(q)$};

        \draw[ultra thick, blue, dashed] (0, 0.222) -- (0.333, 0.444);
        \draw[ultra thick, blue, dashed, domain=0.333:1, samples=50] plot (\x, {2*(\x - \x*\x)});
        \node[blue, above right] at (0.8, 0.45) {$\tilde{\Psi}_N^*(q)$};

        \filldraw[black] (0, 0.222) circle (0.4pt) node[left, xshift=-2pt] {A};
        
        \draw[dashed, gray] (0.333, 0.444) -- (0.333, 0) node[below, black] {$q_1$};
        \filldraw[black] (0.333, 0.444) circle (0.4pt) node[above] {B};

        \draw[dashed, gray] (0.7, 0.42) -- (0.7, 0) node[below, black] {$\alpha_1$};
        \filldraw[black] (0.7, 0.42) circle (0.4pt);

        \draw[dashed, gray] (0.15, 0.322) -- (0.15, 0) node[below, black] {$\alpha_2$};
        \filldraw[black] (0.15, 0.322) circle (0.4pt);
    \end{tikzpicture}
    \caption{A sample potential function $\tilde{\Psi}_N(q)$ and its upper concave envelope $\tilde{\Psi}_N^*(q)$. Over the interval $[0, q_1]$, the concave envelope is defined by the linear chord connecting points A and B, while for $q \in [q_1, 1]$, the envelope coincides with the function itself. To achieve the upper bound in \eqref{eq:upper_bound_recap}, for the case of $\alpha_1$, we use a noise distribution uniformly distributed over the surface of an $N$-sphere with radius $z = \Phi_N^{-1}(\alpha_1)$ as derived in \eqref{eq:noise_single_delta}. For the case of $\alpha_2$, we use a mixed strategy where the noise is uniformly distributed over the surface of an $N$-sphere of radius $z_1 = \Phi_N^{-1}(0)$ with probability $\beta_1$, and uniformly over the surface of an $N$-sphere of radius $z_2 = \Phi_N^{-1}(q_1)$ with probability $1-\beta_1$, as derived in \eqref{eq:noise_vector_mixture}.}
    \label{fig:envelope_achievability}
\end{figure}
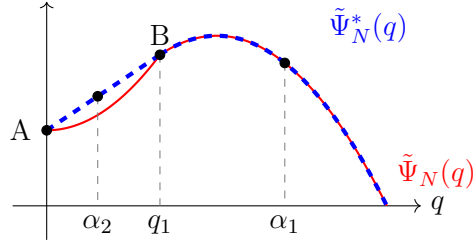

We construct this specific noise distribution by considering the properties of the concave envelope $\tilde{\Psi}_N^*(q)$. Recall that $\tilde{\Psi}_N^*(q)$ is the upper concave envelope of $\tilde{\Psi}_N(q)$ over the interval $q \in [0, 1]$. In the following we distinguish between two cases depending on whether the function $\tilde{\Psi}^*_N(\alpha)=\tilde{\Psi}_N(\alpha)$ (e.g., point $\alpha_1$ in Figure \ref{fig:envelope_achievability}) or $\tilde{\Psi}^*_N(\alpha)$ is on the segment connecting two points on the curve (e.g., point $\alpha_2$ Figure \ref{fig:envelope_achievability}). 
\begin{enumerate}
    \item $\tilde{\Psi}_N^*(\alpha) = \tilde{\Psi}_N(\alpha)$: This implies the function is already on the boundary of its upper concave envelope at $\alpha$ (see $\alpha_1$ in Figure \ref{fig:envelope_achievability}). In this case, the adversary employs a noise magnitude concentrated at a single value $z^* = \Phi_N^{-1}(\alpha)$. In the $N$-dimensional space, this corresponds to an adversarial noise vector $\mathbf{N}_a$ that is uniformly distributed over the surface of the $N$-sphere with radius $z^*$. Mathematically, the probability density function of the vector $\mathbf{N}_a$ is given by
\begin{align}\label{eq:noise_vector_single}
    g_{\mathbf{N}_a}(\mathbf{x}) = \frac{1}{S_N(z^*)} \delta(\|\mathbf{x}\|_2 - z^*),
\end{align}
where $S_N(r) = \frac{2\pi^{N/2}}{\Gamma(N/2)}r^{N-1}$ denotes the surface area of an $N$-sphere of radius $r$. This vector distribution induces the magnitude PDF
\begin{align}\label{eq:noise_single_delta}
    f_Z(z) = \delta(z - z^*).
\end{align}
Substituting this distribution into the acceptance probability integral in \eqref{eq:target_PA}, we obtain
\begin{align}
    \Pr(\mathcal{A}_\eta) = \int \Phi_N(z) \delta(z - z^*) \, dz = \Phi_N(z^*) = \Phi_N(\Phi_N^{-1}(\alpha)) = \alpha.
\end{align}
Thus, the first condition is satisfied. Next, we evaluate the conditional MSE for this distribution. Substituting \eqref{eq:noise_single_delta} into the MSE expression, we get
\begin{align}
    \mathbb{E}\left[ \|\mathbf{U} - \hat{\mathbf{U}}\|_2^2 \mid \mathcal{A}_\eta \right] &= \frac{1}{4 \alpha} \int \Psi_N(z) \delta(z - z^*) \, dz \nonumber \\
    &= \frac{1}{4 \alpha} \Psi_N(z^*) \nonumber \\
    &= \frac{1}{4 \alpha} \Psi_N(\Phi_N^{-1}(\alpha)).
\end{align}
Using the definition $\tilde{\Psi}_N(\alpha) = \Psi_N(\Phi_N^{-1}(\alpha))$ and the assumption $\tilde{\Psi}_N(\alpha) = \tilde{\Psi}_N^*(\alpha)$, we arrive at
\begin{align}
    \mathbb{E}\left[ \|\mathbf{U} - \hat{\mathbf{U}}\|_2^2 \mid \mathcal{A}_\eta \right] = \frac{\tilde{\Psi}_N^*(\alpha)}{4 \alpha}.
\end{align}
This confirms that the single-point distribution defined in \eqref{eq:noise_single_delta} achieves the upper bound when the function touches its envelope.
\item $\tilde{\Psi}_N^*(\alpha) > \tilde{\Psi}_N(\alpha)$: In this case (see $\alpha_2$ in Figure \ref{fig:envelope_achievability}), since $\tilde{\Psi}_N^*(q)$ is the upper concave envelope, the point $(\alpha, \tilde{\Psi}_N^*(\alpha))$ lies on a linear chord connecting two points on the original curve $\tilde{\Psi}_N(q)$. More precisely, there exist $q_1$ and $q_2$ with $0 \le q_1 < \alpha < q_2 \le 1$, such that 
\begin{align}\label{eq:envelope_boundaries}
    \tilde{\Psi}_N^*(q_1) = \tilde{\Psi}_N(q_1), \quad \text{and} \quad \tilde{\Psi}_N^*(q_2) = \tilde{\Psi}_N(q_2).
\end{align}
Furthermore, the point $(\alpha,\tilde{\Psi}_N^*(\alpha))$ lies on the chord connecting $(q_1,\tilde{\Psi}_N^*(q_2))$ and $(q_1,\tilde{\Psi}_N^*(q_2))$, that is,
\begin{align}\label{eq:linear_segment}
   \tilde{\Psi}_N^*(\alpha) = \frac{q_2-\alpha} {q_2 - q_1}\tilde{\Psi}_N(q_1)+ 
   \frac{\alpha-q_1} {q_2 - q_1}\tilde{\Psi}_N(q_2)
   .
\end{align}
 Let $z_1 = \Phi_N^{-1}(q_1)$ and $z_2 = \Phi_N^{-1}(q_2)$. Moreover, we define weights $\beta_1 = \frac{q_2 - \alpha}{q_2 - q_1}$ and $\beta_2 = \frac{\alpha - q_1}{q_2 - q_1}$. In this case, the adversary employs a mixed strategy: With probability $\beta_1$, it selects a noise vector $\mathbf{N}_a$ uniformly distributed over the surface of an $N$-sphere of radius $z_1$, and with probability $\beta_2$, it chooses the noise uniformly over the surface of an $N$-sphere of radius $z_2$. Mathematically, the probability density function of the adversarial noise vector is
\begin{align}\label{eq:noise_vector_mixture}
    g_{\mathbf{N}_a}(\mathbf{x}) = \beta_1 \frac{1}{S_N(z_1)} \delta(\|\mathbf{x}\|_2 - z_1) + \beta_2 \frac{1}{S_N(z_2)} \delta(\|\mathbf{x}\|_2 - z_2),
\end{align}
where $S_N(r)$ is the surface area of the $N$-sphere of radius $r$. This vector distribution induces the following magnitude PDF
\begin{align}\label{eq:noise_two_deltas}
    f_Z(z) = \beta_1 \delta(z - z_1) + \beta_2 \delta(z - z_2).
\end{align}
Note that by construction $\beta_1 + \beta_2 = 1$ and $\beta_1 q_1 + \beta_2 q_2 = \alpha$. We first verify the acceptance probability condition \eqref{eq:target_PA} for this distribution, as follows
\begin{align}
    \Pr(\mathcal{A}_\eta) &= \int \Phi_N(z) \left[ \beta_1 \delta(z - z_1) + \beta_2 \delta(z - z_2) \right] \, dz \nonumber \\
    &= \beta_1 \Phi_N(z_1) + \beta_2 \Phi_N(z_2) \nonumber \\
    &= \beta_1 q_1 + \beta_2 q_2 \nonumber \\
    &= \alpha.
\end{align}
Thus, the distribution yields the required acceptance probability. Finally, we evaluate the conditional MSE. Substituting \eqref{eq:noise_two_deltas} into the MSE integral, we have
\begin{align}
    \mathbb{E}\left[ \|\mathbf{U} - \hat{\mathbf{U}}\|_2^2 \mid \mathcal{A}_\eta \right] &= \frac{1}{4 \alpha} \int \Psi_N(z) \left[ \beta_1 \delta(z - z_1) + \beta_2 \delta(z - z_2) \right] \, dz \nonumber \\
    &= \frac{1}{4 \alpha} \left( \beta_1 \Psi_N(z_1) + \beta_2 \Psi_N(z_2) \right) \nonumber \\
    &= \frac{1}{4 \alpha} \left( \beta_1 \tilde{\Psi}_N(q_1) + \beta_2 \tilde{\Psi}_N(q_2) \right)\nonumber\\
    &\overset{(a)}{=} \frac{\tilde{\Psi}_N^*(\alpha)}{4 \alpha},
\end{align}
where (a) follows from~\eqref{eq:linear_segment}.
This confirms that the mixture distribution defined in \eqref{eq:noise_two_deltas} also achieves the upper bound.
\end{enumerate}
 Since we have constructed a valid noise distribution $f_Z(z)$ for any $\alpha \in [0,1]$ that achieves the bound, the proof of Theorem \ref{theorem:Main_Minimax_Result} is complete.

\section{Illustrative Examples}\label{sec:Illustrative_Examples}

In this section, we present clarifying examples for Theorems \ref{theorem: equivalence_two_problem} and \ref{theorem:Main_Minimax_Result} to demonstrate how these results can be utilized to derive the equilibrium, defined in \eqref{eq:DC_optim}, in various settings. For all the following examples, we assume that $\Delta = 1$. This implies that for any dimension $N$, the noise of the honest node is uniformly distributed within an $N$-dimensional ball of radius $1$, denoted as $\mathcal{B}_N( 1)$.

\begin{example}\label{example1}
    Consider a 2-dimensional system ($N=2$). We assume the utility functions for the adversary and the DC are given by
    \begin{align}
        \mathsf{U}_{\text{AD}}(g(\cdot), \eta) &= \log\left( \mathsf{MSE}(g(\cdot), \eta) \right) + 0.85 \log\left( \mathsf{PA}(g(\cdot), \eta) \right), \label{eq:ex1_adv_util} \\
        \mathsf{U}_{\text{DC}}(g(\cdot), \eta) &= -\mathsf{MSE}(g(\cdot), \eta) + 25 \mathsf{PA}(g(\cdot), \eta). \label{eq:ex1_dc_util}
    \end{align}

    To determine the equilibrium, we first analyze the game from the DC's perspective. For a discrete set of thresholds $\eta \in \{2.0, 2.2, \dots, 8.0\}$, we derive the system's characteristic functions $c_{\eta}(\alpha)$, defined in \eqref{c_small_definition}, which represent the maximum MSE the adversary can strictly enforce for a given acceptance probability $\alpha$. These curves are computed using Theorem \ref{theorem:Main_Minimax_Result} (specifically using the closed-form evaluations for $N=2$ provided in Appendix \ref{sec:case_N2}).
    
    The resulting curves are illustrated in Figure \ref{fig:adversary_metrics_plot}. The curves range from the lowest blue curve, corresponding to the strictest threshold $\eta=2$, to the uppermost red curve, corresponding to the loosest threshold $\eta=8$. As expected, increasing $\eta$ expands the adversary's feasible region, allowing for higher MSE at any given acceptance probability. Note that derivation of these curves is independent of the utility functions in~\eqref{eq:ex1_adv_util} and~\eqref{eq:ex1_dc_util}.  

    \begin{figure}[t]
        \centering
        \includegraphics[width=0.80\linewidth]{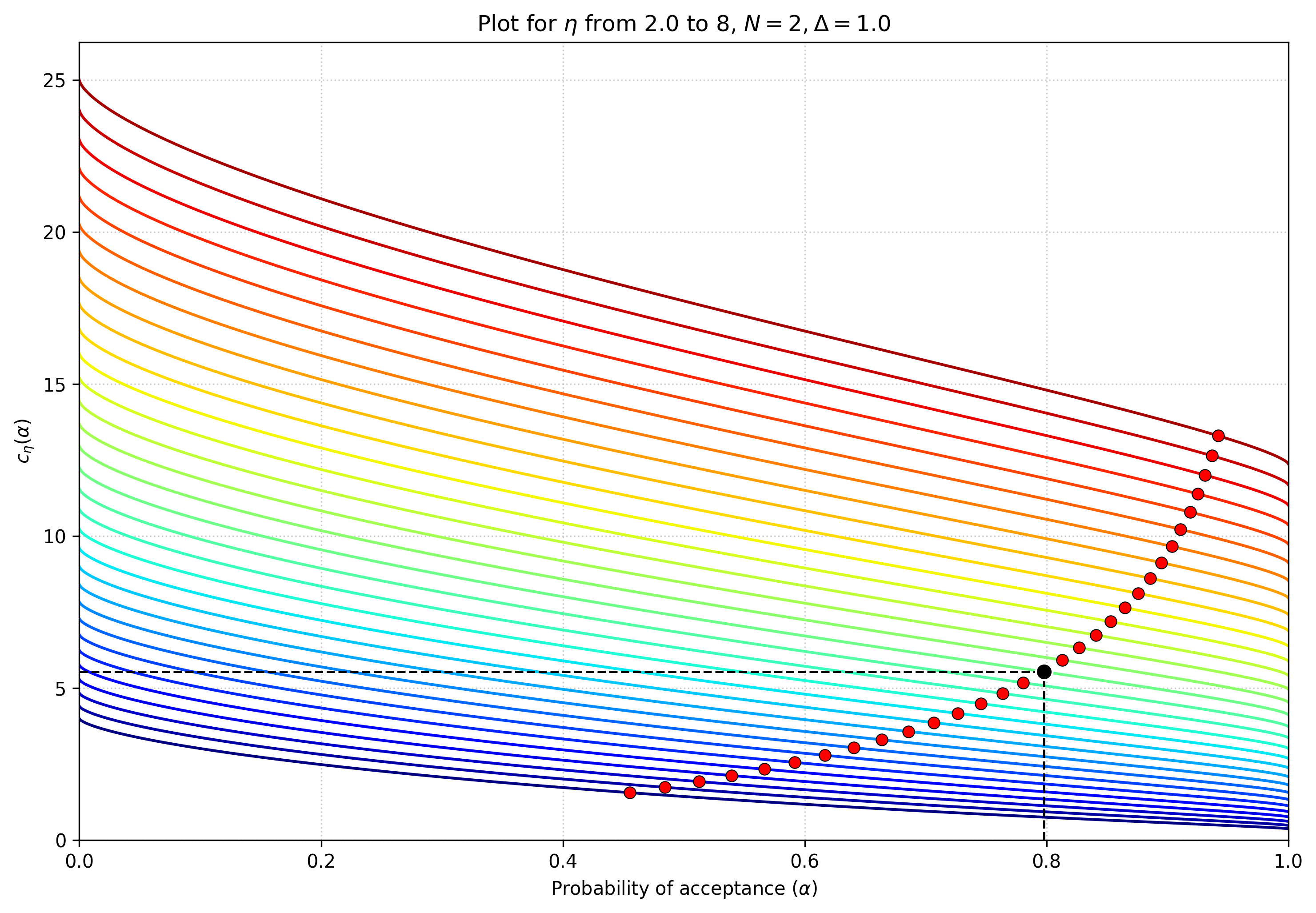}
        \caption{Characteristic curves $c_{\eta}(\alpha)$ defined in \eqref{c_small_definition}, for $N=2$ and $\Delta=1$ (Example \ref{example1}). Each curve corresponds to a specific threshold $\eta \in [2, 8]$, mapping the acceptance probability $\alpha$ (x-axis) to the maximum enforceable MSE (y-axis). The dots (red and black) on each curve represent the adversary's best response operating point $( \PA^*, \MSE^* )$ for that specific $\eta$, and the utility function of the adversary defined in \eqref{eq:ex1_adv_util}. The black dot highlights the global Stackelberg equilibrium of the game where the DC's utility, defined in \eqref{eq:ex1_dc_util}, is maximized.}
        \label{fig:adversary_metrics_plot}
    \end{figure}

    For any specific $\eta$ committed to by the DC, the rational adversary selects the noise distribution $g_{\eta}(\cdot)$ via Algorithm \ref{Alg:finding_noise_N_dim} that maximizes their utility defined in \eqref{eq:ex1_adv_util}. Geometrically, this corresponds to finding the point on the curve $c_{\eta}(\alpha)$ that maximizes the scalar function $Q_{\text{AD}}$. 
    For instance, if DC selects $\eta=2$ (the most strict case), among all the pairs $(\text{PA},\text{MSE})$ on the lowermost curve in Figure~\ref{fig:adversary_metrics_plot}, the rational AD will pick the red dot, that maximizes the AD's utility function: 
        \begin{align*}
            \eta =2: \qquad &\text{MSE} \approx 1.5622, \quad \text{PA} \approx 0.4555, \\
             &\mathsf{U}_{\text{AD}} \approx -\log(1.5622) + 0.85 \log (0.4555) \approx -0.2224,
        \end{align*}
        On the other hand, if the DC selects $\eta=8$ (the loosest in the range of interest), the set of feasible pairs of $(\text{PA},\text{MSE})$ are characterized by the uppermost curve in Figure~\ref{fig:adversary_metrics_plot}. among all these points, the rational AD will chooses the red point to maximize their utility function, that is, 
        \begin{align*}
        \eta =8: \qquad &
            \text{MSE} \approx 13.2991, \quad \text{PA} \approx 0.9419, \\
            &\quad\mathsf{U}_{\text{AD}}
            \approx -\log(13.2991) + 0.85 \log (0.9419) 
            \approx 2.5369.
        \end{align*}
Using a similar approach, the operating point of the AD can be identified for each $\eta$. More precisely,  adversary solves the optimization problem
    \begin{align}
        \alpha^*(\eta) = \underset{0 < \alpha \leq 1}{\arg\max} ~ \left\{ \log\left( c_{\eta}(\alpha) \right) + 0.85 \log(\alpha) \right\},
    \end{align}
    to find their optimum $(\text{PA},\text{MSE}) = \left(\alpha^*(\eta), c_\eta(\alpha^*(\eta))\right)$ for each $\eta$. 
     These optimal operating points are depicted as solid dots on the curves in Figure \ref{fig:adversary_metrics_plot}.
    It is evident that as the DC commits to a larger $\eta$, the adversary exploits the loosened constraint to achieve both higher liveness (PA) and higher error (MSE), strictly increasing their own utility. 
    
    It is worth noting that the similar evaluations can be also performed by the DC, and as a consequence, the DC knows the optimum choice of $\alpha^*(\eta)$ for each $\eta$. In other words, the DC knows the set of operating points depicted by solid dots in Figure \ref{fig:adversary_metrics_plot}, and has an opportunity to choose the best one, by tuning its policy parameter, $\eta$. Recall that the DC  must find the ``sweet spot'' that balances the penalty of error against the reward of liveness as governed by its utility function in~\eqref{eq:ex1_dc_util}. In particular, for the two extreme points studied above, we have  
    \begin{itemize}
        \item At $\eta = 2$: $\mathsf{U}_{\text{DC}} \approx -1.5622 + 25(0.4555) \approx \mathbf{9.8242}$.
        \item At $\eta = 8$: $\mathsf{U}_{\text{DC}} \approx -13.2991 + 25(0.9419) \approx \mathbf{10.2495}$.
    \end{itemize}
    This means the DC prefers the loose threshold $\eta=8$ over the strict $\eta=2$, as the gain in liveness outweighs the cost of increased error in~\eqref{eq:ex1_dc_util}. However, neither is optimal. To find the Stackelberg equilibrium, the DC solves
    \begin{align}
        \eta^* = \underset{\eta \in [2, 8]}{\arg\max} ~ \left\{ -c_{\eta}(\alpha^*(\eta)) + 25 \alpha^*(\eta) \right\}.
    \end{align}
    Solving this optimization reveals that the optimal strategy is an intermediate value, depicted by the sold black dot in Figure~\ref{fig:adversary_metrics_plot}. The equilibrium is achieved at:
    \begin{align*}
        \textbf{Equilibrium } (\eta^* = 5.0): \quad 
        \begin{cases}
            \text{MSE}^* \approx 5.5401 \\
            \text{PA}^* \approx 0.7978 \\
            \mathsf{U}_{\text{AD}}^* \approx 1.5200 \\
            \mathsf{U}_{\text{DC}}^* \approx \mathbf{14.4049}
        \end{cases}
    \end{align*}
    Furthermore, applying Algorithm \ref{Alg:finding_noise_N_dim} reveals that the optimal noise distribution at this equilibrium is a single shell (since the solution lies on the curve rather than a chord). The optimal radius is calculated as $z^* \approx 4.4857$. Consequently, the adversary's best strategy is to add noise uniformly distributed on the circumference of a circle with radius $z^*$:
    \begin{align}
        g_{\mathbf{N}_a}^*(\mathbf{x}) = \frac{1}{2\pi z^*} \delta(\|\mathbf{x}\|_2 - z^*).
    \end{align}
    
    \textbf{Interpretation:} Without the game of coding framework, a naive system designer might default to $\eta=2$. Since the distance between two honest nodes is at most $2\Delta$, setting $\eta=2$ seems logical to reject any obvious attacks. However, our analysis shows this is suboptimal ($\mathsf{U}_{\text{DC}} \approx 9.8$ vs. $\mathsf{U}_{\text{DC}}^* \approx 14.4$). At $\eta=2$, the adversary is forced to attack aggressively to gain utility, resulting in a low probability of acceptance ($\approx 45\%$) which harms the system's liveness. 
    
    By strategically relaxing the threshold to $\eta^*=5$, the DC effectively \emph{bribes} the adversary. The rational adversary, seeking to maximize their own utility (which includes $\log \text{PA}$), shifts their strategy to a noise distribution that is accepted much more frequently ($\approx 80\%$). Although this allows for a higher MSE ($5.54$ vs. $1.56$), the substantial gain in system reliability and liveness leads to a strictly superior outcome for the DC.
\end{example}

\begin{example}\label{example2}
    Consider a high-dimensional system with $N=25$. We assume the utility function for the adversary is given by
    \begin{align}
        \mathsf{U}_{\text{AD}}(g(\cdot), \eta) &= \log\left( \mathsf{MSE}(g(\cdot), \eta) \right) + 0.20 \log\left( \mathsf{PA}(g(\cdot), \eta) \right). \label{eq:ex2_adv_util}
    \end{align}
    For the DC, we analyze the equilibrium under two distinct utility formulations to demonstrate how the choice of metric influences the optimal strategy:
    \begin{align}
        \text{\textbf{Case 1:}} \quad \mathsf{U}_{\text{DC}}^{(1)}(g(\cdot), \eta) &= \frac{\mathsf{PA}(g(\cdot), \eta)}{\sqrt{\mathsf{MSE}(g(\cdot), \eta)}}, \label{eq:ex2_dc_util1} \\
        \text{\textbf{Case 2:}} \quad \mathsf{U}_{\text{DC}}^{(2)}(g(\cdot), \eta) &= \frac{\mathsf{PA}(g(\cdot), \eta)}{\mathsf{MSE}(g(\cdot), \eta)}. \label{eq:ex2_dc_util2}
    \end{align}

    Similar to Example \ref{example1}, we use Theorem \ref{theorem:Main_Minimax_Result} to compute the characteristic curves $c_{\eta}(\alpha)$ for ${\eta \in [2.0, 8.0]}$. The resulting curves are illustrated in Figure \ref{fig:example2_plot}. For each $\eta$, the adversary determines the optimal operating point $(\alpha^*(\eta), c_{\eta}(\alpha^*(\eta)))$ by choosing the noise distribution via Algorithm \ref{Alg:finding_noise_N_dim} and solving
    \begin{align}
        \alpha^*(\eta) = \underset{0 < \alpha \leq 1}{\arg\max} ~ \left\{ \log\left( c_{\eta}(\alpha) \right) + 0.20 \log(\alpha) \right\}.
    \end{align}
     Since the adversary's utility \eqref{eq:ex2_adv_util} remains the same for  both cases, the adversary's response points (marked as red dots in Figure \ref{fig:example2_plot}) are identical for both scenarios. However, the DC's optimal commitment $\eta^*$ changes depending on which utility function is maximized.

    \begin{figure}[t]
        \centering
        \includegraphics[width=0.80\linewidth]{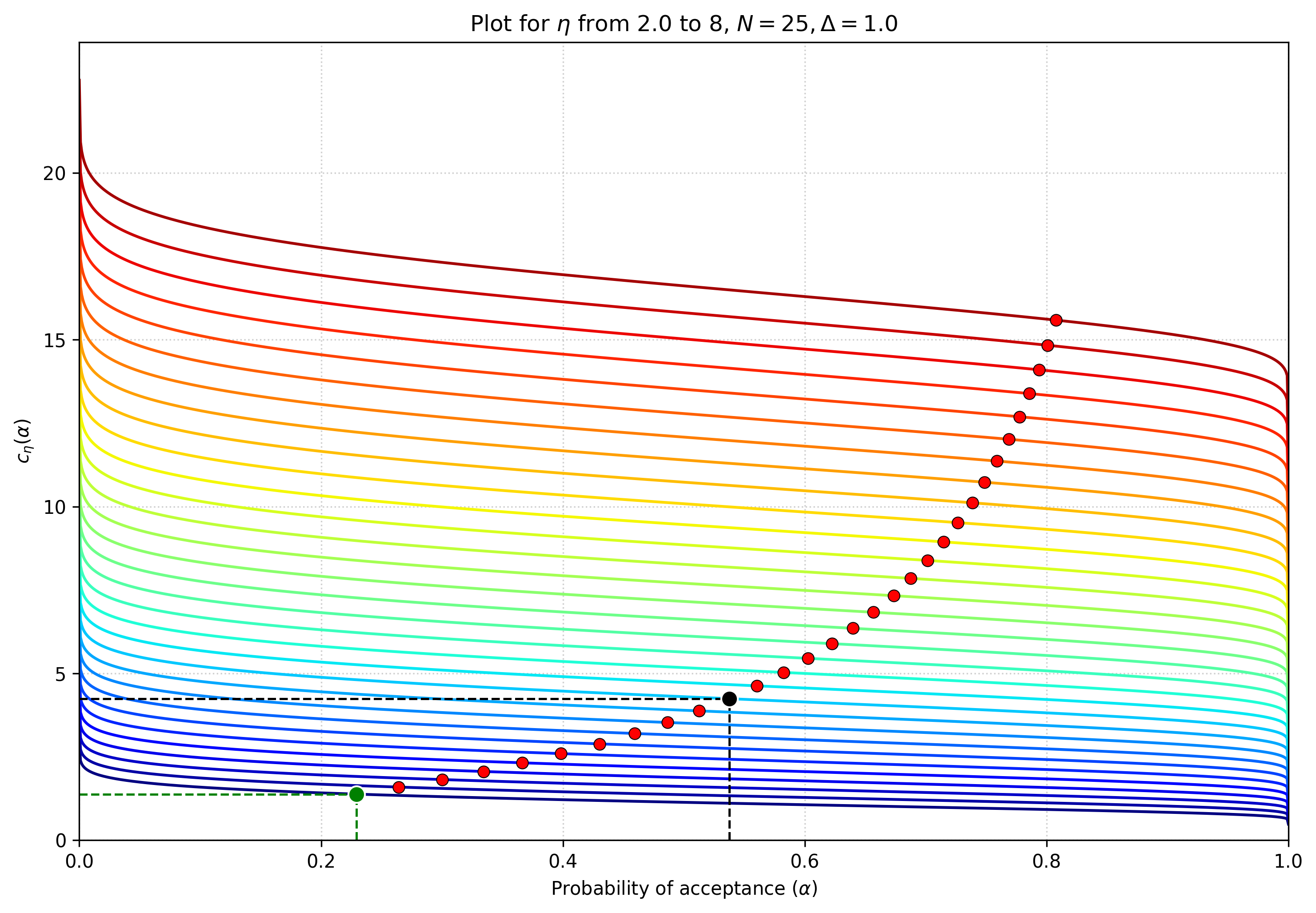}
        \caption{Equilibrium analysis for $N=25$ with $\eta \in [2.0, 8.0]$ (Example \ref{example2}). The curves represent the characteristic functions $c_\eta(\alpha)$. The red dots indicate the adversary's best response for each $\eta$, with respect to the utility function of the adversary, defined in \eqref{eq:ex2_adv_util}. The \textbf{black dot} marks the equilibrium for DC Case 1 \eqref{eq:ex2_dc_util1}, at $\eta^*=4.0$. The \textbf{green dot} marks the equilibrium for DC Case 2 \eqref{eq:ex2_dc_util2}, at $\eta^*=2.0$.}
        \label{fig:example2_plot}
    \end{figure}

    \textbf{Analysis of Case 1:}
    When the DC's utility function is $\mathsf{U}_{\text{DC}}^{(1)}$, it finds the optimum commitment~$\eta^*_1 $ as 
    \begin{align}
        \eta^*_1 = \underset{\eta}{\arg\max} ~ \frac{\alpha^*(\eta)}{\sqrt{c_{\eta}(\alpha^*(\eta))}}.
    \end{align}
    Numerical solving this optimization problem yields an equilibrium at $\eta^* = 4.0$ (indicated by the black dot), given by
    \begin{align*}
        \textbf{Equilibrium 1 } (\eta^* = 4.0): \quad 
        \begin{cases}
            \text{MSE} \approx 4.2409, \\
            \text{PA} \approx 0.5375, \\
            \mathsf{U}_{\text{DC}}^{(1)} \approx \mathbf{0.2610}.
        \end{cases}
    \end{align*}
    Applying Algorithm \ref{Alg:finding_noise_N_dim}, we find that the optimal noise distribution is a single shell with radius $z^* \approx 3.8643$. The adversary's best strategy is to add noise uniformly distributed on the surface of the $N$-dimensional sphere with this radius:
    \begin{align}
        g_{\mathbf{N}_a}^*(\mathbf{x}) = \frac{1}{S_{N}(z^*)} \delta(\|\mathbf{x}\|_2 - z^*),
    \end{align}
    where $S_{N}(r)$ denotes the surface area of the sphere of radius $r$ in $N$ dimensions.

    \textbf{Analysis of Case 2:}
    When the DC optimizes $\mathsf{U}_{\text{DC}}^{(2)}$, the utility function is more sensitive to the error (MSE vs. $\sqrt{\text{MSE}}$). The optimization problem becomes
    \begin{align}
        \eta^*_2 = \underset{\eta}{\arg\max} ~ \frac{\alpha^*(\eta)}{c_{\eta}(\alpha^*(\eta))}.
    \end{align}
    In this case, the equilibrium shifts to the strictest threshold $\eta^* = 2.0$ (indicated by the green dot), and we have
    \begin{align*}
        \textbf{Equilibrium 2 } (\eta^* = 2.0): \quad 
        \begin{cases}
            \text{MSE} \approx 1.3808 \\
            \text{PA} \approx 0.2292 \\
            \mathsf{U}_{\text{DC}}^{(2)} \approx \mathbf{0.1660}
        \end{cases}
    \end{align*}
    Similarly, the optimal noise is a single shell, here with radius $z^* \approx 1.9065$. The noise density is given by:
    \begin{align}
        g_{\mathbf{N}_a}^*(\mathbf{x}) = \frac{1}{S_{N}(z^*)} \delta(\|\mathbf{x}\|_2 - z^*).
    \end{align}

    \textbf{Interpretation:}
    This example highlights how the DC's risk sensitivity dictates the optimal commitment strategy. In Case 1, where the penalty is sublinear with respect to the noise power ($\sqrt{\text{MSE}}$), it is beneficial for the DC to relax the threshold to $\eta=4.0$. This can be seen as ``bribing'' the adversary in order to achieve a significantly higher acceptance rate ($\approx 54\%$ vs $23\%$), which outweighs the cost of the increased error.
    
    Conversely, in Case 2, the penalty is linear with noise power ($\text{MSE}$). Note that the MSE grows rapidly as $\eta$ increases (from $1.38$ at $\eta=2$ to $15.59$ at $\eta=8$), while the acceptance probability is increasing at a much slower pace (from $0.23$ at $\eta=2$ to $0.81$ at $\eta=8$). Therefore, the gain of increasing the  acceptance probability, in the numerator of the utility function, cannot compensate for the explosion in error, in the denominator. Thus, the DC is forced to adopt the strictest policy ($\eta=2.0$) to keep the error bounded, even at the cost of low system liveness.
\end{example}

\begin{example}\label{example3}
    Consider a very high-dimensional system with $N=250$. We assume the utility functions for the adversary and the DC are given by
    \begin{align}
        \mathsf{U}_{\text{AD}}(g(\cdot), \eta) &= \log\left( \mathsf{MSE}(g(\cdot), \eta) \right) + 0.10 \log\left( \mathsf{PA}(g(\cdot), \eta) \right), \label{eq:ex3_adv_util} \\
        \mathsf{U}_{\text{DC}}(g(\cdot), \eta) &= -\log\left( \mathsf{MSE}(g(\cdot), \eta) \right) + 10 \log\left( \mathsf{PA}(g(\cdot), \eta) \right). \label{eq:ex3_dc_util}
    \end{align}

    Following the same methodology described in Examples \ref{example1} and \ref{example2}, we compute the characteristic curves $c_{\eta}(\alpha)$ for $\eta \in [2.0, 8.0]$. For each committed $\eta$, the adversary calculates the best response using Algorithm \ref{Alg:finding_noise_N_dim} that maximizes \eqref{eq:ex3_adv_util}. These optimal operating points are plotted as red dots in Figure \ref{fig:example3_plot}.

    \begin{figure}[t]
        \centering
        \includegraphics[width=0.80\linewidth]{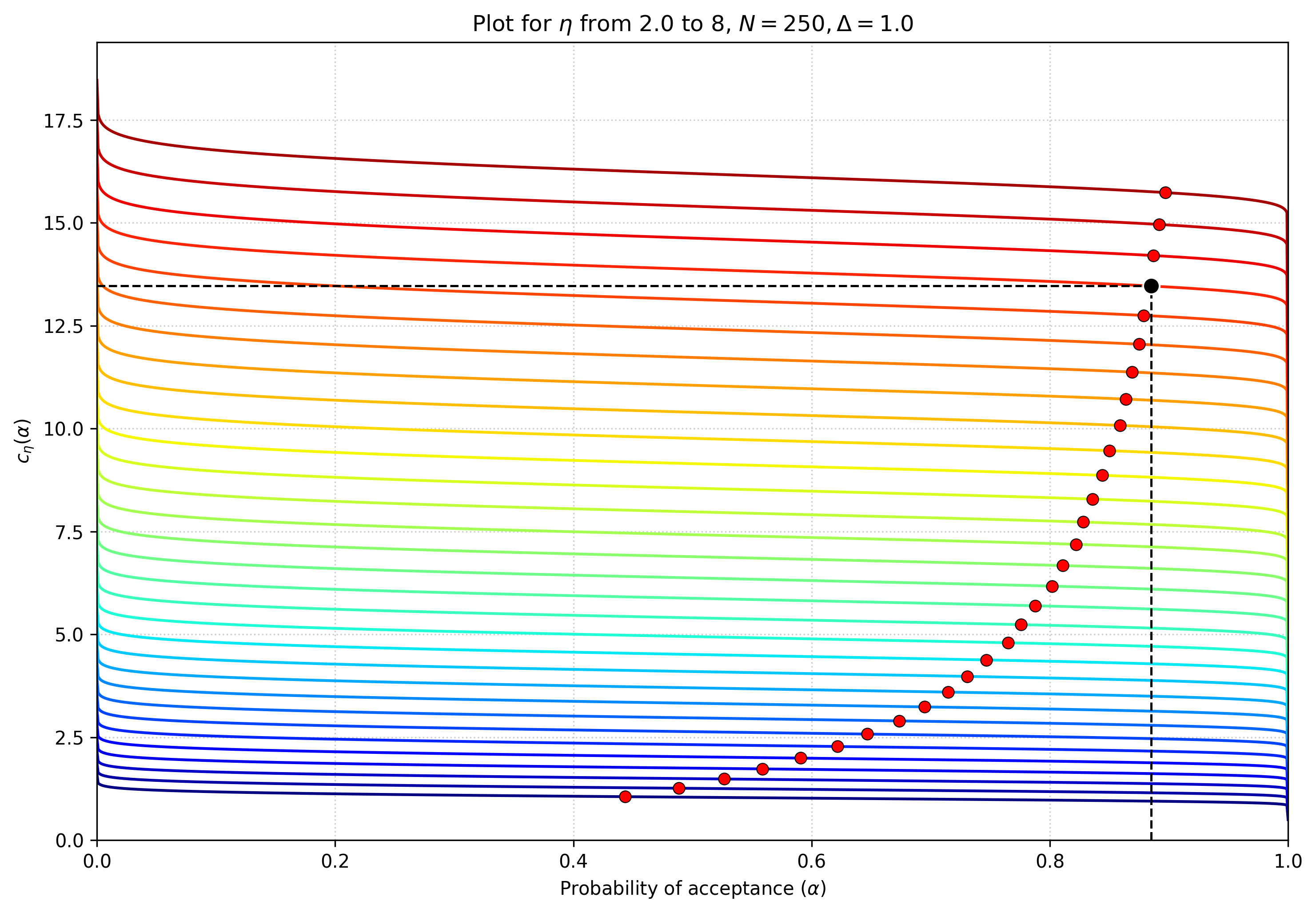}
        \caption{Equilibrium analysis for $N=250$ with $\eta \in [2.0, 8.0]$ (Example \ref{example3}). The curves represent the characteristic functions $c_\eta(\alpha)$. The red dots indicate the adversary's best response for each $\eta$ maximizing \eqref{eq:ex3_adv_util}. The \textbf{black dot} highlights the global Stackelberg equilibrium where the DC's utility \eqref{eq:ex3_dc_util} is maximized.}
        \label{fig:example3_plot}
    \end{figure}

    To determine the equilibrium, the DC evaluates its utility \eqref{eq:ex3_dc_util} across the set of induced operating points. For instance, the DC's utility function at the extreme optimum points evaluates at
    \begin{itemize}
        \item \textbf{Strictest ($\eta = 2.0$):}
        \begin{align*}
            \text{MSE} \approx 1.0567, \quad \text{PA} \approx 0.4434, \quad \mathsf{U}_{\text{DC}} \approx -8.1870.
        \end{align*}
        
        \item \textbf{Loosest ($\eta = 8.0$):}
        \begin{align*}
             \text{MSE} \approx 15.7362, \quad \text{PA} \approx 0.8969, \quad \mathsf{U}_{\text{DC}} \approx -3.8441.
        \end{align*}
    \end{itemize}

    The maximum utility for the DC is achieved at $\eta^* = 7.4$, marked by the black dot in Figure \ref{fig:example3_plot}:
    \begin{align*}
        \textbf{Equilibrium } (\eta^* = 7.4): \quad 
        \begin{cases}
            \text{MSE}^* \approx 13.4655 \\
            \text{PA}^* \approx 0.8849 \\
            \mathsf{U}_{\text{AD}}^* \approx 2.5879 \\
            \mathsf{U}_{\text{DC}}^* \approx \mathbf{-3.8231}
        \end{cases}
    \end{align*}
    Finally, we characterize the optimal noise distribution for this equilibrium. The solution corresponds to a single shell on the characteristic curve with radius $z^* \approx 7.2574$. Thus, the adversary's optimal strategy is to generate noise vectors uniformly from the surface of the 250-dimensional sphere with radius $z^*$:
    \begin{align}
        g_{\mathbf{N}_a}^*(\mathbf{x}) = \frac{1}{S_{N}(z^*)} \delta(\|\mathbf{x}\|_2 - z^*).
    \end{align}
\end{example}
Comparing Figures~\ref{fig:adversary_metrics_plot} ,~\ref{fig:example2_plot}, and  \ref{fig:example3_plot} shows that as the dimension increases ($N\to \infty$), the $c_\eta(\alpha)$ curves tend to become flat. That means, for every given $\eta$, the error ($\MSE$) minimally varies for the entire range of $0\leq \alpha \leq 1$. To understand this, one has to study the kernel functions $\Phi_N(z)$ in~\eqref{eq:Phi_piecewise} and~$\Psi_N(z)$ in~\eqref{eq:Psi_piecewise}, as PA and MSE are the weighted averages of these Kernel functions (see \eqref{eq:prob_acceptance_lemma} and~\eqref{eq:mse_lemma_result}). While $\Phi_N(z)$ captures the intersection volume of the honest noise ball $\mathcal{B}_N(\Delta)$ and the acceptance ball $\mathcal{B}_N(\mathbf{n}_a, \eta\Delta)$ (see Figure~\ref{fig:intersection_geometry}), the function $\Psi_N(z)$ accounts for the average squares norm of points\footnote{To be more accurate, that is indeed, the average squared norm of the average of these points and the adversaries noise. However, the  averaging with the adversarial noise only has a second order affect on the MSE.} in the intersection. It is worth noting that at very high dimension, the majority of the volume of a hypersphere lies close to its shell. Hence, while the size of the intersection (and therefore $\alpha$) may widely vary by changing $z$, the entire mass of the intersection lies close to the shell of the honest ball, and hence, their squared norm highly concentrates at $\Delta^2$. Therefore, MSE becomes much less sensitive to variation of position of the balls. 

Moreover, the vertical gap between the (almost) flat $c_{\eta}(\alpha)$ curves increases quadratically with $\eta$. As mentioned above, the MSE is a weighted mean of the squared norm of average of the points in the intersection of two balls and the adversary's noise. While the points in the intersection have norm of almost $\Delta$, the adversary's noise is about $\eta \Delta$ away from the origin. Hence, the norm of the average grows (almost) linearly with $\eta$, leading to a quadratic growth of the squared norm.

\section{Conclusion and Future Works}\label{sec:conclusion}

In this paper, we have extended the game of coding framework to address vector-valued computations; moving beyond the scalar constraints of prior works \cite{GoCJournal, GoDSybil, nodehi2025unknown, nodehi2026game}. While the previous research established the theoretical viability of the framework, its restriction to scalar values created a distinct gap with practical applications, where vector operations are the norm, particularly in decentralized machine learning. We bridged this gap by providing a rigorous problem formulation for the $N$-dimensional Euclidean space, employing minimal and natural assumptions to ensure practical relevance. Furthermore, we fully characterized the equilibrium of the game, deriving the optimal strategies for both the data collector and the adversary. Through illustrative examples, we demonstrated the dynamics of these strategies in various settings. Crucially, our analysis confirms that the resilience guarantees previously established for scalar settings, specifically the ability to maintain accuracy and liveness, remains valid in the general high-dimensional case. Building on these established foundations, the next step is to extend the game of coding framework in following key directions:
\begin{enumerate}

    \item \textbf{Advanced Coding Techniques:}  The existing works, including this study, relied on repetition coding (assigning the same task to multiple workers). An important directions is  to explore advanced coding techniques to enhance computational efficiency; specifically, this requires deriving new acceptance policies and decoding rules that ensure reliability when using complex codes, such as maximum distance separable (MDS) codes, in an adversarial environment.


    \item \textbf{Unified Learning and Optimization:} The cases where the adversary's utility function is unknown have been studied for the scalar setting in~\cite{nodehi2025unknown}. It is crucial to generalize that result, and  develop a unified framework for learning from non-trustworthy computing agents, that performs distributed training while effectively managing the ambiguity in the adversarial objectives and strategies.
\end{enumerate}

\appendices

\section{Hypersphere: The Second Moment of an \texorpdfstring{$N$}{N}-Ball}
\label{sec:second_moment_n_ball}

In this section, we derive the general formula for the second moment (polar moment of inertia) of an $N$-dimensional ball with uniform density. Let $B(N, r)$ denote an $N$-ball of radius $r$ centered at the origin in $\mathbb{R}^N$, and let $V_N(r)$ denote its volume. We seek to calculate the integral of the squared magnitude of the position vector $\mathbf{u} \in \mathbb{R}^N$ over this volume:
\begin{align}\label{eq:M2_def}
    M_2(N, r) = \int_{B(N, r)} \|\mathbf{u}\|_2^2 \, d\mathbf{u}.
\end{align}

We evaluate this integral using spherical coordinates. We decompose the volume of the $N$-ball into infinitesimal spherical shells of radius $\rho$ (where $0 \le \rho \le r$) and thickness $d\rho$. This decomposition is illustrated in Figure \ref{fig:spherical_shells}.
\begin{figure}[htbp]
    \centering
    \begin{tikzpicture}[scale=1.5]
        \def\rBall{2.0}
        \def\rhoShell{1.2}
        \def\drho{0.2}

        \draw[thick] (0,0) circle (\rBall);
        \node at (45:\rBall+0.2) {$r$};

        \fill[blue!10, even odd rule] (0,0) circle (\rhoShell+\drho) circle (\rhoShell);
        \draw[blue] (0,0) circle (\rhoShell);
        \draw[blue] (0,0) circle (\rhoShell+\drho);

        \draw[->, gray] (-2.5,0) -- (2.5,0);
        \draw[->, gray] (0,-2.5) -- (0,2.5);

        \draw[->, thick] (0,0) -- (135:\rhoShell) node[midway, below left] {$\rho$};
        
        \draw[<->, thick, blue] (135:\rhoShell) -- (135:\rhoShell+\drho) node[midway, above left, blue] {$d\rho$};

        \node[blue!60!black] at (0.8, -0.8) {Spherical Shell};
        \draw[->, blue!60!black] (0.8, -0.6) -- (45:\rhoShell+0.1);

    \end{tikzpicture}
    \caption{Decomposition of the $N$-ball volume into infinitesimal spherical shells. The integral sums the contributions of shells with radius $\rho$ and thickness $d\rho$ from the center to the boundary $r$.}
    \label{fig:spherical_shells}
\end{figure}
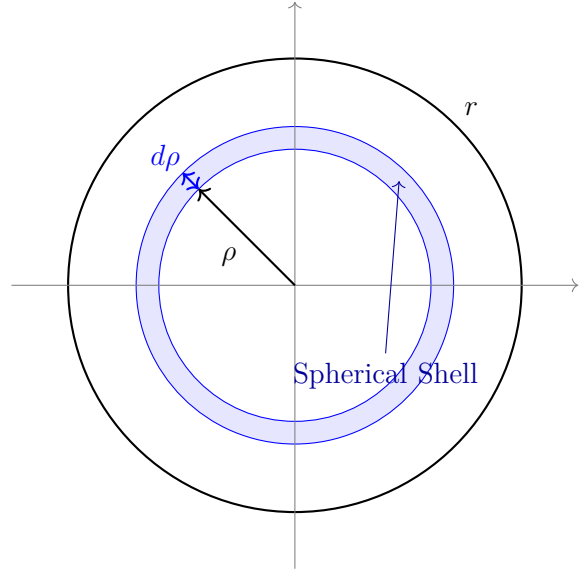
Based on \eqref{eq:ball_general}, we know that the volume of an $N$-ball is given by
\begin{align}\label{eq:volume_N_ball_def}
    V_N(\rho) = C_N \rho^N,
\end{align}
where
\begin{align}
    C_N = \frac{\pi^{N/2}}{\Gamma(\frac{N}{2} + 1)}.
\end{align}
The surface area of the $(N-1)$-sphere (the boundary of the $N$-ball) at radius $\rho$, denoted $A_{N-1}(\rho)$, is the derivative of the volume with respect to the radius. Thus, we have
\begin{align}
    A_{N-1}(\rho) = \frac{d}{d\rho} V_N(\rho) = N C_N \rho^{N-1}.
\end{align}
Consequently, the differential volume element of a shell at radius $\rho$ is
\begin{align}
    d\mathbf{u} = A_{N-1}(\rho) \, d\rho = N C_N \rho^{N-1} \, d\rho.
\end{align}
Since the squared magnitude $\|\mathbf{u}\|_2^2 = \rho^2$ is constant on a spherical shell of radius $\rho$, the integral becomes
\begin{align}
    M_2(N, r) &= \int_{0}^{r} \rho^2 \left( N C_N \rho^{N-1} \, d\rho \right) \nonumber \\
    &= N C_N \int_{0}^{r} \rho^{N+1} \, d\rho.
\end{align}
This implies that
\begin{align}
    M_2(N, r) &= N C_N \left[ \frac{\rho^{N+2}}{N+2} \right]_{0}^{r} \nonumber \\
    &= \frac{N}{N+2} C_N r^{N+2}.
\end{align}
We can rewrite this expression to explicitly include the volume of the ball $V_N(r) = C_N r^N$. More precisely, we have 
\begin{align}
    M_2(N, r) = \frac{N}{N+2} r^2 (C_N r^N) = \frac{N}{N+2} r^2 V_N(r).
\end{align}
Therefore, we arrive at the final result
\begin{align}\label{eq:second_moment_ball_formula}
    \int_{B(N, r)} \|\mathbf{u}\|_2^2 \, d\mathbf{u} = \frac{N}{N+2} r^2 V_N(r).
\end{align}

\section{Hyperspherical Caps}
\label{app:moments_derivation}

In this section, we study hypersperical caps and their geometric properties, along with their moments. A \textbf{hyperspherical cap} is defined as the portion of an $N$-ball cut off by a hyperplane. We call the hyperspherical cap \textbf{canonical} when  the ball is at the origin. More precisely, consider an $N$-ball of radius $r$ centered at the origin, which we denote as the set $\mathcal{B}_N( r) = \{ \mathbf{x} \in \mathbb{R}^N : \|\mathbf{x}\|_2 \le r \}$. If we cut this ball with a hyperplane perpendicular to the $x$-axis at the location $c$, where $-r \le c \le r$, the resulting hyperspherical cap consists of all points in the ball with an $x$-coordinate greater than or equal to $c$. We formally define this region as
\begin{equation} \label{eq:cap_set_def}
    \mathcal{C}_N(r, c) \triangleq \{ (x, x_2, \dots, x_N) \in \mathbb{R}^N : x^2 + x_2^2 + \dots + x_N^2 \le r^2, \, x \ge c \}.
\end{equation}
This geometric concept is illustrated for the 2D case in Figure \ref{fig:integration_slices_QN}.
We derive the volume of a general hyperspherical cap in Section~\ref{sec:cap_derivation}, that is,
\begin{align}\label{def:K_def}
    K_N(r, c) \triangleq \text{Vol}(\mathcal{C}_N(r, c)) = \int_{\mathcal{C}_N(r, c)}  \, d\mathbf{x}.
\end{align}
Next, we characterize the first and second moments of a canonical cap defined as 
\begin{equation}\label{eq:def_Q}
    Q_N(r, c) \triangleq \int_{\mathcal{C}_N(r, c)} x_1 \, d\mathbf{x},
\end{equation}
and
\begin{equation}
\label{eq:definiton_of_J}
    J_N(r, c) \triangleq \int_{\mathcal{C}_N(r, c)} \|\mathbf{x}\|_2^2 \, d\mathbf{x},
\end{equation}
in Sections~\ref{sec:cap-first-moment} and~\ref{sec:cap-second-moment}, respectively. Finally, we consider a general non-canonical cap, which is obtained by cutting an  $N$-ball (with an arbitrary center) by a hyperplane. That is indeed a shifted version of a canonical cap. We study the moments of such a hyperspherical cap in Section~\ref{sec:cap-shifted}.


To evaluate these integrals, we utilize the property that a hyperspherical cap can be viewed as a stack of $(N-1)$-dimensional balls, as shown in Figure~\ref{fig:integration_slices_QN}. More, precisely, let ${\mathbf{x}}_{\sim 1} = (x_2,\dots, x_N)\in \mathbb{R}^{N-1}$. Then, $\mathcal{C}_N(r,c)$ in \eqref{eq:cap_set_def} can be rephrased as $\mathcal{C}_N(r,c) = \{(x_1,{\mathbf{x}_{\sim 1}}): x_1^2 + \|\mathbf{x}_{\sim 1}\|_2^2 \leq r^2, x_1\geq c\}$. For a fixed $x_1 \in [c, r]$, the cross-section of the cap is an $(N-1)$-ball with radius $\sqrt{r^2 - x_1^2}$. The volume of this $(N-1)$-ball is given by $V_{N-1}\left(\sqrt{r^2 - x_1^2}\right)$.
\subsection{Derivation of the Hyperspherical Cap Volume} \label{sec:cap_derivation}

To compute the volume $K_N(r, c)$, we integrate the volumes of its cross-sections along the axis of symmetry $x_1$. Consider a slice of the cap at a position $x_1$ for some $c \le x_1 \le r$ as shown in Figure \ref{fig:integration_slices_QN}. This volume is given by the integral over the set $\mathcal{C}_N(r, c)$ expressed in terms of the first coordinate and the remaining components $\mathbf{x}_{\sim 1}$. That is, based on \eqref{def:K_def}, we have
\begin{equation}\label{first_integral_K_full}
    K_N(r, c) = \int_{\mathcal{C}_N(r, c)}  \, d\mathbf{x} =\int_{c}^{r} \left( \int_{\|\mathbf{x}_{\sim 1}\|_2^2 \le r^2 - x_1^2} d\mathbf{x}_{\sim 1} \right) dx_1.
\end{equation}
The inner integral in \eqref{first_integral_K_full}, represents the volume of an $(N-1)$-ball with radius $\sqrt{r^2 - x_1^2}$. By substituting the $(N-1)$-dimensional volume formula, we write the total volume as
\begin{equation}\label{first_integral_K}
    K_N(r, c) = \int_{c}^{r} V_{N-1}\left(\sqrt{r^2 - x_1^2}\right) \, dx_1.
\end{equation}

Recall from \eqref{eq:Ball_Volume}, that the volume of an $(N-1)$-ball of radius $\rho$ is given by ${V_{N-1}(\rho) = \frac{\pi^{(N-1)/2}}{\Gamma(\frac{N+1}{2})} \rho^{N-1}}$. Substituting $\rho = \sqrt{r^2 - x_1^2}$ into \eqref{first_integral_K} leads to
\begin{equation}\label{eq:K_calculate}
    K_N(r, c) = \frac{\pi^{(N-1)/2}}{\Gamma(\frac{1}{2}(N+1))} \int_{c}^{r} (r^2 - x_1^2)^{\frac{N-1}{2}} \, dx_1.
\end{equation}
By performing the substitution $t = x_1/r$ which implies $dx_1 = r \, dt$, the limits of integration in \eqref{eq:K_calculate} change from $[c, r]$ to $[c/r, 1]$. We thus obtain
\begin{align}\label{eq:integral_simple}
    \int_{c}^{r} (r^2 - x_1^2)^{\frac{N-1}{2}} \, dx_1 &= \int_{c/r}^{1} (r^2 - r^2 t^2)^{\frac{N-1}{2}} (r \, dt) \nonumber \\
    &= r^N \int_{c/r}^{1} (1 - t^2)^{\frac{N-1}{2}} \, dt.
\end{align}
Substituting \eqref{eq:integral_simple} into \eqref{eq:K_calculate}, the final expression for the volume of the hyperspherical cap is given by
\begin{align}\label{final_K}
    K_N(r, c) = \frac{\pi^{(N-1)/2} r^N}{\Gamma(\frac{N+1}{2})} \int_{c/r}^{1} (1 - t^2)^{\frac{N-1}{2}} \, dt.
\end{align}

\subsection{Calculation of the First Moment \texorpdfstring{$Q_N(r, c)$}{QN(r,c)}}
\label{sec:cap-first-moment}

To evaluate the first moment $Q_N(r, c)$, we employ the method of integration by slices perpendicular to the principal axis $x_1$, as illustrated in Figure \ref{fig:integration_slices_QN}. Consequently, the total first moment is obtained as
\begin{align}\label{eq:first_Q}
    Q_{N}(r,c) &= \int_{\mathcal{C}_{N}(r,c)} x_1 d\mathbf{x} = \int_{c}^r \int_{\|\mathbf{x}_{\sim 1}\|_2^2\leq r^2 - x_1^2} x_1 d\mathbf{x}_{\sim 1} dx_1 \nonumber \\
    &=\int_{c}^r x_1 \int_{\|\mathbf{x}_{\sim 1}\|_2^2\leq r^2 - x_1^2} d\mathbf{x}_{\sim 1} dx_1 \nonumber \\
    &= \int_{c}^r x_1 V_{N-1}\left(\sqrt{r^2 -x_1^2}\right) dx_1.
\end{align}

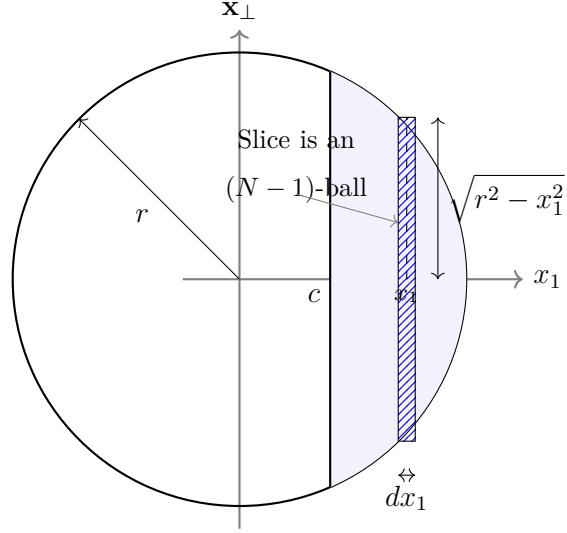
\begin{figure}[htbp]
    \centering
    \begin{tikzpicture}[scale=1.5]
        \def\R{2}
        \def\d{0.8}
        \def\sliceX{1.4}
        \def\sliceW{0.15}
        \def\sliceY{1.428} 

        \draw[->, thick, gray] (-0.5,0) -- (\R+0.5,0) node[right, black] {$x_1$};
        \draw[->, thick, gray] (0,-2.2) -- (0,2.2) node[above, black] {$\mathbf{x}_{\perp}$};

        \draw[thick] (0,0) circle (\R);
        
        \begin{scope}
            \clip (\d,-\R) rectangle (\R, \R);
            \fill[blue!5] (0,0) circle (\R);
        \end{scope}

        \draw[thick] (\d, -1.83) -- (\d, 1.83); 

        \fill[pattern=north east lines, pattern color=blue] (\sliceX, -\sliceY) rectangle (\sliceX+\sliceW, \sliceY);
        \draw (\sliceX, -\sliceY) rectangle (\sliceX+\sliceW, \sliceY);

        \draw[<->] (\sliceX, -\sliceY-0.3) -- (\sliceX+\sliceW, -\sliceY-0.3) node[midway, below] {$dx_1$};
        \draw[dashed] (\sliceX+\sliceW/2, 0) -- (\sliceX+\sliceW/2, \sliceY);
        \node[below, font=\footnotesize] at (\sliceX+\sliceW/2, 0) {$x_1$};

        \draw[<->] (\sliceX+\sliceW+0.2, 0) -- (\sliceX+\sliceW+0.2, \sliceY) node[midway, right] {$ \sqrt{r^2-x_1^2}$};

        \node[below left] at (\d, 0) {$c$};
        
        \draw[->] (0,0) -- (135:\R) node[midway, below left] {$r$};
        
        \node[align=center, font=\small] at (0.5, 1.0) {Slice is an\\$(N-1)$-ball};
        \draw[->, gray] (0.5, 0.75) -- (\sliceX, 0.5);

    \end{tikzpicture}
    \caption{Illustration of a hyperspherical cap (shown in 2D). The cap is decomposed into infinitesimal slices. Each slice at position $x_1$ is an $(N-1)$-dimensional ball of radius $\sqrt{r^2 - x_1^2}$, with volume $V_{N-1}(\sqrt{r^2 - x_1^2}) dx_1$.}
    \label{fig:integration_slices_QN}
\end{figure}

Recall that the volume of an $(N-1)$-ball is given by $V_{N-1}(r) = C_{N-1} r^{N-1}$ where $C_{N-1} = \frac{\pi^{(N-1)/2}}{\Gamma(\frac{N+1}{2})}$. Thus, we can rewrite \eqref{eq:first_Q} as
\begin{equation}
    Q_N(r, c) = C_{N-1} \int_{c}^{r} x_1 (r^2 - x_1^2)^{\frac{N-1}{2}} \, dx_1.
\end{equation}
By employing the substitution $v = r^2 - x_1^2$, which implies $x_1 \, dx_1 = -\frac{1}{2} dv$, and observing that the integration limits transform to $v = r^2 - c^2$ (denoted as the squared intersection height $h^2$) and $v = 0$, we obtain
\begin{equation}
    Q_N(r, c) = \begin{cases}
    \frac{C_{N-1}}{2} \int_{0}^{h^2} v^{\frac{N-1}{2}} \, dv, & c\geq 0\\
    \frac{C_{N-1}}{2}\left[ -\int_{h^2}^{r^2} v^{\frac{N-1}{2}} \, dv + \int_{0}^{r^2} v^{\frac{N-1}{2}} \, dv\right], & c<0
    \end{cases}
    = \frac{C_{N-1}}{2} \int_{0}^{h^2} v^{\frac{N-1}{2}} \, dv.
\end{equation}

Evaluating this integral leads to
\begin{equation}
    Q_N(r, c) = \frac{C_{N-1}}{2} \left[ \frac{2}{N+1} v^{\frac{N+1}{2}} \right]_{0}^{h^2} = \frac{C_{N-1}}{N+1} h^{N+1}.
\end{equation}
Finally, by recognizing that $V_{N-1}(h) = C_{N-1} h^{N-1}$, we arrive at the analytical expression
\begin{equation}
\label{eq:result_Q_n}
    Q_N(r, c) = \frac{(r^2-c^2))}{N+1} V_{N-1}\left(\sqrt{r^2 - c^2}\right).
\end{equation}

\subsection{Calculation of the Second Moment \texorpdfstring{$J_N(r, c)$}{JN(r, c)}}
\label{sec:cap-second-moment}

To calculate the second moment $J_N(r, c)$, defined in \eqref{eq:definiton_of_J}, we employ the same slice-based integration method used for the first moment (see Figure \ref{fig:integration_slices_QN}).
Recall that any vector $\mathbf{x} \in \mathcal{C}_N(r, c)$ can be written as  $\mathbf{x} = (x_1, {\mathbf{x}}_{\sim 1})$, where ${\mathbf{x}}_{\sim 1} = (x_2, \dots, x_N) \in \mathbb{R}^{N-1}$. Consequently, the squared Euclidean norm decomposes as $\|\mathbf{x}\|_2^2 = x_1^2 + \|{\mathbf{x}}_{\sim 1}\|_2^2$. Substituting this decomposition into the volume integral 
in \eqref{eq:definiton_of_J}, we get
\begin{align}\label{eq:nested_j}
    J_N(r, c) =\int_{\mathcal{C}_N(r, c)}  \|{\mathbf{x}}\|_2^2 \, d\mathbf{x} = \int_{c}^{r} \left[ \int_{\|{\mathbf{x}}_{\sim 1}\|_2^2 \leq r^2 - x_1^2} (x_1^2 + \|{\mathbf{x}}_{\sim 1}\|_2^2) \, d{\mathbf{x}}_{\sim 1} \right] \, dx_1.
\end{align}

We now focus on evaluating the inner integral in \eqref{eq:nested_j}. By linearity, we split this  integral into two integrals, one for $x_1^2$ and another one for $\|\mathbf{x}_{\sim 1}\|_2^2$. For the first term, since $x_1$ is constant with respect to ${\mathbf{x}}_{\sim 1}$, we simply have
\begin{align}\label{eq:I_int1}
    \int_{\|{\mathbf{x}}_{\sim 1}\|_2^2 \leq r^2 - x_1^2} x_1^2 \, d{\mathbf{x}}_{\sim 1} = x_1^2 \int_{\|{\mathbf{x}}_{\sim 1}\|_2^2 \leq r^2 - x_1^2} 1 \, d{\mathbf{x}}_{\sim 1} = x_1^2 V_{N-1}\left(\sqrt{r^2 - x_1^2}\right).
\end{align}
For the second term, the integral $\int_{\|{\mathbf{x}}_{\sim 1}\|_2^2 \leq r^2 - x_1^2} \|{\mathbf{x}}_{\sim 1}\|_2^2 \, d{\mathbf{x}}_{\sim 1}$ represents the second moment of the slice about its own center, which is studied in Appendix~\ref{sec:second_moment_n_ball}. Note that the slice is a ball of dimension $k = N-1$ with radius $a = \sqrt{r^2 - x_1^2}$. Hence, applying the general formula derived in \eqref{eq:second_moment_ball_formula}, i.e., ${M_2(k, a) = \frac{k}{k+2} a^2 V_k(a)}$ for $k = N-1$ and $a = \sqrt{r^2 - x_1^2}$, we obtain
\begin{align}\label{eq:I_int2}
    \int_{\|{\mathbf{x}}_{\sim 1}\|_2^2 \leq r^2 - x_1^2} \|{\mathbf{x}}_{\sim 1}\|_2^2 \, d{\mathbf{x}}_{\sim 1} &= \frac{N-1}{(N-1)+2} (r^2 - x_1^2) V_{N-1}\left(\sqrt{r^2 - x_1^2}\right) \nonumber \\
    &= \frac{N-1}{N+1} (r^2 - x_1^2) V_{N-1}\left(\sqrt{r^2 - x_1^2}\right).
\end{align}
Substituting \eqref{eq:I_int1} and \eqref{eq:I_int2} back into \eqref{eq:nested_j}, the integral becomes
\begin{align}
    J_N(r, c) = \int_{c}^{r} \left[ x_1^2 V_{N-1}\left(\sqrt{r^2 - x_1^2}\right) + \frac{N-1}{N+1} (r^2 - x_1^2) V_{N-1}\left(\sqrt{r^2 - x_1^2}\right) \right] \, dx_1.
\end{align}
Simplifying the term in the brackets, the integral splits into two distinct terms
\begin{align}\label{eq:includes_definition_of_I}
    J_N(r, c) = \frac{N-1}{N+1} r^2 \int_{c}^{r} V_{N-1}\left(\sqrt{r^2 - x_1^2}\right) \, dx_1 + \frac{2}{N+1}  
    \int_{c}^{r} x_1^2 V_{N-1}\left(\sqrt{r^2 - x_1^2}\right) \, dx_1.
\end{align}
Let us focus on the first term in~\eqref{eq:includes_definition_of_I}. This integral is exactly the volume of the hyperspherical cap, evaluated in~\eqref{eq:K_calculate}.   Thus, we have
\begin{align}\label{first_term_eq}
    \int_{c}^{r} V_{N-1}\left(\sqrt{r^2 - x_1^2}\right) \, dx_1 = K_N(r, c).
\end{align}

We use integration by parts to evaluate the second term in~\eqref{eq:includes_definition_of_I}, namely,
\begin{align}\label{eq:def-I2}
    I_2 \triangleq \int_{c}^{r} x_1^2 V_{N-1}\left(\sqrt{r^2 - x_1^2}\right) \, dx_1.
\end{align}
Let $u = x_1$, and ${dv = x_1 V_{N-1}\left(\sqrt{r^2 - x_1^2}\right) \, dx_1}$, with 
\begin{align}
    v &= \int dv = \int x_1 \cdot C_{N-1} \left( \sqrt{r^2 - x_1^2} \right)^{N-1} \, dx_1 = \int C_{N-1} x_1 (r^2 - x_1^2)^{\frac{N-1}{2}} \, dx_1\nonumber\\
    &\overset{(a)}{=} C_{N-1} \int  (w)^{\frac{N-1}{2}} \left( -\frac{1}{2} dw \right) = -\frac{C_{N-1}}{2} \int w^{\frac{N-1}{2}} \, dw \nonumber \\
    &= -\frac{C_{N-1}}{2} \left[ \frac{2}{N+1} w^{\frac{N+1}{2}} \right]  = -\frac{C_{N-1}}{N+1} w^{\frac{N+1}{2}}\nonumber\\
    &\overset{(a)}{=} -\frac{C_{N-1}}{N+1} \left(r^2-x_1^2\right)^{\frac{N+1}{2}}  \overset{(b)}{=} -\frac{1}{N+1} (r^2 - x_1^2) V_{N-1}\left(\sqrt{r^2 - x_1^2}\right).
\end{align}
Note that we have used $w = r^2 - x_1^2$ with $dw =-2x_1 dx_1$ in steps indicated by (a), and (b) follows from the definition of the function $V_N(\cdot)$ in~\eqref{eq:Ball_Volume}. Then, we have 
\begin{align}\label{temp_integrand}
    - \int_{c}^{r} v \, du &= \int_{c}^{r} \frac{1}{N+1} (r^2 - x_1^2) V_{N-1}\left(\sqrt{r^2 - x_1^2}\right) \, dx_1 \nonumber \\
    &= \frac{1}{N+1} \left( r^2 \int_{c}^{r} V_{N-1}\left(\sqrt{r^2 - x_1^2}\right) \, dx_1 - \int_{c}^{r} x_1^2 V_{N-1}\left(\sqrt{r^2 - x_1^2}\right) \, dx_1 \right)\nonumber\\
    &\overset{(c)}{=} \frac{1}{N+1} \left(r^2 K_N(r, c) - I_2\right),
\end{align}
where (c) follows from~\eqref{first_term_eq} and~\eqref{eq:def-I2}.

Therefore, the integral in the second term of~\eqref{eq:includes_definition_of_I} can be simplified as
\begin{align}\label{eq:I2-eqn}
    I_2 &= \int_{c}^{r} x_1^2 V_{N-1}\left(\sqrt{r^2 - x_1^2}\right) \, dx_1 = \int_{c}^r u dv = uv \Big|_c^r - \int_{c}^r v du \nonumber\\
    &= \left.-\frac{1}{N+1} x_1(r^2 - x_1^2) V_{N-1}\left(\sqrt{r^2 - x_1^2}\right)\right|_{c}^r + \frac{1}{N+1} \left(r^2 K_N(r, c) - I_2\right) \nonumber\\
    &= \frac{1}{N+1} c(r^2 - c^2) V_{N-1}\left(\sqrt{r^2 - c^2}\right) + \frac{1}{N+1} \left(r^2 K_N(r, c) - I_2\right)\nonumber\\
    &\overset{(d)}{=} c Q_N(r,c) + \frac{1}{N+1} \left(r^2 K_N(r, c) - I_2\right),
\end{align}
where (d) follows from~\eqref{eq:result_Q_n}. Solving \eqref{eq:I2-eqn} for $I_2$, we get
\begin{align}\label{main_Ix^2}
    I_2 = \frac{N+1}{N+2} c Q_N(r, c) + \frac{r^2}{N+2} K_N(r, c).
\end{align}

Now substituting \eqref{first_term_eq}, and \eqref{main_Ix^2} in~\eqref{eq:includes_definition_of_I}, we have
\begin{align}\label{eq:result_J_N}
    J_N(r, c) &= \frac{N-1}{N+1} r^2 K_N(r, c) + \frac{2}{N+1} \left( \frac{N+1}{N+2} c Q_N (r, c)+ \frac{r^2}{N+2} K_N(r, c) \right) \nonumber \\
    &= \frac{N-1}{N+1} r^2 K_N(r, c) + \frac{2c}{N+2} Q_N (r, c)+ \frac{2 r^2}{(N+1)(N+2)} K_N(r, c) \nonumber \\
    & =\frac{N r^2}{N+2} K_N(r, c) + \frac{2 c}{N+2} Q_N(r, c).
\end{align}
\subsection{Moments of a Shifted, Left-Oriented Hyperspherical Cap}
\label{sec:cap-shifted}

 Let $\mathcal{B}'_N( r)$ be an $N$-ball with radius $r$ centered at $\mathbf{z} = (z, 0, \dots, 0)$ on the principal axis. Let the defining hyperplane be located at $x_1 = v$, such that the cap lies to the left of the center (i.e., $v < z$). The region $\mathcal{C}_{\text{left}}$ is defined  as
\begin{align}
    \mathcal{C}_{\text{left}} = \{ \mathbf{x} \in \mathbb{R}^N \mid \|\mathbf{x} - \mathbf{z}\|_2 \le r \text{ and } x_1 \le v \}.
\end{align}
We define the distance parameter $c$ as
\begin{align}
    c = |v - z| = z - v.
\end{align}
Therefore, $\mathcal{C}_{\text{left}}$ defined above is equivalent to 
\begin{align}\label{eq:shifted-cap} \mathcal{C}_N(r,c;z) = \{ \mathbf{x} \in \mathbb{R}^N \mid \|\mathbf{x} - \mathbf{z}\|_2 \le r \text{ and } z-x_1 \ge c \},
\end{align}
which can be seen as a mirrored and shifted version of $\mathcal{C}_N(r,c)$ defined in~\eqref{eq:cap_set_def}. Specifically, the intersection height is $h = \sqrt{r^2 - c^2}$, and the volume is $K_N(r, c)$. The geometry is illustrated in Figure~\ref{fig:shifted_cap}.

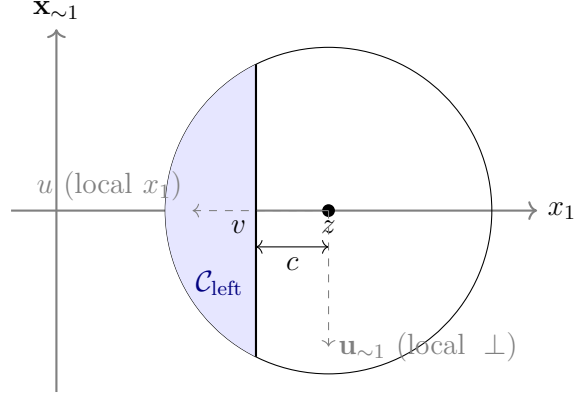
\begin{figure}[htbp]
    \centering
    \begin{tikzpicture}[scale=1.2]
        \def\R{1.8}
        \def\z{3.0}     
        \def\d{0.8}     
        \def\xc{\z-\d}  

        \draw[->, thick, gray] (-0.5,0) -- (\z+\R+0.5,0) node[right, black] {$x_1$};
        \draw[->, thick, gray] (0,-2.0) -- (0,2.0) node[above, black] {${\mathbf{x}}_{\sim 1}$};

        \draw (3,0) circle (\R);

        \begin{scope}
            \clip (\xc-2, -2) rectangle (\xc, 2); 
            \fill[fill=blue!10] (\z,0) circle (\R);
        \end{scope}

        \draw[thick] (\xc, {-\R*sin(acos(\d/\R))}) -- (\xc, {\R*sin(acos(\d/\R))});

        \fill (\z,0) circle (2pt) node[below] {$z$};

        \draw[dashed] (\xc, 0) -- (\xc, {\R*sin(acos(\d/\R))});
        \node[below left] at (\xc, 0) {$v$};
        
        \draw[<->] (\xc, -0.4) -- (\z, -0.4) node[midway, below] {$c$};

        \draw[->, gray, dashed] (\z, 0) -- (\z, -1.5) node[right] {${\mathbf{u}}_{\sim 1} \text{ (local } \perp)$};
        \draw[->, gray, dashed] (\z, 0) -- (\z-1.5, 0) node[above left] {$u \text{ (local } x_1)$};

        \node[blue!50!black] at (\xc-0.4, -0.8) {$\mathcal{C}_{\text{left}}$};
    \end{tikzpicture}
    \caption{Geometry of the shifted, left-oriented hyperspherical cap $\mathcal{C}_{\text{left}}$ (visualized in cross-section). The ball is centered at $\mathbf{z}=(z,0,0,\ldots,0)$. The hyperplane cuts the ball  at $x_1=v$. We define $\mathbf{u}=\mathbf{z}-\mathbf{x}$, which leads to $u_1=z-x_1$ and $\mathbf{u}_{\sim 1} = -\mathbf{x}_{\sim 1}$. 
    }
    \label{fig:shifted_cap}
\end{figure}

Recall that the results derived for a canonical cap in Sections~\ref{sec:cap-first-moment} and~\ref{sec:cap-second-moment} (for $Q_N$ and $J_N$, respectively), are based on the assumption that the is  cap centered at the origin oriented to the right. In order to utilize those results, we reparameterize the shifted cap with left orientation as
\begin{align}\label{eq:shifted-cap:2}
\mathcal{C}_{N}(r,c;z) = \{ (z-u_1, -{\mathbf{u}}_{\sim 1})\in\mathbb{R}^N \mid u_1^2+\|\mathbf{u}_{\sim 1}\|_2 \le r \text{ and } u_1 \le c \} = \mathbf{z} - \mathcal{C}_N(r,c).
\end{align}
Note that the identity above shows that $\mathbf{x} \in \mathcal{C}_N(r,c;z)$ if and only if $\mathbf{x}=\mathbf{z} - \mathbf{u}$ for some $\mathbf{u} \in \mathcal{C}_N(r,c)$. 

Now, we are ready to evaluate the first and second moments for $\mathcal{C}_N(r,x;z)$. We can write
\begin{align}
    Q_N(r,c;z) &= \int_{\mathcal{C}_N(r,c;z)} x_1 d\mathbf{x} \overset{(a)}{=} -\int_{\mathcal{C}_N(r,c)} (z-u_1) (-d\mathbf{u}) \nonumber\\
    &=  z \int_{\mathcal{C}_N(r,c)} d\mathbf{u} - \int_{\mathcal{C}_N(r,c)} u_1 d\mathbf{u}\nonumber\\
    &\overset{(b)}{=} z K_N(r,c) - Q_N(r,c) \label{eq:shifted_Q_n}.
\end{align}
where (a) follows from the change of variable $\mathbf{u} = \mathbf{z}-\mathbf{x}$ as shown in~\eqref{eq:shifted-cap:2}. Note that the negative sign behind the integral in (a) is due to reversing the orientation at which the integral runs over the space $\mathcal{C}_N(r,c)$. Moreover, (b) follows from~\eqref{def:K_def} and~\eqref{eq:def_Q}.

Similarly, for the second moment, we can write

\begin{align}
    J_N(r,c;z) &= \int_{\mathcal{C}_N(r,c;z)} \|\mathbf{x}\|^2 d\mathbf{x} 
    \overset{(a)}{=} -\int_{\mathcal{C}_N(r,c)} \|\mathbf{z}-\mathbf{u}\|^2 (-d\mathbf{u}) \nonumber\\
    &= \|\mathbf{z}\|^2 \int_{\mathcal{C}_N(r,c)}  d\mathbf{u}
    -2 \int_{\mathcal{C}_N(r,c)} \mathbf{z}^\mathsf{T} \mathbf{u}  d\mathbf{u}
    +
    \int_{\mathcal{C}_N(r,c)} \|\mathbf{u}\|^2 d\mathbf{u}\nonumber\\
    &\overset{(b)}{=} z^2 \int_{\mathcal{C}_N(r,c)}  d\mathbf{u}
    -2 z\int_{\mathcal{C}_N(r,c)} u_1  d\mathbf{u}
    +
    \int_{\mathcal{C}_N(r,c)} \|\mathbf{u}\|^2 d\mathbf{u}\nonumber\\
    &\overset{(c)}{=} z^2 K_N(r,c) -2z Q_N(r,c) + J_N(r,c) \label{eq:shifted_J_N}.
\end{align}
Note that, (a) is again due to change of variable $\mathbf{u} = \mathbf{z}-\mathbf{x}$, in (b) we used the fact that ${\mathbf{z}=(z,0,0,\ldots,0)}$, and (c) follows from~\eqref{def:K_def},~\eqref{eq:def_Q}, and~\eqref{eq:definiton_of_J}.

\section{Two Hyperspheres: Volume of the Intersection }\label{sec:intersection_proof}

In this section, we derive the general formula for the intersection volume between two $N$-dimensional balls in $\mathbb{R}^N$. Let the two balls be defined as the sets
\begin{align}
    \mathcal{B}_N( r_1, \mathbf{o}_1) &= \{ \mathbf{x} \in \mathbb{R}^N : \|\mathbf{x} - \mathbf{o}_1\|_2 \le r_1 \}, \\
    \mathcal{B}_N( r_2, \mathbf{o}_2) &= \{ \mathbf{x} \in \mathbb{R}^N : \|\mathbf{x} - \mathbf{o}_2\|_2 \le r_2 \},
\end{align}
where $\mathbf{o}_1$ and $\mathbf{o}_2$ are the center vectors in $\mathbb{R}^N$. 
Because the intersection volume is invariant under rotation and translation, and both balls are spherically symmetric, the intersection volume is a function of only the radii $r_1, r_2$ and the Euclidean distance between the centers $d = \|\mathbf{o}_1 - \mathbf{o}_2\|_2$. Hence, we are interested in
\begin{align}\label{intersection_volume_def}
    V(r_1, r_2, d) \triangleq \text{Vol}\left( \mathcal{B}_N( r_1, \mathbf{o}_1) \cap \mathcal{B}_N( r_2, \mathbf{o}_2) \right).
\end{align}

Recall that the volume of an $N$-ball of radius $r$ is given by
\begin{equation} \label{eq:Ball_Volume_ref}
    V_N(r) = \frac{\pi^{N/2}}{\Gamma(\frac{N}{2} + 1)} r^N,
\end{equation}
where $\Gamma(\cdot)$ is the Euler Gamma function defined in \eqref{def:Gamma}. We consider the three distinct cases below.
\subsection{Case 1: No Intersection}
If the distance between the centers is greater than or equal to the sum of the radii, i.e., $d \ge r_1 + r_2$, the balls are disjoint or touch at a single point. Thus, the intersection volume is 
\begin{equation}\label{cap_case1}
    V(r_1, r_2, d) = 0.
\end{equation}

\subsection{Case 2: Complete Containment}
If the distance is sufficiently small such that one ball is entirely contained within the other, which occurs when $d \le |r_1 - r_2|$, the intersection set is simply the smaller ball. More precisely, we have
\begin{equation}\label{cap_case2}
    V(r_1, r_2, d) = V_N(\min(r_1, r_2)).
\end{equation}

\subsection{Case 3: Partial Overlap}\label{Sec:Partial_Overlap}
Partial overlap occurs when the boundaries of the two $N$-balls intersect, a condition satisfied when $|r_1 - r_2| < d < r_1 + r_2$. In this scenario, the overlap between the balls will be the union of two \emph{hyperspherical caps} 
(see Figure~\ref{fig:integration_slices_QN}). The volume of a general hyperspherical cap depends on the radius of its ball and the distance between the cutting hyper plane and the center of the ball, and is evaluated in Appendix~\ref{sec:cap_derivation}.   However, depending on the configuration of the balls, two sub-cases can be identified. These cases are illustrated in Figure~\ref{fig:opposite_sides} and Figure~\ref{fig:same_side}. In the following, we first characterize the conditions for these two cases, and then formalize the parameters of the cap, and finally use the result of Appendix~\ref{sec:cap_derivation} to compute the volume of the intersection.

We denote the boundary of a set $\mathcal{S}$ by $\partial\mathcal{S}$. In this scenario, the set of all points belonging to the boundaries of both balls, denoted by the intersection $\partial\mathcal{B}_N( r_1, \mathbf{o}_1) \cap \partial\mathcal{B}_N( r_2, \mathbf{o}_2)$, lies entirely within a flat $(N-1)$-dimensional surface known as the \textbf{radical hyperplane} (see Figure~\ref{fig:opposite_sides}).

Formally, the radical hyperplane is defined by the locus of points having equal distance with respect to both spheres. A point $\mathbf{x}$ lies on this hyperplane if and only if
\begin{equation} \label{eq:radical_def}
    \|\mathbf{x} - \mathbf{o}_1\|_2^2 - r_1^2 = \|\mathbf{x} - \mathbf{o}_2\|_2^2 - r_2^2.
\end{equation}
It is worth noting that the condition above can be rephrased as 
\begin{equation} 
\label{eq:radical_def:2}
    2\mathbf{x}^{\mathsf T}(\mathbf{o}_2 - \mathbf{o}_1) = ( \|\mathbf{o}_2\|_2^2 -  \|\mathbf{o}_1\|_2^2) - (r_2^2-r_1^2),
\end{equation}
which is a linear constraint, and clearly characterizes an $(N-1)$-dimensional hyperplane. Moreover, from~\eqref{eq:radical_def:2}, it can be seen that the radical hyperplane is perpendicular to the direction $\mathbf{o}_2 - \mathbf{o}_1$.  Furthermore, for any point $\mathbf{x}$ on the intersection of the boundaries of two balls, we have $\|\mathbf{x}-\mathbf{o}_1\|_2^2 = r_1^2$ and $\|\mathbf{x}-\mathbf{o}_2\|_2^2 = r_2^2$, which make both sides of~\eqref{eq:radical_def} equal zero, and hence lie on the radical hyperplane. 

The geometry of the intersection depends on the position of the radical hyperplane relative to the centers: Sub-cases 3a happens if two centers lie on opposite sides of the radical hyperplane, and Sub-case 3b indicates the both centers are on one side of the radical hyperplane. In order to formally characterize this distinction, without loss of generality, we assume $\mathbf{o}_1$ is at the origin and $\mathbf{o}_2$ is at $(d, 0, \dots, 0)$ on the $x_1$-axis. 
Then, the radical hyperplane is perpendicular to the $x_1$-axis. 
Let $\mathbf{y}$ be the intersection of the radical hyperplane at $x_1$-axis, and assume $c_1= \|\mathbf{o}_1-\mathbf{y}\|_2$ and $c_2=\|\mathbf{o}_2-\mathbf{y}\|_2$ are the geometric distances between the radical hyperplane and the centers $\mathbf{o}_1$ and $\mathbf{o}_2$, respectively. Comparing  Figure~\ref{fig:opposite_sides} and Figure~\ref{fig:same_side}, it turns out that transition from Sub-case 3a to 3b happens right at $c_2=0$, i.e., when $\mathbf{y}= \mathbf{o}_2$.  Plugging $\mathbf{x}=\mathbf{y}= \mathbf{o}_2$ in~\eqref{eq:radical_def}, we get $\|\mathbf{o}_2-\mathbf{o}_1\|_2^2 - r_1^2 = \|\mathbf{o}_2-\mathbf{o}_2\|_2^2 - r_2^2 $, or equivalently, $d^2= r_1^2 - r_2^2$. Then, we can characterize the two Sub-cases as follows.

\subsubsection{Sub-case 3a: Centers on opposite sides of the hyperplane}
When $d^2 \ge r_1^2 - r_2^2$, the radical hyperplane lies between the centers. Our goal is to determine $c_1 $ and $c_2 $. In this configuration,  we have
\begin{equation} \label{eq:dist_sum_case1}
    c_1  + c_2  = d.
\end{equation}
Moreover, plugging $\mathbf{y}$ in~\eqref {eq:radical_def}, we get 
\begin{equation}\label{eq:y-on-hyp}
    c_1 ^2 - r_1^2 =  c_2^2 - r_2^2.
\end{equation}
Solving~\eqref{eq:dist_sum_case1} and~\eqref{eq:y-on-hyp} for $c_1 $ and $c_2$, we arrive at
\begin{align} \label{eq:c_sol_3a}
     c_1^{(1)} = \frac{d^2 + r_1^2 - r_2^2}{2d}, \quad c_2^{(1)} = \frac{d^2 + r_2^2 - r_1^2}{2d}.
\end{align}
Then, the volume of the intersection can be found from 
\begin{align} \label{eq:vol_3a_final}
    V(r_1, r_2, d) = K_N(r_1, c_1^{(1)}) + K_N(r_2, c_2^{(1)}),
\end{align}
where $K_N(r, c)$ is the volume of a hyperspherical cap in an $N$-ball of radius $r$ with a cutting hyperplane at distance $c$ from the center. This volume is evaluated in \eqref{def:K_def}.

\begin{figure}[htbp]
    \centering
    \begin{tikzpicture}[scale=1.1]
        \def\rOne{2.6} \def\rTwo{2.2} \def\dDist{3.5}
        \coordinate (O1) at (0,0); \coordinate (O2) at (\dDist,0);
        \path [name path=c1] (O1) circle (\rOne); \path [name path=c2] (O2) circle (\rTwo);
        \path [name intersections={of=c1 and c2, by={A,B}}];
        \begin{scope} \clip (O1) circle (\rOne); \fill[gray!20] (O2) circle (\rTwo); \end{scope}
        \draw[dashed, red, thick] (A |- 0,-2.5) -- (A |- 0,2.5) node[above] {Radical Plane};
        \draw[thick] (O1) circle (\rOne); \draw[thick] (O2) circle (\rTwo);
        \fill (O1) circle (2pt) node[below left] {$\mathbf{o}_1$}; \fill (O2) circle (2pt) node[below right] {$\mathbf{o}_2$};
        \draw[thick] (O1) -- (O2) node[midway, above] {$d$};
        
        \coordinate (Y) at (A |- O1);
        \fill (Y) circle (2pt) node[below, xshift=0.2cm] {$\mathbf{y}$};

        \draw[<->, blue] (O1 |- 0,-0.6) -- (A |- 0,-0.6) node[midway, below] {$c_1^{(1)}$};
        \draw[<->, blue] (A |- 0,-0.6) -- (O2 |- 0,-0.6) node[midway, below] {$c_2^{(1)}$};
    \end{tikzpicture}
    \caption{Sub-case 3a: Both centers are outside the intersection, resulting in $c_1^{(1)} + c_2^{(1)} = d$.}
    \label{fig:opposite_sides}
\end{figure}
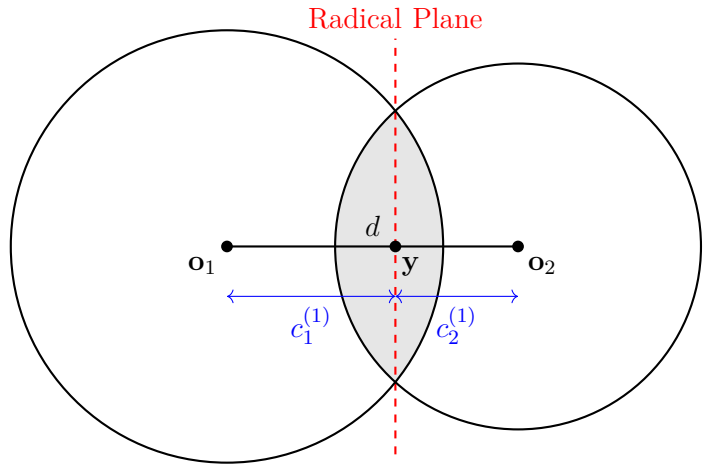
\subsubsection{Sub-case 3b: Centers on the same side of the hyperplane}
When $d^2 < r_1^2 - r_2^2$, the radical hyperplane lies to the right of both centers. Therefore, we have  
\begin{equation} \label{eq:dist_diff_case2}
    c_1  - c_2  = d.
\end{equation}
Solving this equation together with~\eqref{eq:y-on-hyp} for $c_1$ and $c_2$, leads to
\begin{align} \label{eq:c_sol_3b}
     c_1^{(2)} = \frac{d^2 + r_1^2 - r_2^2}{2d}, \quad c_2^{(2)} = \frac{r_1^2 - r_2^2 - d^2}{2d}.
\end{align}
As illustrated in Figure \ref{fig:same_side}, in this case, for the intersection volume, we have 
\begin{align} \label{eq:vol_3b_final}
    V(r_1, r_2, d) &= K_N(r_1, c_1^{(2)}) + (V_N(r_2) - K_N(r_2, c_2^{(2)})) \nonumber \\
    &\overset{(a)}{=} K_N(r_1, c_1^{(2)}) + K_N(r_2, -c_2^{(2)}),
\end{align}
where (a) follows from the fact that based on the definition of a hyperspherical cap in \eqref{eq:cap_set_def} and its volume in \eqref{def:K_def}, for any $c \in [0, r]$, we have $V_N(r) = K_N(r, c) + K_N(r, -c)$. This identity reflects that a hyperplane divides a ball into two caps whose volumes sum to the total volume $V_N(r)$.

\begin{figure}[htbp]
    \centering
    \begin{tikzpicture}[scale=1.1]
        \def\rOne{3.5} \def\rTwo{2.5} \def\dDist{1.8}
        \coordinate (O1) at (0,0); \coordinate (O2) at (\dDist,0);
        \path [name path=c1] (O1) circle (\rOne); \path [name path=c2] (O2) circle (\rTwo);
        \path [name intersections={of=c1 and c2, by={A,B}}];
        \begin{scope} \clip (O1) circle (\rOne); \fill[gray!20] (O2) circle (\rTwo); \end{scope}
        
        \draw[dashed, red, thick] (A |- 0,-3.2) -- (A |- 0,3.2) node[above] {Radical Plane};
        
        \draw[thick] (O1) circle (\rOne); \draw[thick] (O2) circle (\rTwo);
        
        \coordinate (Y) at (A |- O1);
        
        \fill (O1) circle (2pt) node[below left] {$\mathbf{o}_1$}; 
        \fill (O2) circle (2pt) node[below right] {$\mathbf{o}_2$};
        
        \draw[thick] (O1) -- (Y); 
        \path (O1) -- (O2) node[midway, above] {$d$}; 
        \fill (Y) circle (2pt) node[below right] {$\mathbf{y}$}; 
        
        \draw[<->, blue] (O1 |- 0,-0.8) -- (A |- 0,-0.8) node[midway, below] {$c_1^{(2)}$};
        \draw[<->, blue] (O2 |- 0,-1.4) -- (A |- 0,-1.4) node[midway, below] {$c_2^{(2)}$};
    \end{tikzpicture}
    \caption{Sub-case 3b: Center $\mathbf{o}_2$ is inside the intersection, resulting in $c_1^{(2)} - c_2^{(2)} = d$.}
    \label{fig:same_side}
\end{figure}
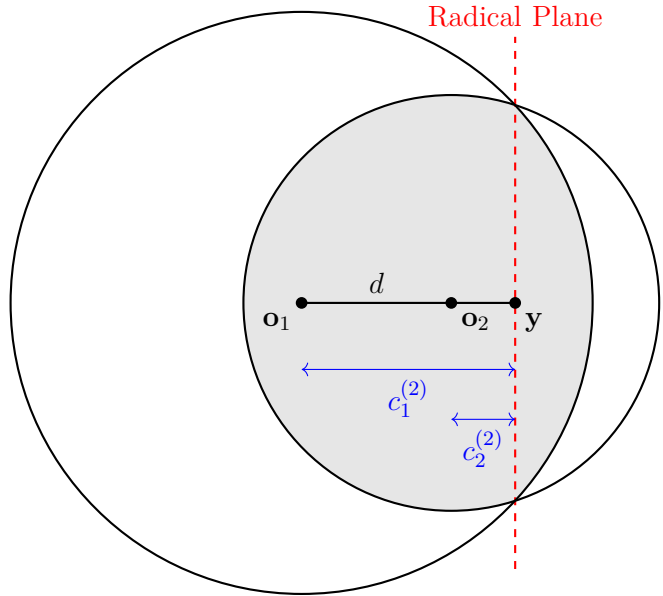
\subsubsection{Aggregation of Cases}
We observe that \eqref{eq:vol_3a_final} and \eqref{eq:vol_3b_final} can be unified into a single expression by allowing $c_2$ to be a signed distance. If we define $c_2 = \frac{d^2 + r_2^2 - r_1^2}{2d}$ as in \eqref{eq:c_sol_3a}, then in Case 3b, $c_2$ becomes naturally negative ($c_2 = -c_2^{(2)}$). Thus, for all configurations of partial overlap, we have
\begin{equation} \label{eq:unified_vol}
    V(r_1, r_2, d) = K_N(r_1, c_1) + K_N(r_2, c_2),
\end{equation}
where
\begin{align}\label{eq:c_1}
    c_1 = \frac{d^2 + r_1^2 - r_2^2}{2d}
\end{align}
 and 
 \begin{align}\label{eq:c_2}
     c_2 = \frac{d^2 + r_2^2 - r_1^2}{2d}.
 \end{align}

\subsection{General Expression for the Intersection Volume}
By aggregating the results from Case 1 \eqref{cap_case1}, Case 2 \eqref{cap_case2}, and Case 3 \eqref{eq:unified_vol}, we obtain a comprehensive expression for the intersection volume of two $N$-balls. The general formula $V(r_1, r_2, d)$ is defined as the following piecewise function
\begin{equation} \label{eq:general_intersection}
V(r_1, r_2, d) = 
\begin{cases} 
    0 & \text{if } d \ge r_1 + r_2 \quad (\text{No Intersection}), \\[10pt]
    V_N(\min(r_1, r_2)) & \text{if } d \le |r_1 - r_2| \quad (\text{Complete Containment}), \\[10pt]
    V_{\text{lens}}(r_1, r_2, d) & \text{if } |r_1 - r_2| < d < r_1 + r_2 \quad (\text{Partial Overlap}),
\end{cases}
\end{equation}
where $V_N(r)$ is the volume of an $N$-ball of radius $r$, as defined in \eqref{eq:Ball_Volume}, and $V_{\text{lens}}$ is defined in~\eqref{eq:unified_vol}--\eqref{eq:c_2}.

\section{Proof of Lemma \ref{lemma:prob_acceptance_general}}\label{proof:lemma:prob_acceptance_general}
To prove Lemma \ref{lemma:prob_acceptance_general}, we apply the law of total probability to express $\Pr(\mathcal{A}_\eta)$ as
\begin{align}\label{eq:total_prob_integral}
    \Pr(\mathcal{A}_\eta) = \int_{0}^{\infty} \Pr(\mathcal{A}_\eta \mid Z=z) f_Z(z) \, dz.
\end{align}
Comparing \eqref{eq:total_prob_integral} with \eqref{eq:prob_acceptance_lemma}, it is sufficient to derive the kernel function 
\begin{align}\label{definition_I(z)}
    \Phi_N(z) \triangleq \Pr(\mathcal{A}_\eta \mid Z=z).
\end{align}

Let $\mathcal{S}_z$ denote the surface of the $N$-ball with radius $z$ centered at the origin. Given a magnitude $Z=z$, the vector $\mathbf{N}_a$ is distributed over this surface with a conditional probability density 
\begin{align*}
f_{\mathbf{N}_a|Z}(\mathbf{n}_a | z)=\begin{cases}
    \frac{g(\mathbf{n}_a)}{f_Z(z)} & \mathbf{n}_a\in \mathcal{S}_z\\
    0 & \textrm{otherwise,}
\end{cases}
\end{align*}
where $g(\cdot)$ is the adversarial noise distribution. Therefore, we can express the conditional acceptance probability as an integral over the surface $\mathcal{S}_z$ as
\begin{align}\label{eq:surface_integral_PA}
    \Pr(\mathcal{A}_\eta \mid Z=z) = \int_{\mathcal{S}_z} \Pr(\mathcal{A}_\eta \mid \mathbf{N}_a = \mathbf{n}_a) f_{\mathbf{N}_a|Z}(\mathbf{n}_a | z) \, d\mathbf{n}_a.
\end{align}
Now, let $\mathcal{R}_{\text{honest}} = \{ \mathbf{x} \in \mathbb{R}^N \mid \|\mathbf{x}\|_2 \le \Delta \}$ be the support of the honest noise and
\begin{align}
    \mathcal{R}_{\text{acc}}(\mathbf{n}_a) = \{ \mathbf{x} \in \mathbb{R}^N \mid \|\mathbf{x} - \mathbf{n}_a\|_2 \le \eta \Delta \},
\end{align}
be the acceptance region of the honest noise, for a fixed adversarial noise vector $\mathbf{n}_a$. Since $\mathbf{N}_h$ is uniformly distributed over $\mathcal{R}_{\text{honest}}$, for any fixed $\mathbf{n}_a$, we have
\begin{align}\label{eq:fixed_na_ratio}
    \Pr(\mathcal{A}_\eta \mid \mathbf{N}_a = \mathbf{n}_a) = \frac{\text{Vol}\left( \mathcal{R}_{\text{honest}} \cap \mathcal{R}_{\text{acc}}(\mathbf{n}_a) \right)}{\text{Vol}(\mathcal{R}_{\text{honest}})}.
\end{align}
Note that, due to the uniform distribution of $\mathbf{N}_h$, this probability only depends on the volume of the intersection. Furthermore, due to the spherical symmetry of the honest support $\mathcal{R}_{\text{honest}}$, 
this volume only depends on the distance between the two centers, which in turn, depends only on the magnitude $\|\mathbf{n}_a\|_2 = z$,  and remains invariant regardless of the direction of $\mathbf{n}_a$. 

As discussed above, even though $\mathcal{R}_{\text{acc}}(\mathbf{n}_a)$ depends on both magnitude and direction of $\mathbf{n}_a$, the quantity of interest, i.e., $\text{Vol}\left( \mathcal{R}_{\text{honest}} \cap \mathcal{R}_{\text{acc}}(\mathbf{n}_a) \right)$ only depends on $z=\|\mathbf{n}_a\|$. Hence, with slightly abuse of notation and for simplicity, we let $\mathcal{R}_{\text{acc}}(z)$ denote the acceptance region for an arbitrary vector $\mathbf{n}_a$ on the shell $\mathcal{S}_z$. Consequently, the ratio in \eqref{eq:fixed_na_ratio} is constant for all $\mathbf{n}_a \in \mathcal{S}_z$. This allows us to move this constant term outside the integral in \eqref{eq:surface_integral_PA}, and arrive at
\begin{align}
    \Pr(\mathcal{A}_\eta \mid Z=z) &= \frac{\text{Vol}\left( \mathcal{R}_{\text{honest}} \cap \mathcal{R}_{\text{acc}}(z) \right)}{\text{Vol}(\mathcal{R}_{\text{honest}})} \int_{\mathcal{S}_z} f_{\mathbf{N}_a|Z}(\mathbf{n}_a | z) \, d\mathbf{n}_a \nonumber \\
    &\overset{(a)}{=} \frac{V(\Delta, \eta\Delta, z)}{V_N(\Delta)} \cdot 1,
\end{align}
where (a) follows from the fact that the conditional PDF $f_{\mathbf{N}_a|Z}(\mathbf{n}_a | z)$ integrates to unity over its support $\mathcal{S}_z$, and $V(\Delta, \eta\Delta, z)$ represents the general intersection volume of two $N$-balls at distance $z$ and radii $\Delta$ and $\eta \Delta$. This volume is formally defined and evaluated in \eqref{intersection_volume_def} of Appendix \ref{sec:intersection_proof}.

Substituting the piecewise characterization of $V(\Delta, \eta\Delta, z)$ based on the geometric cases described in \eqref{eq:general_intersection} into the volume ratio \eqref{eq:fixed_na_ratio} and subsequently into the total probability integral \eqref{eq:total_prob_integral} yields the three cases for $\Phi_N(z)$ specified in \eqref{eq:Phi_piecewise}. This completes the proof of Lemma \ref{lemma:prob_acceptance_general}.

\section{Proof of Lemma \ref{lemma:mse_general}}\label{proof:lemma:mse_general}
To prove Lemma \ref{lemma:mse_general}, recall that
    the conditional expected error is defined as $\mathbb{E}[ \|\mathbf{U} - \hat{\mathbf{U}}\|_2^2 \mid \mathcal{A}_\eta ]$. Using the law of total expectation conditioned on the adversarial noise $\mathbf{N}_a$, we write
    \begin{align}\label{eq:cond-mse-int}
        \mathbb{E}[ \|\mathbf{U} - \hat{\mathbf{U}}\|_2^2 \mid \mathcal{A}_\eta ] 
        = \int_{\mathbf{n}_a} \mathbb{E}[ \|\mathbf{U} - \hat{\mathbf{U}}\|_2^2 \mid \mathcal{A}_\eta, 
        \mathbf{N}_a=\mathbf{n}_a ] g_{\mathbf{N}_a|\mathcal{A}_\eta}(\mathbf{n}_a) d\mathbf{n}_a
        .
    \end{align}
    
    Using Bayes' theorem, the posterior density 
    $g_{\mathbf{N}_a|\mathcal{A}_\eta}(\mathbf{n}_a)$ is
    \begin{align}\label{eq:g-cond-acc}
        g_{\mathbf{N}_a|\mathcal{A}_\eta}(\mathbf{n}_a) = \frac{\Pr(\mathcal{A}_\eta \mid \mathbf{N}_a = \mathbf{n}_a) g_{\mathbf{N}_a}(\mathbf{n}_a)}{\Pr(\mathcal{A}_\eta)} = \frac{\Pr(\mathcal{A}_\eta \mid \mathbf{N}_a = \mathbf{n}_a) \int_{0}^\infty f_{\mathbf{N}_a|Z}(\mathbf{n}_a|z) f_Z(z) dz}{\Pr(\mathcal{A}_\eta)}.
    \end{align}
    Note that, here we have 
\begin{align*}
f_{\mathbf{N}_a|Z}(\mathbf{n}_a | z)=\begin{cases}
    \frac{g(\mathbf{n}_a)}{f_Z(z)} & \mathbf{n}_a\in \mathcal{S}_z\\
    0 & \textrm{otherwise,}
\end{cases}
\end{align*}
and $\mathcal{S}_z \triangleq \{ \mathbf{x} \in \mathbb{R}^N : \|\mathbf{x}\|_2 = z \}$ denotes the surface of the $N$-ball with radius $z$.
    Substituting \eqref{eq:g-cond-acc}
    back into \eqref{eq:cond-mse-int} yields
    \begin{align}\label{eq:proof_integral_sub}
        \mathbb{E}[ &\|\mathbf{U} - \hat{\mathbf{U}}\|_2^2 \mid \mathcal{A}_\eta ] 
        = \int_{\mathbf{n}_a} \mathbb{E}[ \|\mathbf{U} - \hat{\mathbf{U}}\|_2^2 \mid \mathcal{A}_\eta, 
        \mathbf{N}_a=\mathbf{n}_a ] 
        \frac{\Pr(\mathcal{A}_\eta \mid \mathbf{N}_a = \mathbf{n}_a) \int_{0}^\infty f_{\mathbf{N}_a|Z}(\mathbf{n}_a|z) f_Z(z) dz}{\Pr(\mathcal{A}_\eta)}
        d\mathbf{n}_a\nonumber\\
        &= \frac{1}{\Pr(\mathcal{A}_\eta)} \int_{0}^\infty \int_{\mathcal{S}_z} \mathbb{E}[ \|\mathbf{U} - \hat{\mathbf{U}}\|_2^2 \mid \mathcal{A}_\eta, 
        \mathbf{N}_a=\mathbf{n}_a ]  \Pr(\mathcal{A}_\eta \mid \mathbf{N}_a = \mathbf{n}_a) f_{\mathbf{N}_a|Z}(\mathbf{n}_a|z) d\mathbf{n}_a f_Z(z)dz.
    \end{align}
    Let us define 
    \begin{align}\label{eq:surface_integral_MSE_case}
        \mathcal{I}(z) \triangleq \int_{\mathcal{S}_z} \mathbb{E}[ \|\mathbf{U} - \hat{\mathbf{U}}\|_2^2 \mid \mathcal{A}_\eta, 
        \mathbf{N}_a=\mathbf{n}_a ]  \Pr(\mathcal{A}_\eta \mid \mathbf{N}_a = \mathbf{n}_a) f_{\mathbf{N}_a|Z}(\mathbf{n}_a|z) d\mathbf{n}_a.
    \end{align}
    Comparing \eqref{eq:proof_integral_sub} with \eqref{eq:mse_lemma_result}, to prove Lemma \ref{lemma:mse_general}, it is sufficient to derive the kernel function \begin{align}\label{kernel:defeq:to_IZ}
        \Psi_N(z) \triangleq 4\mathcal{I}(z).
    \end{align}
    
    Expanding the quadratic form of the estimation error for a fixed vector $\mathbf{n}_a \in \mathcal{S}_z$, we obtain
\begin{align}\label{eq:expanded_mse}
    \mathbb{E}[ \|\mathbf{U} - \hat{\mathbf{U}}\|_2^2 \mid\! \mathcal{A}_\eta, \mathbf{N}_a\!=\! \mathbf{n}_a ] &= \frac{1}{4} \Big(  \|\mathbf{n}_a\|_2^2 + 2\mathbf{n}_a^\top \mathbb{E}[\mathbf{N}_h \mid\! \mathcal{A}_\eta, \mathbf{N}_a\!=\!\mathbf{n}_a] +\mathbb{E}[\|\mathbf{N}_h\|_2^2 \mid\! \mathcal{A}_\eta, \mathbf{N}_a\!=\!\mathbf{n}_a]  \Big).
\end{align}
Note that $\|\mathbf{n}_a\|_2^2=z^2$ for every $\mathbf{n}_a\in \mathcal{S}_z$. However, we need to find $f_{\mathbf{N}_h | \mathcal{A}_\eta, \mathbf{n}_a}(\mathbf{n}_h|\mathbf{n}_a)$ to further simplify~\eqref{eq:expanded_mse}, which is needed to compute $\mathcal{I}(z)$. In the following, we  
evaluate the term $\mathcal{I}(z)$ for the following three cases, which are defined based on the overlap between honest noise ball ${\mathcal{R}_{\text{honest}} = \mathcal{B}_N( \Delta)}$ and the acceptance region $\mathcal{R}_{\text{acc}}(\mathbf{n}_a) = \mathcal{B}_N( \eta\Delta, \mathbf{n}_a)$.

\subsection{Case 1: Complete Containment (\texorpdfstring{$0 \le z \le (\eta-1)\Delta$}{0 <= z <= (eta-1)Delta})}

In this regime, the magnitude of the adversarial noise $z$ is sufficiently small that the support of the honest noise  is entirely contained within the acceptance region, i.e., $\mathcal{R}_{\text{honest}} \subseteq \mathcal{R}_{\text{acc}}(\mathbf{n}_a)$, for any adversarial vector $\mathbf{n}_a$ with magnitude $z$. This implies that for any fixed $\mathbf{n}_a$ with ${\|\mathbf{n}_a\|_2 = z\in[0,(\eta-1)\Delta]}$ we have
\begin{align} \label{eq:containment_identity}
    \Pr(\mathcal{A}_\eta \mid \mathbf{N}_a = \mathbf{n}_a)=\Pr(\mathcal{A}_\eta \mid \mathbf{N}_a = \mathbf{n}_a, \mathbf{N}_h = \mathbf{n}_h) = 1, \quad \forall \mathbf{n}_h \in \mathcal{B}_N( \Delta).
\end{align}
Therefore, we have 
\begin{align}\label{eq:bayes_posterior_nh}
    f_{\mathbf{N}_h | \mathcal{A}_\eta, \mathbf{N}_a}(\mathbf{n}_h|\mathbf{n}_a) &= \frac{\Pr(\mathcal{A}_\eta \mid \mathbf{N}_a= \mathbf{n}_a, \mathbf{N}_h=\mathbf{n}_h) f_{\mathbf{N}_h | \mathbf{N}_a}(\mathbf{n}_h|\mathbf{n}_a)}{\Pr(\mathcal{A}_\eta \mid \mathbf{N}_a=\mathbf{n}_a)} \nonumber \\
    &\overset{(a)}{=} \frac{\Pr(\mathcal{A}_\eta \mid \mathbf{N}_a=\mathbf{n}_a, \mathbf{N}_h=\mathbf{n}_h) f_{\mathbf{N}_h}(\mathbf{n}_h)}{\Pr(\mathcal{A}_\eta \mid \mathbf{N}_a=\mathbf{n}_a)} \nonumber \\
    &\overset{(b)}{=} \frac{1 \cdot f_{\mathbf{N}_h}(\mathbf{n}_h)}{1} = f_{\mathbf{N}_h}(\mathbf{n}_h),
\end{align}
where (a) follows from the fact that the honest noise $\mathbf{N}_h$ is generated independently of the adversarial noise $\mathbf{N}_a$, and (b) follows from the containment condition in \eqref{eq:containment_identity}. The identity in \eqref{eq:bayes_posterior_nh} demonstrates that the posterior distribution of $\mathbf{N}_h$ remains a uniform distribution over the ball $\mathcal{B}_N( \Delta)$, i.e., conditioning on the acceptance event $\mathcal{A}_\eta$ and the realization $\mathbf{n}_a$ with $\|\mathbf{n}_a\|_2=z$ provides no additional information about the honest noise in this regime.
Therefore, we can evaluate the terms in~\eqref{eq:expanded_mse} as follows. 
\begin{itemize}
\item \textbf{Cross Term (First Moment):} Since $\mathbf{N}_h$ is uniformly distributed over the ball $\mathcal{B}_N( \Delta)$, which is centered at the origin, its expected value is the zero vector. Thus, the cross term vanishes
    \begin{align} \label{eq:item_cross_term}
        2\mathbf{n}_a^\top \mathbb{E}[\mathbf{N}_h \mid \mathcal{A}_\eta, \mathbf{N}_a=\mathbf{n}_a] = 2\mathbf{n}_a^\top \mathbb{E}[\mathbf{N}_h] = 2\mathbf{n}_a^\top \mathbf{0} = 0.
    \end{align}
    
    \item \textbf{Second Moment of Honest Noise:} The expectation of the squared magnitude is calculated by integrating over the uniform ball $\mathcal{B}_N( \Delta)$. Using the result from Equation \eqref{eq:second_moment_ball_formula}, we have
    \begin{align} \label{eq:item_second_moment}
        \mathbb{E}[\|\mathbf{N}_h\|_2^2 \mid \mathcal{A}_\eta, \mathbf{N}_a=\mathbf{n}_a] = \mathbb{E}[\|\mathbf{N}_h\|_2^2] = \frac{N}{N+2}\Delta^2.
    \end{align}
    

\end{itemize}

Substituting the results from \eqref{eq:item_second_moment}, and \eqref{eq:item_cross_term} back into \eqref{eq:expanded_mse}, we find that for any $\mathbf{n}_a \in \mathcal{S}_z$
\begin{align} \label{eq:error_invariant_case1}
    \mathbb{E}[ \|\mathbf{U} - \hat{\mathbf{U}}\|_2^2 \mid \mathcal{A}_\eta, \mathbf{N}_a=\mathbf{n}_a ] = \frac{1}{4} \left( z^2 + \frac{N}{N+2}\Delta^2  \right).
\end{align}

Crucially, the expression in \eqref{eq:error_invariant_case1} depends only on the magnitude $z$ and is invariant to the direction of $\mathbf{n}_a$. Hence, it can be moved out of the integral in the definition of $\mathcal{I}(z)$. Plugging~\eqref{eq:containment_identity} and~\eqref{eq:error_invariant_case1} into the integral in \eqref{eq:surface_integral_MSE_case}, we arrive at
\begin{align} \label{eq:final_reduction_case1}
    \mathcal{I}(z) &= \frac{1}{4}\left( z^2 + \frac{N}{N+2}\Delta^2 \right) \int_{\mathcal{S}_z} f_{\mathbf{N}_a|Z}(\mathbf{n}_a | z) \, d\mathbf{n}_a \nonumber \\
    &\overset{(c)}{=} \frac{1}{4}\left( z^2 + \frac{N}{N+2}\Delta^2 \right),
\end{align}
where (c) follows from the fact that the conditional PDF $f_{\mathbf{N}_a|Z}(\mathbf{n}_a | z)$ must integrate to unity over its support $\mathcal{S}_z$. By substituting \eqref{eq:final_reduction_case1} into \eqref{kernel:defeq:to_IZ} and comparing the result to \eqref{eq:Psi_piecewise}, we find that it matches the first branch of the piecewise function $\Psi_N(z)$.

\subsection{Case 2: Partial Overlap (\texorpdfstring{$(\eta-1)\Delta < z < (\eta+1)\Delta$}{(eta-1)Delta < z < (eta+1)Delta})}

In this regime, the magnitude of the adversarial noise $z$ results in a partial intersection between the support of the honest noise and the acceptance region. Conditioned on $\mathcal{A}_\eta$ and any fixed realization $\mathbf{n}_a \in \mathcal{S}_z$ of the adversarial noise, the honest noise vector $\mathbf{N}_h$ is constrained by two distinct geometric requirements: its prior support $\mathcal{R}_{\text{honest}} = \mathcal{B}_N( \Delta)$ and the acceptance region $\mathcal{R}_{\text{acc}}(\mathbf{n}_a) = \mathcal{B}_N( \eta\Delta, \mathbf{n}_a)$.
Hence, the posterior support of $\mathbf{N}_h$ is determined by
$\mathcal{R}_{\text{lens}}(\mathbf{n}_a)$ (See Figure \ref{fig:region_decomposition}), where
\begin{align} \label{eq:R_lens_def}
    \mathcal{R}_{\text{lens}}(\mathbf{n}_a) \triangleq \mathcal{R}_{\text{honest}} \cap \mathcal{R}_{\text{acc}}(\mathbf{n}_a)= \{ \mathbf{n} \in \mathbb{R}^N \mid \|\mathbf{n}\|_2 \le \Delta \text{ and } \|\mathbf{n} - \mathbf{n}_a\|_2 \le \eta\Delta \}.
\end{align}
Based on the general formula for the intersection of two $N$-balls, the volume of this region is $V_{\text{lens}}(\Delta, \eta\Delta, z)$, as defined in \eqref{eq:general_intersection}. Recall from~\eqref{eq:fixed_na_ratio} in the proof of Lemma \ref{lemma:prob_acceptance_general}, that
\begin{align}\label{eq:PA_to_V_lens}
    \Pr(\mathcal{A}_\eta \mid \mathbf{N}_a=\mathbf{n}_a) = \frac{V_{\text{lens}}(\Delta, \eta\Delta, z)}{V_N(\Delta)}.
\end{align}
We now determine the posterior distribution of the honest noise $\mathbf{N}_h$ conditioned on both the acceptance event $\mathcal{A}_\eta$ and the fixed vector $\mathbf{n}_a$. Applying Bayes' theorem, we have
\begin{align}\label{case2_bayes_posterior}
    f_{\mathbf{N}_h \mid \mathcal{A}_\eta, \mathbf{N}_a}(\mathbf{n}_h |\mathbf{n}_a) &= \frac{\Pr(\mathcal{A}_\eta \mid \mathbf{N}_h = \mathbf{n}_h, \mathbf{N}_a = \mathbf{n}_a) f_{\mathbf{N}_h | \mathbf{N}_a}(\mathbf{n}_h|\mathbf{n}_a)}{\Pr(\mathcal{A}_\eta \mid \mathbf{N}_a=\mathbf{n}_a)} \nonumber \\
    &\overset{(a)}{=} \frac{\Pr(\mathcal{A}_\eta \mid \mathbf{N}_h = \mathbf{n}_h, \mathbf{N}_a = \mathbf{n}_a) f_{\mathbf{N}_h}(\mathbf{n}_h)}{\Pr(\mathcal{A}_\eta \mid \mathbf{N}_a=\mathbf{n}_a)} \nonumber \\
    &\overset{(b)}{=} \frac{\mathbb{I}(\mathbf{n}_h \in \mathcal{R}_{\text{lens}}(\mathbf{n}_a)) \cdot \frac{1}{V_N(\Delta)}}{V_{\text{lens}}(\Delta, \eta\Delta, z) / V_N(\Delta)} \nonumber \\
    &= \frac{1}{V_{\text{lens}}(\Delta, \eta\Delta, z)} \mathbb{I}(\mathbf{n}_h \in \mathcal{R}_{\text{lens}}(\mathbf{n}_a)),
\end{align}
where (a) follows from the independence of $\mathbf{N}_h$ and $\mathbf{N}_a$, and (b) follows from substituting \eqref{eq:PA_to_V_lens}, applying the uniform prior of $\mathbf{N}_h$ over $\mathcal{B}_N( \Delta)$, and utilizing the geometric definition of the acceptance event for a fixed $\mathbf{n}_a$. This confirms that $\mathbf{N}_h$ is uniformly distributed over the intersection region $\mathcal{R}_{\text{lens}}(\mathbf{n}_a)$.

Now, we are ready to evaluate the terms in~\eqref{eq:expanded_mse}. To this end, without loss of generality, we align $\mathbf{n}_a$ with the first axis, i.e., we assume $\mathbf{n}_a = [z, 0, \dots, 0]^\top$. In this alignment, we have $\mathbf{n}_a^\top \mathbf{N}_h = z N_{h,1}$, where $N_{h,1}$ is the first component of the honest noise vector. Substituting this into \eqref{eq:expanded_mse} leads to
\begin{align}\label{eq:expected_expanded}
    \mathbb{E}\Big[ \|\mathbf{U} - \hat{\mathbf{U}}\|_2^2 \mid &\mathcal{A}_\eta, \mathbf{N}_a=\mathbf{n}_a \Big] = \frac{1}{4} \left( z^2 + 2z \mathbb{E}[N_{h,1} \mid \mathcal{A}_\eta, \mathbf{N}_a=\mathbf{n}_a] + \mathbb{E}\left[ \|\mathbf{N}_h\|_2^2 \mid \mathcal{A}_\eta, \mathbf{N}_a=\mathbf{n}_a \right]  \right)\nonumber\\
    &=\frac{1}{4} \left(z^2 + 2z \int_{\mathcal{B}_N(\Delta)} n_{h,1} f_{\mathbf{N}_h \mid \mathcal{A}_\eta, \mathbf{n}_a}(\mathbf{n}_h) d\mathbf{n}_h+ \int_{\mathcal{B}_N(\Delta)} \|\mathbf{n}_h\|_2^2 f_{\mathbf{N}_h \mid \mathcal{A}_\eta, \mathbf{n}_a}(\mathbf{n}_h) d\mathbf{n}_h \right)\nonumber\\
    &\overset{(c)}{=}\frac{1}{4} \left(z^2 + 2z \int_{\mathcal{R}_{\text{lens}}(\mathbf{n}_a)} n_{h,1} \frac{1}{V} d\mathbf{n}_h+ \int_{\mathcal{R}_{\text{lens}}(\mathbf{n}_a)} \|\mathbf{n}_h\|_2^2 \frac{1}{V} d\mathbf{n}_h \right)\nonumber\\
    &=\frac{1}{4V} \left( z^2 V + 2z I_1+ I_2  \right),
\end{align}
where $V = V_{\text{lens}}(\Delta, \eta\Delta, z)$, and
\begin{align}
    I_1 &= \int_{\mathcal{R}_{\text{lens}}(\mathbf{n}_a)} n_{h,1} \, d\mathbf{n}_h, \label{eq:integral_first_moment} \\
    I_2 &= \int_{\mathcal{R}_{\text{lens}}(\mathbf{n}_a)} \|\mathbf{n}_h\|_2^2 \, d\mathbf{n}_h. \label{eq:integral_second_moment}
\end{align}
Moreover, in (c) we replaced the conditional PDF of $\mathbf{N}_h$ from~\eqref{case2_bayes_posterior}. 

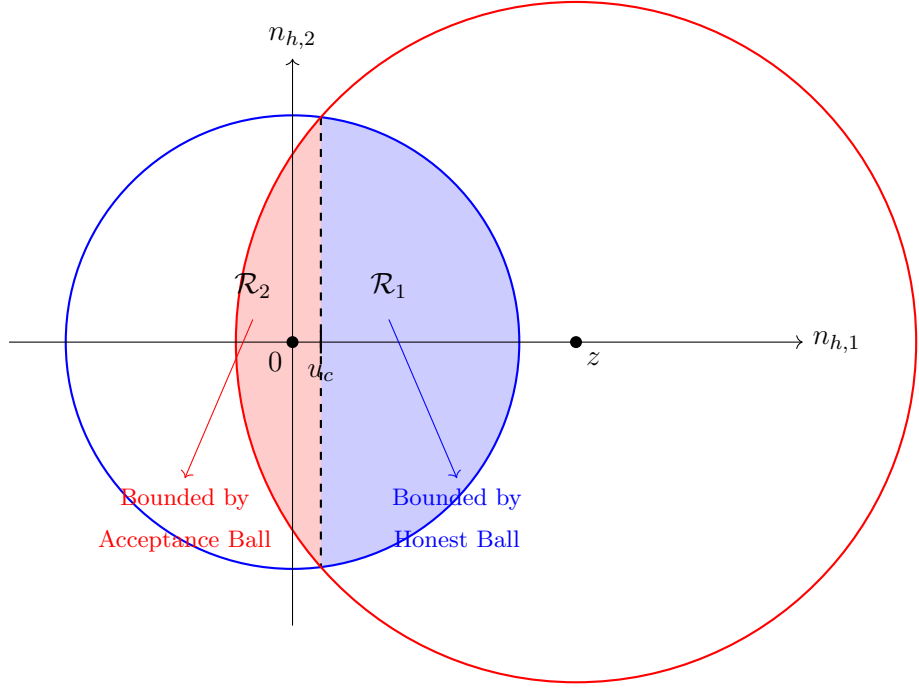
\begin{figure}[htbp]
    \centering
    \begin{tikzpicture}[scale=1.5]
        \def\rOne{2.0}      
        \def\rTwo{3}        
        \def\d{2.5}         
        
        \coordinate (O1) at (0,0);
        \coordinate (O2) at (\d,0);
        
        \pgfmathsetmacro{\uc}{(\d*\d - \rTwo*\rTwo + \rOne*\rOne)/(2*\d)}
        \pgfmathsetmacro{\yint}{sqrt(\rOne*\rOne - \uc*\uc)}
        
        \coordinate (Top) at (\uc, \yint);
        \coordinate (Bottom) at (\uc, -\yint);

        \begin{scope}
            \clip (\uc, -3) rectangle (\rOne, 3); 
            \fill[blue!20] (O1) circle (\rOne);
        \end{scope}
        
        \begin{scope}
            \clip (\d-\rTwo, -3) rectangle (\uc, 3); 
            \fill[red!20] (O2) circle (\rTwo);
        \end{scope}

        \draw[thick, dashed] (Top) -- (Bottom);

        \draw[thick, blue] (O1) circle (\rOne);
        \draw[thick, red] (O2) circle (\rTwo);

        \draw[->] (-2.5,0) -- (4.5,0) node[right] {$n_{h,1}$};
        \draw[->] (0,-2.5) -- (0,2.5) node[above] {$n_{h,2}$};

        \fill (O1) circle (1.5pt) node[below left] {$0$};
        \fill (O2) circle (1.5pt) node[below right] {$z$};
        \draw[thick] (\uc, 0.1) -- (\uc, -0.1) node[below] {$u_c$};

        \node at (\uc + 0.6, 0.5) {$\mathcal{R}_1$};
        \node at (\uc - 0.6, 0.5) {$\mathcal{R}_2$};
        
        \draw[->, blue] (\uc + 0.6, 0.2) -- (\uc + 1.2, -1.2) node[below, align=center, font=\footnotesize] {Bounded by\\Honest Ball};
        \draw[->, red] (\uc - 0.6, 0.2) -- (\uc - 1.2, -1.2) node[below, align=center, font=\footnotesize] {Bounded by\\Acceptance Ball};
    \end{tikzpicture}
    \caption{Cross-sectional decomposition of the $N$-dimensional intersection region $\mathcal{R}_{\text{lens}}(\mathbf{n}_a)$. The vertical line at $n_{h,1} = u_c$ represents the radical hyperplane, which divides the volume into standard cap $\mathcal{R}_1$ and shifted cap $\mathcal{R}_2$.}
    \label{fig:region_decomposition}
\end{figure}

Before calculating $I_1$ and $I_2$, we analyze the geometry of the integration domain. As discussed in Appendix~\ref{Sec:Partial_Overlap}, the integration domain $\mathcal{R}_{\text{lens}}(\mathbf{n}_a)$ comprises of two hyperspherical caps, obtained by cutting the balls by the radical plane. More precisely, we have  $\mathcal{R}_{\text{lens}}(\mathbf{n}_a) = \mathcal{R}_1 \cup \mathcal{R}_2$, where
    \begin{equation}
        \mathcal{R}_1 = \{ \mathbf{r} \in \mathbb{R}^N \mid \|\mathbf{r}\|_2 \le \Delta \text{ and } r_1 \ge u_c \},
    \end{equation}
    and
    \begin{equation}
        \mathcal{R}_2 = \{ \mathbf{r} \in \mathbb{R}^N \mid \|\mathbf{r} - \mathbf{n}_a\|_2 \le \eta\Delta \text{ and } r_1 \le u_c \},
    \end{equation}
    and the cutting point $u_c$ is given by
    \begin{align}
    u_c = \frac{z^2 + \Delta^2(1 - \eta^2)}{2z}.
\end{align}
Based on the definition of the hyperspherical cap volume in \eqref{eq:K_calculate}, the volume of $\mathcal{R}_1$ is given by $V_1 = K_N(\Delta, u_c)$. Moreover, $\mathcal{R}_2$ corresponds to a shifted, left-oriented cap. As established in the geometric analysis of Figure \ref{fig:shifted_cap}, the volume of $\mathcal{R}_2$ is determined by the radius $\eta\Delta$ and the distance from the center $z - u_c$. Thus, based on \eqref{eq:K_calculate}, the volume of $\mathcal{R}_2$ is $V_2 = K_N(\eta\Delta, z - u_c)$. 

\begin{itemize}
\item \textbf{Cross Term (First Moment):} Using this decomposition above, we can rewrite the moment integrals $I_1$ in \eqref{eq:integral_first_moment} as 
    \begin{align}
        I_1 &= \int_{\mathcal{R}_1} n_{h,1} \, d\mathbf{n}_h + \int_{\mathcal{R}_2} n_{h,1} \, d\mathbf{n}_h, \label{seperate_I1}
    \end{align}
Based on definition \eqref{eq:def_Q}, we have     \begin{equation}\label{Part1_I1}
            \int_{\mathcal{R}_1} n_{h,1} \, d\mathbf{n}_h = Q_N(\Delta, u_c).
        \end{equation}
        Similarly, using~\eqref{eq:shifted_Q_n}, we can write
        \begin{equation}\label{Part2_I1}
            \int_{\mathcal{R}_2} n_{h,1} \, d\mathbf{n}_h = -Q_N(\eta\Delta, z-u_c) + z V_2.
        \end{equation}
    It is important to note from \eqref{eq:result_Q_n} that $Q_N(r, d)$ depends only on the intersection height ${h = \sqrt{r^2 - d^2}}$. Since both caps share the same intersection boundary (the $(N-1)$-sphere of radius $h$), we have $Q_N(\Delta, u_c) = Q_N(\eta\Delta, z-u_c)$.
    Thus, substituting  \eqref{Part1_I1} and \eqref{Part2_I1} in \eqref{seperate_I1}, we have 
    \begin{equation}
    \label{eq:I1_result}
        I_1 = Q_N(\Delta, u_c) - Q_N(\eta\Delta, z-u_c) + z V_2 = z V_2.
    \end{equation}     
    
    \item \textbf{Second Moment of Honest Noise:} Applying the same decomposition for the integration region, we have 
    \begin{align}
        I_2 &= \int_{\mathcal{R}_1} \|\mathbf{n}_h\|_2^2 \, d\mathbf{n}_h + \int_{\mathcal{R}_2} \|\mathbf{n}_h\|_2^2 \, d\mathbf{n}_h. \label{seperate_I2}
    \end{align}
     Using the definition in~\eqref{eq:definiton_of_J}, we have
    \begin{equation}\label{part1_I2}
            \int_{\mathcal{R}_1} \|\mathbf{n}_h\|_2^2 \, d\mathbf{n}_h = J_N(\Delta, u_c).
        \end{equation}
       Similarly, using the identity in~\eqref{eq:shifted_J_N}, we get
   \begin{equation}\label{part2_I2}
            \int_{\mathcal{R}_2} \|\mathbf{n}_h\|_2^2 \, d\mathbf{n}_h = J_N(\eta\Delta, z-u_c) - 2z Q_N(\eta\Delta, z-u_c) + z^2 V_2.
        \end{equation}
    
    Therefore, plugging \eqref{part1_I2} and \eqref{part2_I2} into~\eqref{seperate_I2}, we arrive at
    \begin{equation}
    \label{eq:I2_result}
        I_2 = J_N(\Delta, u_c) + J_N(\eta\Delta, z-u_c) - 2z Q_N(\eta\Delta, z-u_c) + z^2 V_2.
    \end{equation}

    Now, we have all the terms to evaluate~\eqref{eq:expected_expanded}. Let us define 
    \begin{align}\label{eq:Psi-lens}
        \Psi_N^{\text{lens}}(z) &\triangleq z^2 V + 2zI_1 + I_2 \nonumber\\
        &\overset{(a)}{=} z^2 (V_1+V_2) + 2z^2 V_2 + J_N(\Delta, u_c) + J_N(\eta\Delta, z-u_c) - 2z Q_N(\eta\Delta, z-u_c) + z^2 V_2\nonumber\\
        &=
        \Big[ J_N(\Delta, u_c) + z^2 V_1 \Big] + \Big[ J_N(\eta\Delta, z-u_c) + 4z^2 V_2 - 2z Q_N(\eta\Delta, z-u_c) \Big],
    \end{align}
    where (a) follows from $V=V_1+V_2$ and substituting $I_1$ and $I_2$ from \eqref{eq:I1_result} and \eqref{eq:I2_result}, respectively. Using~\eqref{eq:Psi-lens} in~\eqref{eq:expected_expanded}, we get \begin{align}\label{eq:final_condition_on_accep_z}
        \mathbb{E}\left[ \|\mathbf{U} - \hat{\mathbf{U}}\|_2^2 \mid \mathcal{A}_\eta, \mathbf{N}_a=\mathbf{n}_a \right] = \frac{\Psi_N^{\text{lens}}(z)}{4 V_{\text{lens}}(\Delta, \eta\Delta, z)}.
    \end{align}
    Finally, plugging~\eqref{eq:PA_to_V_lens} and~\eqref{eq:final_condition_on_accep_z} into~\eqref{eq:surface_integral_MSE_case}, we obtain $I(z)$ for the second case as  
\begin{align}\label{I_z_case2_final_derivation}
        \mathcal{I}(z) &= \int_{\mathcal{S}_z} \left( \frac{\Psi_N^{\text{lens}}(z)}{4 V_{\text{lens}}(\Delta, \eta\Delta, z)} \right) \left( \frac{V_{\text{lens}}(\Delta, \eta\Delta, z)}{V_N(\Delta)} \right) f_{\mathbf{N}_a|Z}(\mathbf{n}_a | z) \, d\mathbf{n}_a \nonumber \\
        &= \frac{\Psi_N^{\text{lens}}(z)}{4 V_N(\Delta)} \int_{\mathcal{S}_z} f_{\mathbf{N}_a|Z}(\mathbf{n}_a | z) \, d\mathbf{n}_a \nonumber \\
        &= \frac{\Psi_N^{\text{lens}}(z)}{4 V_N(\Delta)},
    \end{align}
    where the last equality follows from the fact that the conditional density $f_{\mathbf{N}_a|Z}(\mathbf{n}_a | z)$ integrates to unity over its support $\mathcal{S}_z$. 

\end{itemize}

By substituting \eqref{I_z_case2_final_derivation} into \eqref{kernel:defeq:to_IZ} and comparing the result to \eqref{eq:Psi_piecewise}, we find that it matches the second branch of the piecewise function $\Psi_N(z)$.

\subsection{Case 3: No Intersection (\texorpdfstring{$z \ge (\eta+1)\Delta$}{z >= (eta+1)Delta})}

In this final regime, since $z \ge (\eta+1)\Delta$, the distance between the centers of honest noise ${\mathcal{R}_{\text{honest}} = \mathcal{B}_N( \Delta)}$ and the acceptance region $\mathcal{R}_{\text{acc}}(\mathbf{n}_a) = \mathcal{B}_N( \eta\Delta, \mathbf{n}_a)$ is greater than the sum of their radii. Thus, the two balls are disjoint, i.e.,  ${\mathcal{R}_{\text{honest}}\cap \mathcal{R}_{\text{acc}}(\mathbf{n}_a)= \varnothing}$. Therefore,  
\begin{align}
    \Pr(\mathcal{A}_\eta \mid Z=z) =  0.
\end{align}
Substituting this into the definition of  $\mathcal{I}(z)$ in \eqref{eq:surface_integral_MSE_case}, we obtain
\begin{align}\label{Iz_case3}
    \mathcal{I}(z)  =\int_{\mathcal{S}_z} 0 d\mathbf{n}_h =0.
\end{align}
Based on the definition of \eqref{kernel:defeq:to_IZ}, 
the result of \eqref{Iz_case3} corresponds exactly to the third branch of the piecewise function $\Psi_N(z)$ defined in \eqref{eq:Psi_piecewise}.
This completes the proof of Lemma \ref{lemma:mse_general}.

\section{Proof of Lemma \ref{lemma:optimal_support}}\label{proof:lemma:optimal_support}

To prove Lemma \ref{lemma:optimal_support}, we proceed in two steps. First, we show that any probability mass located in the region $z > (\eta+1)\Delta$ can be removed and redistributed to the other region, thereby increasing the probability of acceptance while maintaining the conditional MSE. Second, we show that any probability mass located in the region $z < (\eta-1)\Delta$ can be shifted to the point $z = (\eta-1)\Delta$ to increase the conditional MSE without affecting the probability of acceptance.

\subsection*{Step 1: Removing mass from $z > (\eta+1)\Delta$}

Consider an initial noise distribution 
$f_Z(z)$ and let 
\begin{align}
    P_{\text{tail}} \triangleq \int_{(\eta+1)\Delta}^{\infty} f_Z(z) \, dz.
\end{align}
If $P_{\text{tail}} = 0$, the support is already bounded from above\footnote{Note that if $P_{\text{tail}} = 1$, the entire probability mass of $f_Z(z)$ lies in the region $z \ge (\eta+1)\Delta$. According to Lemma~\ref{lemma:prob_acceptance_general} and~\eqref{eq:Phi_piecewise}, $\Phi_N(z) = 0$ for $z \ge (\eta+1)\Delta$. This immediately implies that $\Pr(\mathcal{A}_\eta; f_Z) = 0$, which contradicts our initial assumption in the statement of Lemma \ref{lemma:optimal_support}.}. Otherwise, as illustrated in Figure \ref{fig:noise_step0}, assume that some mass exists beyond the desired boundary $(\eta+1)\Delta$. According to \eqref{eq:Phi_piecewise} and \eqref{eq:Psi_piecewise}, we have $\Phi_N(z) = 0$ and $\Psi_N(z) = 0$ for $z \ge (\eta+1)\Delta$.

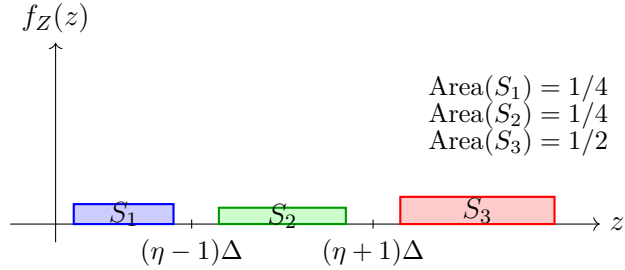
\begin{figure}[htbp]
    \centering
    \begin{tikzpicture}[scale=1.2]
        \draw[->] (-0.5,0) -- (6,0) node[right] {$z$};
        \draw[->] (0,-0.2) -- (0,2) node[above] {$f_Z(z)$};
        
        \draw (1.5,0.05) -- (1.5,-0.05) node[below, font=\small] {$(\eta-1)\Delta$};
        \draw (3.5,0.05) -- (3.5,-0.05) node[below, font=\small] {$(\eta+1)\Delta$};
        
        \fill[blue!20] (0.2,0) rectangle (1.3, 0.22);
        \draw[thick, blue] (0.2,0) rectangle (1.3, 0.22);
        \node at (0.75, 0.11) {$S_1$};
        
        \fill[green!20] (1.8,0) rectangle (3.2, 0.18);
        \draw[thick, green!60!black] (1.8,0) rectangle (3.2, 0.18);
        \node at (2.5, 0.09) {$S_2$};
        
        \fill[red!20] (3.8,0) rectangle (5.5, 0.3);
        \draw[thick, red] (3.8,0) rectangle (5.5, 0.3);
        \node at (4.65, 0.15) {$S_3$};

        \node[anchor=west, font=\small] at (4, 1.5) {$\text{Area}(S_1) = 1/4$};
        \node[anchor=west, font=\small] at (4, 1.2) {$\text{Area}(S_2) = 1/4$};
        \node[anchor=west, font=\small] at (4, 0.9) {$\text{Area}(S_3) = 1/2$};
    \end{tikzpicture}
    \caption{Visualization of an initial adversarial noise density $f_Z(z)$. The total probability mass is partitioned into three regions: $S_1$ (head), $S_2$ (support), and $S_3$ (tail), where the mass $P_{\text{tail}} = \text{Area}(S_3)$ lies in the zero-acceptance region $z > (\eta+1)\Delta$.}
    \label{fig:noise_step0}
\end{figure}
We define an intermediate distribution $f_1(z)$ by truncating and normalizing $f_Z(z)$, as shown in Figure \ref{fig:noise_step1}, where the remaining mass is scaled up to maintain a valid PDF:
\begin{align}\label{define_f_1}
    f_1(z) = \begin{cases}
        \frac{f_Z(z)}{1 - P_{\text{tail}}} & \text{if } z \le (\eta+1)\Delta, \\
        0 & \text{if } z > (\eta+1)\Delta.
    \end{cases}
\end{align}

\begin{figure}[htbp]
    \centering
    \begin{tikzpicture}[scale=1.2]
        \draw[->] (-0.5,0) -- (6,0) node[right] {$z$};
        \draw[->] (0,-0.2) -- (0,2) node[above] {$f_1(z)$};
        
        \draw (1.5,0.05) -- (1.5,-0.05) node[below, font=\small] {$(\eta-1)\Delta$};
        \draw (3.5,0.05) -- (3.5,-0.05) node[below, font=\small] {$(\eta+1)\Delta$};
        
        \fill[blue!20] (0.2,0) rectangle (1.3, 0.44);
        \draw[thick, blue] (0.2,0) rectangle (1.3, 0.44);
        \node at (0.75, 0.22) {$S'_1$};
        
        \fill[green!20] (1.8,0) rectangle (3.2, 0.36);
        \draw[thick, green!60!black] (1.8,0) rectangle (3.2, 0.36);
        \node at (2.5, 0.18) {$S'_2$};

        \draw[dashed, gray] (3.8,0) rectangle (5.5, 0.05);
        \node[gray, font=\tiny] at (4.65, 0.2) {Removed};

        \node[anchor=west, font=\small] at (4, 1.5) {$\text{Area}(S'_1) = 1/2$};
        \node[anchor=west, font=\small] at (4, 1.2) {$\text{Area}(S'_2) = 1/2$};
    \end{tikzpicture}
    \caption{Intermediate distribution $f_1(z)$ following the truncation and renormalization process defined in \eqref{define_f_1}. By removing the tail mass $P_{\text{tail}}$ and scaling the remaining density, the probability of acceptance is increased according to \eqref{new_noise_no_tail_Prob} while the conditional MSE remains constant.}
    \label{fig:noise_step1}
\end{figure}
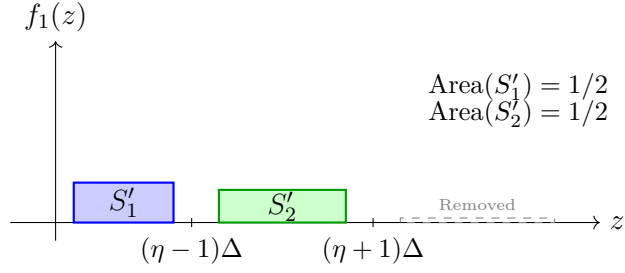

Using Lemma \ref{lemma:prob_acceptance_general}, the acceptance probability for $f_1(z)$ is
\begin{align}
    \Pr(\mathcal{A}_\eta; f_1) &= \int_{0}^{\infty} \Phi_N(z) f_1(z) \, dz \nonumber \\
    &\overset{(a)}{=} \frac{1}{1 - P_{\text{tail}}} \int_{0}^{(\eta+1)\Delta} \Phi_N(z) f_Z(z) \, dz \nonumber \\
    &\overset{(b)}{=} \frac{1}{1 - P_{\text{tail}}} \int_{0}^{\infty} \Phi_N(z) f_Z(z) \, dz \nonumber\\
    &\overset{(c)}{=} \frac{1}{1 - P_{\text{tail}}} \Pr(\mathcal{A}_\eta; f_Z) \label{new_noise_no_tail_Prob} \\
    &\ge \Pr(\mathcal{A}_\eta; f_Z),\label{eq:Pr-Acc-ge}
\end{align}
where (a) follows from \eqref{define_f_1} and (b) is due to the fact that $\Phi_N(z)=0$ for $z > (\eta+1)\Delta$, and (c) follows from Lemma \ref{lemma:prob_acceptance_general}.
Thus, the acceptance probability increases (or stays the same if $P_{\text{tail}}=0$).

Next, we check the conditional MSE, i.e., 
$\mathbb{E}\left[ \|\mathbf{U} - \hat{\mathbf{U}}\|_2^2 \mid \mathcal{A}_\eta; f_1 \right]$. 
 Using Lemma \ref{lemma:mse_general}, we can write
\begin{align}\label{equal_mse_skipping_tail}
    \mathbb{E}\left[ \|\mathbf{U} - \hat{\mathbf{U}}\|_2^2 \mid \mathcal{A}_\eta; f_1 \right] &= \frac{1}{4 \Pr(\mathcal{A}_\eta; f_1)} \int_{0}^{\infty} \Psi_N(z) f_1(z) \, dz \nonumber \\
    &\overset{(a)}{=} \frac{1 - P_{\text{tail}}}{4 \Pr(\mathcal{A}_\eta; f_Z)} \cdot \frac{1}{1 - P_{\text{tail}}} \int_{0}^{(\eta+1)\Delta} \Psi_N(z) f_Z(z) \, dz \nonumber \\
    &\overset{(b)}{=} \frac{1}{4 \Pr(\mathcal{A}_\eta; f_Z)} \int_{0}^{\infty} \Psi_N(z) f_Z(z) \, dz \nonumber \\
    &\overset{(c)}{=} \mathbb{E}\left[ \|\mathbf{U} - \hat{\mathbf{U}}\|_2^2 \mid \mathcal{A}_\eta; f_Z \right],
\end{align}
where (a) follows from  \eqref{define_f_1} and \eqref{new_noise_no_tail_Prob}, (b) follows from fact that $\Psi_N(z)=0$ for $z > (\eta+1)\Delta$, and (c) follows from Lemma \ref{lemma:mse_general}. Thus, removing the tail does not change the conditional MSE.

\subsection*{Step 2: Shifting mass from $z < (\eta-1)\Delta$}

Now consider the distribution $f_1(z)$ from Step 1, which is supported on $[0, (\eta+1)\Delta]$. We construct the final distribution $f^*_Z(z)$ by shifting all mass from $[0, (\eta-1)\Delta)$ to a Dirac delta function at $z = (\eta-1)\Delta$, as illustrated in Figure \ref{fig:noise_step2}.
Let 
\begin{align}\label{def:P_head}
    P_{\text{head}} \triangleq \int_{0}^{(\eta-1)\Delta} f_1(z) \, dz.
\end{align}
 We define
\begin{align}\label{def:f*}
    f^*_Z(z) = P_{\text{head}} \delta(z - (\eta-1)\Delta) + f_1(z) \mathbb{I}\left( (\eta-1)\Delta \le z \le (\eta+1)\Delta \right).
\end{align}

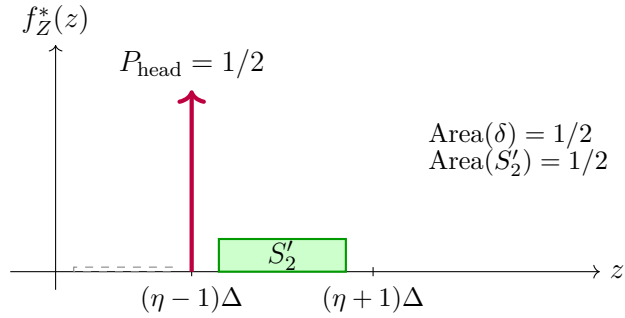
\begin{figure}[htbp]
    \centering
    \begin{tikzpicture}[scale=1.2]
        \draw[->] (-0.5,0) -- (6,0) node[right] {$z$};
        \draw[->] (0,-0.2) -- (0,2.5) node[above] {$f^*_Z(z)$};
        
        \draw (1.5,0.05) -- (1.5,-0.05) node[below, font=\small] {$(\eta-1)\Delta$};
        \draw (3.5,0.05) -- (3.5,-0.05) node[below, font=\small] {$(\eta+1)\Delta$};
        
        \draw[ultra thick, ->, purple] (1.5,0) -- (1.5,2.0) node[above, black] {$P_{\text{head}} = 1/2$};
        
        \draw[dashed, gray] (0.2,0) rectangle (1.3, 0.05);

        \fill[green!20] (1.8,0) rectangle (3.2, 0.36);
        \draw[thick, green!60!black] (1.8,0) rectangle (3.2, 0.36);
        \node at (2.5, 0.18) {$S'_2$};

        \node[anchor=west, font=\small] at (4, 1.5) {$\text{Area}(\delta) = 1/2$};
        \node[anchor=west, font=\small] at (4, 1.2) {$\text{Area}(S'_2) = 1/2$};
    \end{tikzpicture}
    \caption{Final optimal distribution $f^*_Z(z)$ (illustrating the construction in \eqref{def:f*}). All mass from the region $z < (\eta-1)\Delta$ is concentrated at the boundary point $(\eta-1)\Delta$. This shift maximizes the conditional MSE,  without decreasing the probability of acceptance.}
    \label{fig:noise_step2}
\end{figure}

First, we analyze the acceptance probability. We have
\begin{align}\label{eq:Pr-Acc-same}
    \Pr(\mathcal{A}_\eta; f^*_Z) &= \int_0^{\infty} \Phi_N(z) f^*_Z(z) \, dz \nonumber \\
    &\overset{(a)}{=} P_{\text{head}} \Phi_N((\eta-1)\Delta) + \int_{(\eta-1)\Delta}^{(\eta+1)\Delta} \Phi_N(z) f_1(z) \, dz \nonumber \\
    &\overset{(b)}{=} P_{\text{head}} \cdot 1 + \int_{(\eta-1)\Delta}^{(\eta+1)\Delta} \Phi_N(z) f_1(z) \, dz \nonumber \\
    &\overset{(c)}{=} \int_{0}^{(\eta-1)\Delta} 1 \cdot f_1(z) \, dz + \int_{(\eta-1)\Delta}^{(\eta+1)\Delta} \Phi_N(z) f_1(z) \, dz \nonumber \\
    &\overset{(d)}{=} \int_{0}^{\infty} \Phi_N(z) f_1(z) \, dz \nonumber \\
    &= \Pr(\mathcal{A}_\eta; f_1),
\end{align}
where (a) follows from the sifting property of the Dirac delta function, $\int_0^{\infty} h(z)\delta(z-z_0)\,dz = h(z_0)$ for $z_0 \ge 0$, and the definition of $f^*_Z(z)$ in \eqref{def:f*}; (b) follows from \eqref{eq:Phi_piecewise}, which implies $\Phi_N((\eta-1)\Delta) = 1$; (c) follows from the definition $P_{\text{head}}$ in \eqref{def:P_head}; and (d) follows from the fact that $\Phi_N(z) = 1$ for $0 \le z \le (\eta-1)\Delta$, and $\Phi_N(z) = 0$ for $z > (\eta+1)\Delta$, allowing us to combine the integration domains $[0, (\eta-1)\Delta]$ and $[(\eta-1)\Delta, \infty)$ back into $[0, \infty)$. Combining \eqref{eq:Pr-Acc-ge} and \eqref{eq:Pr-Acc-same}, we arrive at ${\Pr(\mathcal{A}_\eta; f^*_Z) \geq \Pr(\mathcal{A}_\eta; f_Z)}$, as claimed in the lemma.

Next, we examine the conditional MSE. Note that based on Lemma \ref{lemma:mse_general}, we have
\begin{align}
    \mathbb{E}\left[ \|\mathbf{U} - \hat{\mathbf{U}}\|_2^2 \mid \mathcal{A}_\eta; f^*_Z \right] &= \frac{1}{4 \Pr(\mathcal{A}_\eta; f^*_Z)}\int_{0}^{\infty} \Psi_N(z) f^*_Z(z) \, dz.
\end{align}
We start with analyzing the integral, and write
\begin{align}\label{numerator_compare}
    \int_{0}^{\infty} \Psi_N(z) f^*_Z(z) \, dz &\overset{(a)}{=} P_{\text{head}} \Psi_N((\eta-1)\Delta) + \int_{(\eta-1)\Delta}^{(\eta+1)\Delta} \Psi_N(z) f_1(z) \, dz \nonumber \\
    &\overset{(b)}{=} \Psi_N((\eta-1)\Delta) \left( \int_{0}^{(\eta-1)\Delta} f_1(z) \, dz \right) + \int_{(\eta-1)\Delta}^{(\eta+1)\Delta} \Psi_N(z) f_1(z) \, dz \nonumber \\
    &\overset{(c)}{\ge} \int_{0}^{(\eta-1)\Delta} \Psi_N(z) f_1(z) \, dz + \int_{(\eta-1)\Delta}^{(\eta+1)\Delta} \Psi_N(z) f_1(z) \, dz \nonumber \\
    &\overset{(d)}{=} \int_{0}^{\infty} \Psi_N(z) f_1(z) \, dz,
\end{align}
where (a) follows from the definition of $f^*_Z$ and the sifting property of the Dirac delta function, i.e.,  ${\int_0^{\infty} h(z)\delta(z-z_0)\,dz = h(z_0)}$ for $z_0 \ge 0$, ; in (b) we substituted the definition of $P_{\text{head}}$ from~\eqref{def:P_head}; (c) follows from the fact that $\Psi_N(z) = z^2 + \frac{N}{N+2}\Delta^2$ is strictly increasing for $0\leq z \leq (\eta-1)\Delta$, which implies $\Psi_N((\eta-1)\Delta) \ge \Psi_N(z)$ in the range of the first integral; and (d) follows from combining the integrals over $[0, (\eta+1)\Delta]$ and the fact that $\Psi_N(z) = 0$ for $z > (\eta+1)\Delta$, as implied from~\eqref{eq:Psi_piecewise}.

Therefore, using Lemma \ref{lemma:mse_general}, we can write 
\begin{align}
    \mathbb{E}\left[ \|\mathbf{U} - \hat{\mathbf{U}}\|_2^2 \mid \mathcal{A}_\eta; f^*_Z \right] &= \frac{1}{4 \Pr(\mathcal{A}_\eta; f^*_Z)}\int_{0}^{\infty} \Psi_N(z) f^*_Z(z) \, dz \nonumber \\
    &\overset{(a)}{\geq} \frac{1}{4 \Pr(\mathcal{A}_\eta; f^*_Z)} \int_{0}^{\infty} \Psi_N(z) f_1(z) \, dz\nonumber \\
    &\overset{(b)}{=} \frac{1}{4 \Pr(\mathcal{A}_\eta; f_1)} \int_{0}^{\infty} \Psi_N(z) f_1(z) \, dz\nonumber \\
    &= \mathbb{E}\left[ \|\mathbf{U} - \hat{\mathbf{U}}\|_2^2 \mid \mathcal{A}_\eta; f_1 \right] \nonumber \\
    &\overset{(c)}{=} \mathbb{E}\left[ \|\mathbf{U} - \hat{\mathbf{U}}\|_2^2 \mid \mathcal{A}_\eta; f_Z \right],
\end{align}
where (a) follows from \eqref{numerator_compare}, we used~\eqref{eq:Pr-Acc-same} in (b), and (c) follows from \eqref{equal_mse_skipping_tail}.
This completes the proof.

\section{Proof of Lemma \ref{lemma:exact_acc_noise_existence}}\label{proof:lemma:exact_acc_noise_existence}

Let $f_{Z,1}(z)$ be the probability density function of the adversarial noise magnitude satisfying the support condition defined in Lemma \ref{lemma:optimal_support}, i.e, $f_{Z,1}(z)=0$ for $z\notin [(\eta-1)\Delta, (\eta+1)\Delta]$. Let its acceptance probability be $\Pr(\mathcal{A}_\eta; f_{Z,1}) = \alpha_1$, where $\alpha_1 > \alpha$. The initial state of this distribution is visualized in Figure \ref{fig:noise_initial_acc}.

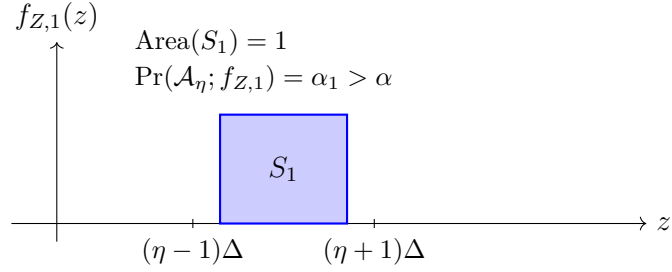
\begin{figure}[htbp]
    \centering
    \begin{tikzpicture}[scale=1.2]
        \draw[->] (-0.5,0) -- (6.5,0) node[right] {$z$};
        \draw[->] (0,-0.2) -- (0,2.0) node[above] {$f_{Z,1}(z)$};
        
        \draw (1.5,0.05) -- (1.5,-0.05) node[below, font=\small] {$(\eta-1)\Delta$};
        \draw (3.5,0.05) -- (3.5,-0.05) node[below, font=\small] {$(\eta+1)\Delta$};
        
        \fill[blue!20] (1.8,0) rectangle (3.2, 1.2);
        \draw[thick, blue] (1.8,0) rectangle (3.2, 1.2);
        \node at (2.5, 0.6) {$S_1$};

        \node[anchor=west, font=\small] at (0.75, 2) {$\text{Area}(S_1) =  1$};
        \node[anchor=west, font=\small, text width=6cm] at (0.75, 1.6) {$\Pr(\mathcal{A}_\eta; f_{Z,1}) = \alpha_1 > \alpha$};
    \end{tikzpicture}
    \caption{The initial noise distribution $f_{Z,1}(z)$ supported entirely within the acceptance region. In this example, $\alpha_1$ is higher than the target acceptance probability $\alpha$.}
    \label{fig:noise_initial_acc}
\end{figure}

We construct a new noise magnitude PDF, $f_{Z,2}(z)$ as follows
\begin{align}
    f_{Z,2}(z) = 
    \frac{\alpha}{\alpha_1} f_{Z,1}(z) + \left(1 - \frac{\alpha}{\alpha_1}\right) \delta(z - z_{\text{out}}),
\end{align}
where $z_{\text{out}} = (\eta+1)\Delta$. The adjustment process is shown in Figure \ref{fig:noise_adjusted_acc}, where the valid mass is reduced and the remainder is moved to a zero-acceptance point.

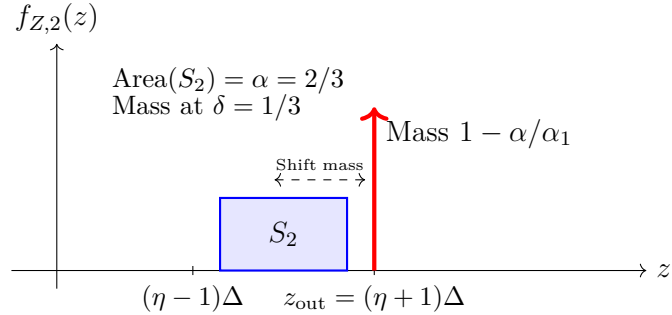
\begin{figure}[htbp]
    \centering
    \begin{tikzpicture}[scale=1.2]
        \draw[->] (-0.5,0) -- (6.5,0) node[right] {$z$};
        \draw[->] (0,-0.2) -- (0,2.5) node[above] {$f_{Z,2}(z)$};
        
        \draw (1.5,0.05) -- (1.5,-0.05) node[below, font=\small] {$(\eta-1)\Delta$};
        \draw (3.5,0.05) -- (3.5,-0.05) node[below, font=\small] {$z_{\text{out}} = (\eta+1)\Delta$};
        
        \fill[blue!10] (1.8,0) rectangle (3.2, 0.8);
        \draw[thick, blue] (1.8,0) rectangle (3.2, 0.8);
        \node at (2.5, 0.4) {$S_2$};

        \draw[ultra thick, ->, red] (3.5,0) -- (3.5,1.8) node[below right, black] {Mass $1 - \alpha/\alpha_1$};
        
        \node[anchor=west, font=\small] at (0.5, 2.1) {$\text{Area}(S_2) = \alpha = 2/3$};
        \node[anchor=west, font=\small] at (0.5, 1.8) {$\text{Mass at } \delta = 1/3$};

        \draw[<->, dashed] (2.4, 1.0) -- (3.4, 1.0) node[midway, above, font=\tiny] {Shift mass};
    \end{tikzpicture}
    \caption{The adjusted distribution $f_{Z,2}(z)$. The original mass is scaled down to exactly $\alpha$ (e.g., $2/3$). The discarded mass (e.g., $1/3$) is concentrated into a Dirac delta at the boundary $z_{\text{out}} = (\eta+1)\Delta$, ensuring the total integral is 1 while the acceptance probability is exactly $\alpha$.}
    \label{fig:noise_adjusted_acc}
\end{figure}

First, we calculate the probability of acceptance for the new distribution using \eqref{eq:prob_acceptance_lemma}. Noting that $\Phi_N(\cdot)$ is a continuous function and  $\Phi_N(z_{\text{out}}) = 0$ for $z_{\text{out}} = (\eta+1)\Delta$, we have
\begin{align}\label{eq:new_prob_acc_proof}
    \Pr(\mathcal{A}_\eta; f_{Z,2}) &= \int_{0}^{\infty} \Phi_N(z) f_{Z,2}(z) \, dz \nonumber \\
    &= \int_{0}^{(\eta+1)\Delta} \left( \frac{\alpha}{\alpha_1} f_{Z,1}(z) \right) \Phi_N(z)  \, dz + 
    \int_{(\eta+1)\Delta}^\infty\left(1 - \frac{\alpha}{\alpha_1} \right)\delta(z- z_{\text{out}}) \Phi_N(z)  \, dz\nonumber\\
    &=  \frac{\alpha}{\alpha_1} \underbrace{\int_{0}^{(\eta+1)\Delta} \Phi_N(z) f_{Z,1}(z) \, dz}_{\Pr(\mathcal{A}_\eta; f_{Z,1}) = \alpha_1} + 
     \left(1 - \frac{\alpha}{\alpha_1}\right)  \Phi_N(z_{\text{out}})\nonumber \\
    &= \frac{\alpha}{\alpha_1} (\alpha_1)+0 = \alpha.
\end{align}
Thus, the new noise magnitude PDF satisfies the constraint with equality.

Next, we evaluate the conditional MSE using \eqref{eq:mse_lemma_result}. We observe that the weighting function $\Psi_N(z)$, defined in \eqref{eq:Psi_piecewise}, is also continuous and zero for $z \ge (\eta+1)\Delta$. Therefore, $\Psi_N(z_{\text{out}}) = 0$. Thus, we have
\begin{align}\label{eq:Psi-2-in-1}
    \int_{0}^{\infty} \Psi_N(z) f_{Z,2}(z) \, dz &= \int_{0}^{(\eta+1)\Delta}  \left( \frac{\alpha}{\alpha_1} f_{Z,1}(z) \right) \Psi_N(z) \, dz + \left(1 - \frac{\alpha}{\alpha_1}\right) \Psi_N(z_{\text{out}})  \nonumber \\
    &= \frac{\alpha}{\alpha_1} \int_{0}^{(\eta+1)\Delta} \Psi_N(z) f_{Z,1}(z) \, dz.
\end{align}
Substituting this into the conditional MSE formula \eqref{eq:mse_lemma_result} yields
\begin{align}
    \mathbb{E}\left[ \|\mathbf{U} - \hat{\mathbf{U}}\|_2^2 \mid \mathcal{A}_\eta; f_{Z,2} \right] 
    &= \frac{1}{4 \Pr(\mathcal{A}_\eta; f_{Z,2})} \int_{0}^{\infty} \Psi_N(z) f_{Z,2}(z) \, dz \nonumber \\
    &\overset{(a)}{=} \frac{1}{4 \alpha} \left( \frac{\alpha}{\alpha_1} \int_{0}^{(\eta+1)\Delta} \Psi_N(z) f_{Z,1}(z) \, dz \right) \nonumber \\
    &= \frac{1}{4 \alpha_1} \int_{0}^{\infty} \Psi_N(z) f_{Z,1}(z) \, dz \nonumber \\
    &=\frac{1}{4 \Pr(\mathcal{A}_\eta; f_{Z,1})} \int_{0}^{\infty} \Psi_N(z) f_{Z,1}(z) \, dz \nonumber \\
    &= \mathbb{E}\left[ \|\mathbf{U} - \hat{\mathbf{U}}\|_2^2 \mid \mathcal{A}_\eta; f_{Z,1} \right],
\end{align}
where (a) follows from~\eqref{eq:new_prob_acc_proof} and~\eqref{eq:Psi-2-in-1}. 
This completes the proof of the lemma.

\section{Evaluation of the Game of coding for the Two-Dimensional Case (\texorpdfstring{$N=2$}{N=2})}\label{sec:case_N2}

In this section, we evaluate the general results established in Theorems~\ref{theorem: equivalence_two_problem} and \ref{theorem:Main_Minimax_Result} for the specific case of $N=2$. Note that based on \eqref{def:Gamma}, we have $\Gamma(2) = 1$, and thus based on \eqref{eq:Ball_Volume}, the volume (area) of a two-dimensional ball (disk) is given by   
\begin{align} \label{eq:V2_calc}
    V_2(r) = \frac{\pi^{2/2} r^2}{\Gamma(1 + 2/2)} = \frac{\pi r^2}{\Gamma(2)} = \pi r^2.
\end{align}

Next, we evaluate the  function $K_N(r, c)$. Recalling the definition in \eqref{eq:K_cap_def}, we have
\begin{align}
    K_N(r, c) = \frac{\pi^{(N-1)/2} r^N}{\Gamma(\frac{N+1}{2})} \int_{c/r}^{1} (1 - t^2)^{\frac{N-1}{2}} \, dt.
\end{align}
Substituting $N=2$, the coefficient depends on $\Gamma(3/2)$. Using the definition of the Gamma function in \eqref{def:Gamma}, we have $\Gamma(1.5) = \frac{\sqrt{\pi}}{2}$. Consequently, the coefficient simplifies to
\begin{align}
    \frac{\pi^{(2-1)/2} r^2}{\Gamma(\frac{2+1}{2})} = \frac{\sqrt{\pi} r^2}{\sqrt{\pi}/2} = 2r^2.
\end{align}
Thus, the kernel expression becomes
\begin{align} \label{eq:K2_integral_setup}
    K_2(r, c) = 2r^2 \int_{c/r}^{1} \sqrt{1 - t^2} \, dt.
\end{align}
We compute the integral in \eqref{eq:K2_integral_setup} using the standard substitution $t = \sin \theta$, which yields 
\begin{align}
    \int \sqrt{1-t^2} \, dt = \frac{1}{2} \left( t\sqrt{1-t^2} + \arcsin(t) \right).
\end{align}
Evaluating this from $c/r$ to $1$, and using the identity $\frac{\pi}{2} - \arcsin(x) = \arccos(x)$, we obtain
\begin{align} \label{eq:K2_result}
    K_2(r, c) &= 2r^2 \left[ \frac{\pi}{4} - \frac{1}{2} \left( \frac{c}{r}\sqrt{1-\frac{c^2}{r^2}} + \arcsin\left(\frac{c}{r}\right) \right) \right] \nonumber \\
    &= r^2 \arccos\left(\frac{c}{r}\right) - c\sqrt{r^2 - c^2}.
\end{align}

With the explicit form of $K_2(r, c)$ established in \eqref{eq:K2_result}, we proceed to evaluate $\Phi_2(z)$. First, recalling the definition of the intersection volume in \eqref{eq:V_lens_def}, for $N=2$ we have
\begin{align}
    V_{\text{lens}}(\Delta, \eta\Delta, z) = K_2(\Delta, u_c(z)) + K_2(\eta\Delta, z - u_c(z)),
\end{align}
where $u_c(z)$ is defined in \eqref{eq:uc_def} as
\begin{align}
    u_c(z) =\frac{z^2 + \Delta^2(1 - \eta^2)}{2z}.
\end{align}
Next, substituting this into the definition of $\Phi_N(z)$ in \eqref{eq:Phi_n_def}, and using the volume $V_2(\Delta) = \pi \Delta^2$ derived in \eqref{eq:V2_calc}, we obtain
\begin{align} \label{eq:Phi2_result}
    \Phi_2(z) = \frac{1}{\pi \Delta^2} &\Bigg[ \left( \Delta^2 \arccos\left(\frac{u_c(z)}{\Delta}\right) - u_c(z)\sqrt{\Delta^2 - u_c(z)^2} \right) \nonumber \\
    &\quad + \left( \eta^2\Delta^2 \arccos\left(\frac{z - u_c(z)}{\eta\Delta}\right) - (z - u_c(z))\sqrt{\eta^2\Delta^2 - (z - u_c(z))^2} \right) \Bigg].
\end{align}
Next, we determine the auxiliary functions $Q_2(r, d)$ and $J_2(r, d)$. Specifically, using $V_1(\rho) = 2\rho$ in~\eqref{eq:Q_n_def} with $N=2$ yields
\begin{align} \label{eq:Q2_calc}
    Q_2(r, d) &= \frac{r^2 - d^2}{3} \cdot 2\sqrt{r^2 - d^2} = \frac{2}{3} \left( r^2 - d^2 \right)^{3/2}.
\end{align}
For $J_2(r, d)$, substituting $N=2$ into \eqref{eq:J_n_def} gives
\begin{align} \label{eq:J2_calc}
    J_2(r, d) &= \frac{1}{2} r^2 K_2(r, d) + \frac{1}{2} d Q_2(r, d).
\end{align}
We now calculate the terms required for $\Psi_2(z)$, defined in \eqref{eq:Psi_n_def}. Note that, based on \eqref{def:v1} and \eqref{def:v2}, we have
$V_1 = K_2(\Delta, u_c(z))$ and $V_2 = K_2(\eta\Delta, z - u_c(z))$. We define $T_1$ as the first term in the numerator of \eqref{eq:Psi_n_def}
\begin{align} \label{eq:Psi2_term1}
    T_1 &\triangleq J_2(\Delta, u_c(z)) + z^2 V_1 \nonumber \\
    &\overset{(a)}{=} \left( \frac{1}{2} \Delta^2 V_1 + \frac{1}{2} u_c(z) Q_2(\Delta, u_c(z)) \right) + z^2 V_1 \nonumber \\
    &= \left( \frac{\Delta^2}{2} + z^2 \right) V_1 + \frac{u_c(z)}{2} Q_2(\Delta, u_c(z)),
\end{align}
where (a) follows from substituting \eqref{eq:J2_calc} and using the definition of $V_1$. Similarly, we define $T_2$ as the second term in the numerator of \eqref{eq:Psi_n_def}
\begin{align} \label{eq:Psi2_term2_calc}
    T_2 &\triangleq J_2(\eta\Delta, z - u_c(z)) + 4z^2 V_2 - 2z Q_2(\eta\Delta, z - u_c(z)) \nonumber \\
    &\overset{(b)}{=} \left[ \frac{1}{2} \eta^2 \Delta^2 V_2 + \frac{1}{2} (z - u_c(z)) Q_2(\eta\Delta, z - u_c(z)) \right] \nonumber \\
    &\quad + 4z^2 V_2 - 2z Q_2(\eta\Delta, z - u_c(z)) \nonumber \\
    &= \left( \frac{\eta^2 \Delta^2}{2} + 4z^2 \right) V_2 - \frac{3z + u_c(z)}{2} Q_2(\eta\Delta, z - u_c(z)),
\end{align}
where (b) follows from substituting \eqref{eq:J2_calc} and using the definition of $V_2$. Finally, combining $T_1$ and $T_2$ and dividing by $V_2(\Delta) = \pi \Delta^2$ results in the complete expression
\begin{align} \label{eq:Psi2_final}
    \Psi_2(z) &= \frac{1}{\pi \Delta^2} \Bigg[ \left( \frac{\Delta^2}{2} + z^2 \right) K_2(\Delta, u_c(z)) + \left( \frac{\eta^2 \Delta^2}{2} + 4z^2 \right) K_2(\eta\Delta, z - u_c(z)) \nonumber \\
    &\quad + \frac{u_c(z)}{2} Q_2(\Delta, u_c(z)) - \frac{3z + u_c(z)}{2} Q_2(\eta\Delta, z - u_c(z)) \Bigg],
\end{align}
where $Q_2$ is given by \eqref{eq:Q2_calc} and $K_2$ by \eqref{eq:K2_result}.

With the explicit expressions for $\Phi_2(z)$ and $\Psi_2(z)$ established in \eqref{eq:Phi2_result} and \eqref{eq:Psi2_final}, we have fully characterized all the underlying functions required by Theorem \ref{theorem:Main_Minimax_Result}. While the presence of transcendental terms in $\Phi_2(z)$ prevents an analytical derivation of the inverse function, the function $c_{\eta}(\alpha)$ defined in \eqref{eq:c_eta_result} can be  evaluated by applying the numerical procedure described in Remark \ref{rem:Numerical_Calculation} to these specific 2D case.

\end{document}